\documentclass[11pt,reqno]{amsart}
\usepackage[T1]{fontenc}
\usepackage{lmodern}
\usepackage[expansion=false]{microtype}
\usepackage[a4paper,margin=30mm,headheight=14pt]{geometry}
\usepackage{amsmath,amssymb,mathtools}
\providecommand{\coloneq}{\coloneqq}
\usepackage{mathrsfs,aliascnt}
\usepackage{graphicx}
\usepackage{tikz}
\usepackage{etoolbox}
\makeatletter
\patchcmd{\@maketitle}{\topskip42\p@}{\topskip18\p@}{}{}
\patchcmd{\@maketitle}{\dimen@34\p@}{\dimen@14\p@}{}{}
\patchcmd{\@setauthors}{\@topsep30\p@}{\@topsep20\p@}{}{}
\makeatother

\usepackage[hidelinks,unicode]{hyperref}
\usepackage[nameinlink,noabbrev]{cleveref}
\numberwithin{equation}{section}
\newtheorem{theorem}{Theorem}[section]
\newaliascnt{proposition}{theorem}
\newtheorem{proposition}[proposition]{Proposition}
\aliascntresetthe{proposition}
\crefname{proposition}{Proposition}{Propositions}
\Crefname{proposition}{Proposition}{Propositions}
\newaliascnt{lemma}{theorem}
\newtheorem{lemma}[lemma]{Lemma}
\aliascntresetthe{lemma}
\crefname{lemma}{Lemma}{Lemmas}
\Crefname{lemma}{Lemma}{Lemmas}
\newaliascnt{corollary}{theorem}
\newtheorem{corollary}[corollary]{Corollary}
\aliascntresetthe{corollary}
\crefname{corollary}{Corollary}{Corollaries}
\Crefname{corollary}{Corollary}{Corollaries}
\theoremstyle{definition}
\newaliascnt{definition}{theorem}
\newtheorem{definition}[definition]{Definition}
\aliascntresetthe{definition}
\crefname{definition}{Definition}{Definitions}
\Crefname{definition}{Definition}{Definitions}
\crefname{theorem}{Theorem}{Theorems}
\Crefname{theorem}{Theorem}{Theorems}
\crefname{section}{Section}{Sections}
\Crefname{section}{Section}{Sections}
\crefname{figure}{Figure}{Figures}
\Crefname{figure}{Figure}{Figures}
\crefname{equation}{Equation}{Equations}
\Crefname{equation}{Equation}{Equations}
\newcommand{\R}{\mathbb R}
\newcommand{\C}{\mathbb C}
\newcommand{\N}{\mathbb N}
\newcommand{\E}{\mathbb E}
\newcommand{\PP}{\mathbb P}
\newcommand{\G}{\mathcal G}
\newcommand{\D}{\mathbb D}
\newcommand{\one}{\mathbf 1}
\newcommand{\ac}{\mathrm{ac}}
\newcommand{\TV}{\mathrm{TV}}
\DeclareMathOperator{\supp}{supp}
\DeclareMathOperator{\Tr}{Tr}
\DeclareMathOperator{\diag}{diag}
\DeclareMathOperator{\lcap}{cap}
\DeclareMathOperator{\ImPart}{Im}
\DeclareMathOperator{\RePart}{Re}
\newcommand{\ip}[2]{\langle #1,#2\rangle}
\newcommand*{\dd}{\mathop{}\!\mathrm{d}}
\newcommand{\norm}[1]{\lVert #1\rVert}

\newcommand{\eref}[1]{\eqref{#1}}
\newcommand{\arxiv}[1]{\href{https://arxiv.org/abs/#1}{arXiv:#1}}
\newcommand{\doi}[1]{\href{https://doi.org/#1}{doi:\nolinkurl{#1}}}
\allowdisplaybreaks[2]
\makeatletter
\renewcommand{\paragraph}{\@startsection{paragraph}{4}%
  {\z@}{.35\linespacing\@plus.15\linespacing}{-\fontdimen2\font}%
  {\normalfont\bfseries}}
\makeatother

\newcommand{\sing}{\mathrm{s}}
\DeclareMathOperator{\Var}{Var}
\DeclareMathOperator{\adj}{adj}
\newcommand{\Ereg}{\mathcal E_{\mathrm{reg}}}

\newcommand{\Z}{\mathbb Z}
\newcommand{\eps}{\varepsilon}
\hypersetup{pdftitle={Spectral delocalization on the hyperbolic square lattice and its Euclidean products},pdfauthor={Simon Becker and Izak Oltman},pdfsubject={Spectral delocalization for the Anderson model with bounded potentials on the hyperbolic square lattice and its Euclidean products},pdfkeywords={Anderson model, bounded disorder, hyperbolic square lattice, absolutely continuous spectrum, truncated Cauchy distribution, stability under perturbations}}
\title[Disorder on the hyperbolic square lattice I: AC spectrum]{Disorder on the hyperbolic square lattice I: Anderson delocalization and absolutely continuous spectrum}
\author{Simon Becker}
\address{Department of Decision Sciences, Bocconi University, Via Roentgen 1, 20136 Milan, Italy}
\email{simon.becker@unibocconi.it}
\author{Izak Oltman}
\address{Department of Mathematics, Northwestern University, Evanston, Illinois 60208, USA}
\email{ioltman@northwestern.edu}
\begin{document}
\begin{abstract}
We study the Anderson model $A+\lambda\diag(\omega_x)$, where $A$ is the
unit-weight adjacency operator, $\lambda>0$, and the $\omega_x$ are bounded independent
and identically distributed random variables. The graph is
$\mathbb Z^d\square\G_{4,5}$, $d\ge0$: the Cartesian product of the Euclidean
lattice with the regular hyperbolic square lattice having five squares at
each vertex. We prove the existence of absolutely continuous spectrum and
the absence of singular spectrum on a deterministic set of positive
measure at weak disorder, and for $0<\lambda\le2$ when the single-site
density is bounded below on $[-1,1]$. Truncated Cauchy distributions
and controlled perturbations give purely absolutely continuous spectrum on intervals.
\end{abstract}
\maketitle

\section{Introduction}\label{sec:intro}

How does the geometry of a graph affect Anderson localization?
On $\Z$, localization occurs at every disorder strength, while on $\mathbb Z^2$ such a result is still open (but conjecturally believed to be true \cite{AWBook}). 
On the other hand, regular trees of degree at least three support absolutely continuous
spectrum at weak disorder \cite{Klein}.
Hyperbolic lattices combine the exponential volume growth of
trees with the polygonal loops of Euclidean lattices.
In this paper, we study the regular hyperbolic square lattice $\G_{4,5}$:
five squares meet at each vertex, compared with four on
$\mathbb Z^2=\G_{4,4}$.\footnote{More generally, $\G_{p,q}$ denotes the graph of a regular
tessellation by $p$-gons with $q$ faces meeting at each vertex.
The tessellation is hyperbolic when $(p-2)(q-2)>4$
\cite{DG,BridsonHaefliger}.
For square faces, this means $q>4$, so $\G_{4,5}$ is the
first hyperbolic member of the regular square family.}
This change permits weak-disorder delocalization while
preserving the square faces.
\Cref{fig:anderson_novelty} summarizes this comparison, and
\Cref{fig:square-lattice2} illustrates $\G_{4,5}$.

\begin{figure}[!htbp]
 \setlength{\abovecaptionskip}{6pt}
 \setlength{\belowcaptionskip}{6pt}
 \centering
 \includegraphics[width=0.6\linewidth]{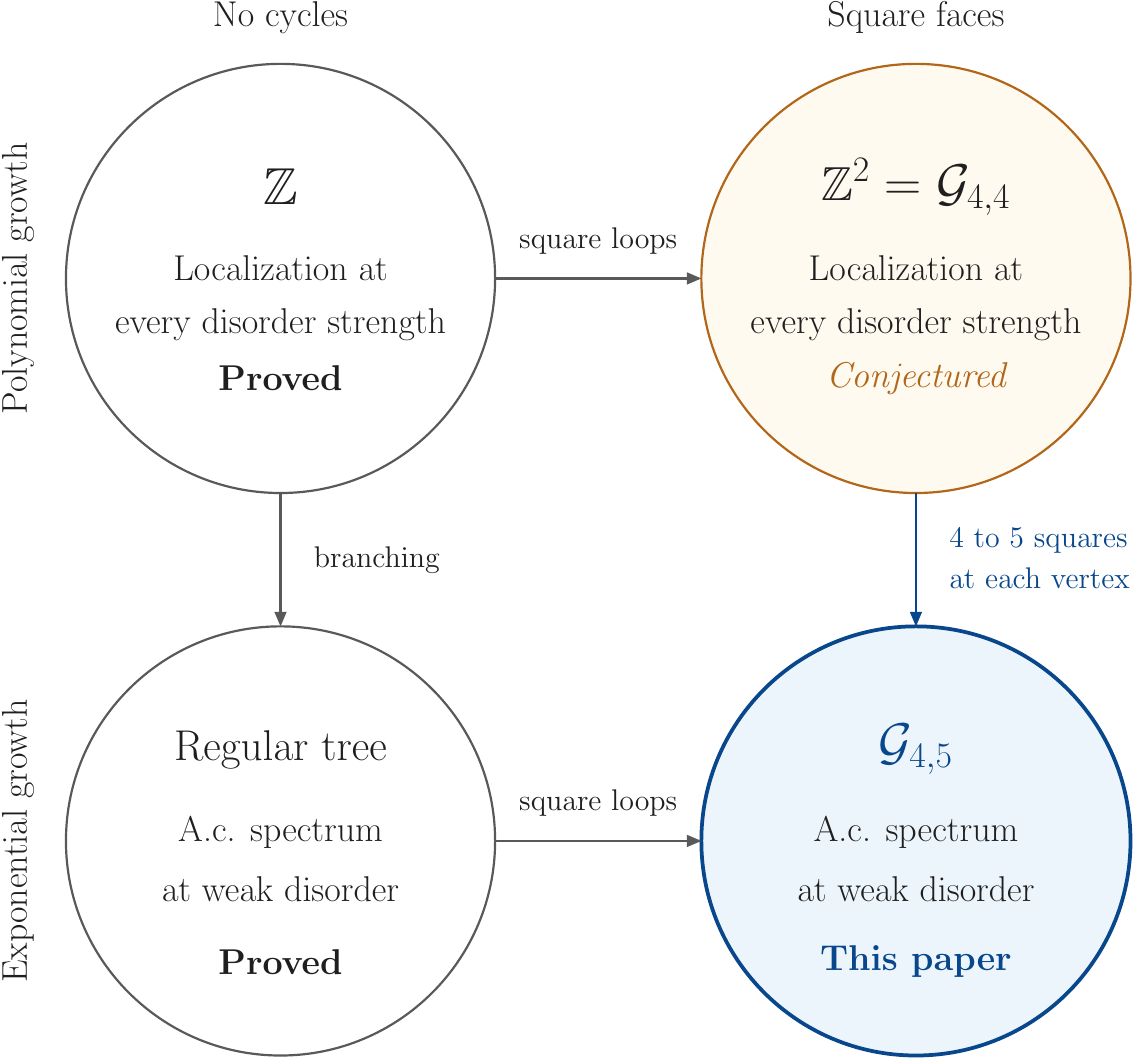}
 \caption{Spectral behavior for Anderson models \cite{AWBook,Klein} in relation to \Cref{thm:square-bounded}.}
 \label{fig:anderson_novelty}
\end{figure}

\subsection{Main result}

For a finite integer $d\ge0$, let
\[
 \mathcal X_d=\mathbb Z^d\mathbin{\square}\G_{4,5}.
\]
Here $\square$ denotes the Cartesian product of graphs:
two vertices are adjacent when one coordinate is unchanged
and the other moves along an edge.
We regard $\mathbb Z^0$ as a single vertex, so
$\mathcal X_0=\G_{4,5}$.

Let $X$ be the vertex set of $\mathcal X_d$.
Its adjacency operator $A_d$ on $\ell^2(X)$ is defined by
\[
 (A_d\psi)(x)\coloneq\sum_{y\sim x}\psi(y),
\]
where $y\sim x$ means that $x$ and $y$ are connected by an edge.
We consider the Anderson operator
\[
 (H_{\lambda,\omega}\psi)(x)
 =(A_d\psi)(x)+\lambda\omega_x\psi(x),
 \qquad \lambda>0,
\]
where the random variables $\omega_x$ are independent and
identically distributed with probability density $\rho$.
We assume that $\rho$ is supported in $[-1,1]$, is bounded
above, and is bounded below by a positive constant on
some interval $[-b,b]$ with $0<b\le1$.
The parameter $\lambda$ measures the disorder strength.

For each vertex $x\in X$, let $\delta_x\in\ell^2(X)$ be the
unit vector supported at $x$ and define its spectral measure by
\[
 \mu_x^\omega(J)
 =\langle\delta_x,\mathbf 1_J(H_{\lambda,\omega})\delta_x\rangle,
 \qquad J\subset\mathbb R\ \text{Borel},
\]
where $\mathbf 1_J$ is the indicator function of $J$ and
$\mathbf 1_J(H_{\lambda,\omega})$ is the corresponding spectral
projection, defined by the spectral theorem.
We write
$\mu_x^\omega=(\mu_x^\omega)^{\mathrm{ac}}
+(\mu_x^\omega)^{\mathrm{sing}}$
for its decomposition into absolutely continuous and singular
parts with respect to Lebesgue measure.
Our main result gives a deterministic energy set of positive
measure on which, almost surely, every local spectral measure
has an almost everywhere positive density and no singular part.

\begin{theorem}\label{thm:square-bounded}
Under the preceding assumptions on $\rho$, suppose that
\[
 0<\lambda\le
 \begin{cases}
  10^{-4},&d=0,\\
  1/5,&d\ge1.
 \end{cases}
\]
There is a deterministic Borel set
$D_{d,\lambda}\subset\sigma(A_d)$ of positive Lebesgue measure
such that, almost surely, simultaneously for every $x\in X$,
\[
 \frac{d(\mu_x^\omega)^{\mathrm{ac}}}{dE}(E)>0
 \quad\text{for almost every }E\in D_{d,\lambda},
 \qquad
 (\mu_x^\omega)^{\mathrm{sing}}(D_{d,\lambda})=0.
\]
In particular, $H_{\lambda,\omega}$ almost surely has
a nonzero absolutely continuous spectral subspace.
\end{theorem}

The theorem applies, for example, to the uniform distribution
on $[-1,1]$.
For each fixed density $\rho$, dimension $d$, and disorder
strength $\lambda$, the energy set is independent of the
realization of the potential.
The theorem does not assert that this set contains an interval.

We additionally obtain a larger disorder range when the
single-site density is bounded below throughout $[-1,1]$.
For suitable truncated Cauchy distributions, we strengthen
the positive-measure conclusion to purely absolutely continuous
spectrum on whole intervals, with positive local spectral
densities almost everywhere.
These interval conclusions persist under the perturbations
of the distribution specified below.
Precise statements appear in
\Cref{thm:full-support-bounded,thm:bounded-cauchy,thm:stability_under_change_density}.

\begin{figure}[tbp]
 \centering
 \includegraphics[width=0.42\textwidth]{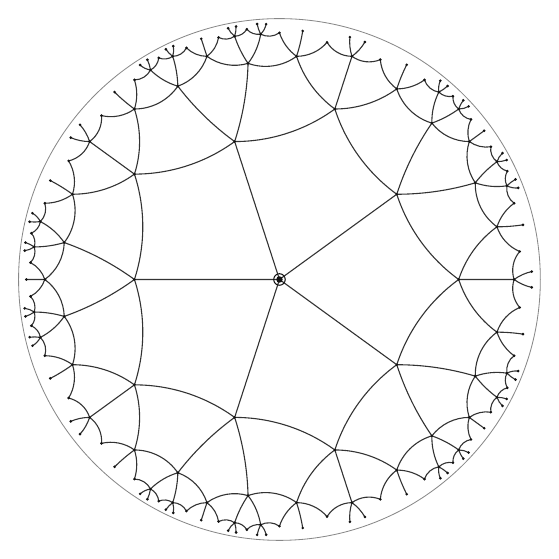}
 \caption{A finite portion of the graph $\G_{4,5}$. Each vertex is the corner of exactly five squares.}
 \label{fig:square-lattice2}
\end{figure}

\subsection{Connection to the literature}\label{sec:literature}

The closest rigorous results concern weak disorder on trees.
Klein proved purely absolutely continuous spectrum on energy
intervals for the Anderson model on the Bethe lattice
\cite{Klein}.
Aizenman, Sims, and Warzel established stability of the
absolutely continuous spectrum under weak disorder, including
certain correlations between the site potentials
\cite{ASWstability}.
Further results treat trees of finite cone type
\cite{KellerLenzWarzel} and products of trees with finite
graphs \cite{KleinSadel,SadelTreeStrip}.
Froese, Hasler, and Spitzer developed methods based on transfer
matrices and hyperbolic geometry for related graphs, including
graphs obtained by adding edges to trees, under additional
structural assumptions on the potentials
\cite{FroeseHaslerSpitzer}.

Hyperbolic lattices combine features of Euclidean lattices
and regular trees: they contain polygonal loops and have
exponential volume growth.
This combination changes the balance between interference
and the number of sites available to a propagating wave.
Numerical studies have found evidence of Anderson transitions
at nonzero disorder strengths on several regular hyperbolic
lattices \cite{ChenMaciejkoBoettcher,LiPengWangHu}.
These transitions separate extended states at weak disorder
from localized states at strong disorder.
The computations also identify mobility edges, the energies
separating localized and extended states
\cite{LiPengWangHu}.
Our results establish delocalization in the spectral sense,
through positive absolutely continuous densities and absence
of singular spectrum on the stated energy sets.
They do not establish dynamical transport.

Our argument builds on the resonance mechanism developed
by Aizenman and Warzel \cite{AWtree,AW}.
The idea is that the number of distant sites can compensate
for the small probability that any particular site resonates
with a fixed starting site.
On a tree, deleting a vertex separates the forward branches,
giving exact recursive identities for the resolvent.
On $\G_{4,5}$, square loops connect these branches, so the
same recursive argument does not apply.

\subsection{Further results}\label{sec:further-results}

We obtain a larger disorder range when the density is
bounded below throughout $[-1,1]$.
We also obtain interval results for truncated Cauchy
distributions and stability under changes of density.
Define
\[
 r_\square\coloneq\|A_0\|.
\]
By \Cref{bare:lem:radius}, $4\le r_\square\le\sqrt{22}$.
The Cartesian product structure gives $\|A_d\|=r_\square+2d$.
For a self-adjoint operator $T$, let $P_{\mathrm{sing}}(T)$
be the projection onto its singular spectral subspace and
let $\sigma_{\mathrm{ac}}(T)$ denote its absolutely continuous
spectrum.

\paragraph{A larger disorder range.}
\Cref{thm:square-bounded} only requires the density $\rho$ to be
bounded below near zero.
If this lower bound holds throughout $[-1,1]$, we obtain
the following stronger conclusion.

\begin{theorem}
\label{thm:full-support-bounded}
Let $\rho$ be a probability density supported in $[-1,1]$
such that, for some constants $0<\rho_-\le\rho_+<\infty$,
\[
 \rho_-\le\rho(u)\le\rho_+
 \quad\text{for almost every }u\in[-1,1].
\]
For every finite integer $d\ge0$ and
\begin{equation}\label{eq:full-support-window}
 0<\lambda<4+2\sqrt2-r_\square,
\end{equation}
there is a deterministic Borel set
\[
 D_{d,\lambda}\subset
 \Sigma_{d,\lambda}\coloneq\sigma(A_d)+[-\lambda,\lambda]
\]
of positive Lebesgue measure such that, almost surely,
simultaneously for every $x\in X$,
\begin{equation*}
 \begin{gathered}
 \frac{d(\mu_x^\omega)^{\mathrm{ac}}}{dE}(E)>0
 \quad\text{for almost every }E\in D_{d,\lambda},\\
 \mathbf 1_{D_{d,\lambda}}(H_{\lambda,\omega})
 P_{\mathrm{sing}}(H_{\lambda,\omega})=0.
 \end{gathered}
\end{equation*}
\end{theorem}

Since $r_\square\le\sqrt{22}$, the permitted disorder range
includes $0<\lambda\le2$ in every finite dimension.

\paragraph{Intervals for truncated Cauchy distributions.}
For the next result, we specify the distribution of the site
potentials directly.
For $0<\eps\le1$, let
\[
 c_\eps(v)\coloneq\frac{\eps}{\pi(v^2+\eps^2)}
\]
be the Cauchy density of scale $\eps$.
Conditioning this distribution on $[-M,M]$, where $M\ge\eps$,
gives the probability density
\begin{equation*}
 \varrho_{\eps,M}(v)
 \coloneq
 \frac{\eps\,\mathbf 1_{[-M,M]}(v)}
 {2\arctan(M/\eps)(v^2+\eps^2)}.
\end{equation*}
Let the variables $V_x$ be independent with common density
$\varrho_{\eps,M}$ and define
\begin{equation}\label{bc:eq:model}
 (H^{(d)}_{\eps,M}\psi)(x)
 \coloneq(A_d\psi)(x)+V_x\psi(x).
\end{equation}
Thus $|V_x|\le M$ at every vertex.
We denote the local spectral measures of this operator
by $\mu_{x,\eps,M}$.

\begin{theorem}
\label{thm:bounded-cauchy}
Fix a finite integer $d\ge0$ and an interval $J=[a,b]$ with
\begin{equation*}
 r_\square+2d<a<b<5+2d.
\end{equation*}
There exists $\eps_0=\eps_0(d,J)\in(0,1]$ such that,
for each fixed $0<\eps\le\eps_0$, there is a finite
$M_0=M_0(\eps,d,J)\ge\eps$ with the following property.
For every fixed finite $M\ge M_0$, almost surely,
\begin{equation}\label{bc:eq:conclusion}
 \begin{gathered}
 J\cup(-J)\subset\sigma_{\mathrm{ac}}(H^{(d)}_{\eps,M}),\\
 \mathbf 1_{J\cup(-J)}(H^{(d)}_{\eps,M})
 P_{\mathrm{sing}}(H^{(d)}_{\eps,M})=0.
 \end{gathered}
\end{equation}
Moreover, almost surely, simultaneously for every $x\in X$,
\begin{equation}\label{bc:eq:local-density}
 \frac{d\mu_{x,\eps,M}^{\mathrm{ac}}}{dE}(E)>0
 \quad\text{for almost every }E\in J\cup(-J).
\end{equation}
\end{theorem}

The bound $r_\square\le\sqrt{22}$ gives the explicit choices
\begin{equation*}
 J_d\coloneq 2d+[4.75,4.9]
 \subset(2d+\sqrt{22},2d+5).
\end{equation*}
In particular, suitable bounded independent potentials on
$\G_{4,5}$ give purely absolutely continuous spectrum on
$[4.75,4.9]\cup[-4.9,-4.75]$.
These intervals lie outside the spectrum of the free operator.

\paragraph{Stability under changes of distribution.}
The interval conclusion persists when the Cauchy density is
replaced by a sufficiently close compactly supported density
satisfying the bounds below.

For probability densities $f,g$, define the total variation distance
\[
 d_{\mathrm{TV}}(f,g)
 \coloneq\frac12\int_{\mathbb R}|f(t)-g(t)|\,\dd t.
\]
For independent random variables $V_x$ with density $f$, we define the Anderson operator $H_f$ by
\[
 (H_f \psi)(x)=(A_d\psi)(x)+V_x\psi(x)
\]
and denote its local spectral measures by $\mu_{x,f}$.

We impose three conditions on the perturbed density $f$.
It must be bounded above by twice the Cauchy density
$c_\eps$ and sufficiently close to $c_\eps$ in total variation.
It must also be bounded below by a positive constant on
an interval large enough to ensure positivity of the averaged
local spectral density throughout the target energy intervals.
The theorem makes this last requirement precise.

\begin{theorem}\label{thm:stability_under_change_density}
Fix a finite integer $d\ge0$ and an interval $J=[a,b]$ with
\[
 \|A_d\|<a<b<5+2d.
\]
For every sufficiently small $\eps>0$, there exists $t_*>0$
such that the following holds.

Let $f$ be a compactly supported probability density satisfying
\[
 f\le2c_\eps\quad\text{almost everywhere},
 \qquad
 d_{\mathrm{TV}}(f,c_\eps)<t_*.
\]
Assume also that $f$ is bounded below by a positive constant
on some interval $[-\ell,\ell]$ with $\ell>b-\|A_d\|$.

Then $H_f$ almost surely has purely absolutely continuous
spectrum on $J\cup(-J)$, with every local spectral density
positive almost everywhere on these intervals, simultaneously
at every vertex.
\end{theorem}

A more general version of \Cref{thm:stability_under_change_density}, allowing bounded independent noise
to be added to the Cauchy variable, is stated and proved in
\Cref{thm:law-stability} in \Cref{sec:law-stability}.

\subsection{Proof outline}\label{sec:proof-outline}

The argument compares the number of distant resonant sites with the decay of
the coupling to each site. Exponential volume growth provides many candidates;
the main issue is to control interference between paths and correlations
between resonance events in the presence of loops.

\paragraph{Tunneling.}
Eliminating all vertices except two gives an effective two-site operator. Its
off-diagonal entry, denoted by $\tau_{xy}$ below, measures the coupling between
the two sites. Although contributions from different paths may cancel at a
fixed energy, the leading term of $\tau_{xy}(z)$ as $|z|\to\infty$ counts
shortest paths and has a fixed sign. Averaging $\log|\tau_{xy}|$ over energy
and applying Jensen's formula after a conformal change of variable gives the
required lower bound; see \Cref{sec:resolvent,sec:log}.

\paragraph{Resonant sites.}
We select distant vertices with comparable numbers of shortest paths. The
selected sets grow exponentially, while the path multiplicities improve the
tunneling bound. At energies where the averaged local spectral density is
positive, one-site resampling gives a uniform lower bound on the resonance
probability. Correlations between different sites are controlled by a
rapid-decay estimate for translations of the random environment, with only a
polynomial loss in the distance. A second-moment argument then yields a
positive probability of finding a resonant site; see
\Cref{sec:averaging,sec:geometry}.

\paragraph{Absolutely continuous density.}
For $H=H_{\lambda,\omega}$, write
\[
 G_{xx}(z)=\langle\delta_x,(H-z)^{-1}\delta_x\rangle.
\]
At almost every energy, the local absolutely continuous density is
$\pi^{-1}\operatorname{Im}G_{xx}(E+i0)$. If this imaginary part vanished
almost surely, a distant resonance would force a large fluctuation of the
response at the root. The Poincar\'e inequality rules out such fluctuations
when the tunneling decay rate is smaller than the growth rate of the selected
sites; this is \Cref{prop:selected-ac-point}.

\paragraph{Singular spectrum.}
We condition on all potentials except one and apply spectral averaging in the
remaining site variable. Since the exceptional energy set is then fixed, the
boundary-value identities imply that the local spectral measure has no
singular mass on the deterministic energy set under consideration. Applying
this at every vertex gives \Cref{lem:no-singular}.

\paragraph{Bounded disorder.}
Shortest-path counts, together with estimates on the free spectrum, verify the
resonance criterion on a deterministic set of positive Lebesgue measure. For
$\G_{4,5}$ this uses the gaps and logarithmic capacity of the free spectrum;
for $\mathbb Z^d\square\G_{4,5}$ the Euclidean directions fill the free gaps
and give a larger weak-disorder range. The details are in
\Cref{sec:bounded-applications}.

\paragraph{Interval results.}
For the interval theorem we use the Cauchy identity
\[
 \int_{\mathbb R}\frac{c_\eps(v)}{v-\zeta}\,\dd v
 =\frac{1}{-i\eps-\zeta},
 \qquad \operatorname{Im}\zeta>0.
\]
On finite balls, this replaces Cauchy averaging with a fixed imaginary potential.
Combined with random-walk estimates, it gives a strict finite-block inequality
comparing logarithmic propagation cost with endpoint entropy. The strictness
allows truncation of the Cauchy law and the perturbations in
\Cref{thm:law-stability}; see \Cref{sec:interval-stability}.

\textbf{Statement on the use of artificial intelligence.}
This project originated in an attempt to extend results for tree graphs to lattices which retain exponential volume growth while, unlike trees, containing nontrivial loops. Our initial objective was to use the Simon--Wolff criterion to exclude point spectrum for the Anderson model on $\mathcal G_{8,8}$, and this approach was successful. During the subsequent development of the project, ideas of the authors were combined with assistance from ChatGPT 5.6 Sol, leading to a substantial strengthening of the original result and, ultimately, to a proof of the main theorem. The resulting argument was subsequently extensively revised, reorganized, and independently verified by the authors, with further assistance from large language models during this process. In particular, LLM-assisted combinatorial analysis helped identify a route by which the argument could be sharpened from $\mathcal G_{8,8}$ to the geometrically more Euclidean lattice $\mathcal G_{4,5}$.

Further interactions with ChatGPT 6 suggested possible extensions of the method to non-amenable geometries even closer to Euclidean lattices, including certain non-amenable deformations of $\mathbb Z^d$. These directions are not pursued in the present work.

The authors have independently checked the entire final manuscript, including all mathematical arguments, computations, references, and conclusions, and take full responsibility for its contents.

\section{Analytic framework}\label{sec:analytic-framework}
\subsection{Model and notation}\label{sec:model}
We collect the notation used in the proofs.
For a graph $Y$, let $A_Y$ be its unit-weight adjacency operator
and write $A_{4,5}=A_{\G_{4,5}}=A_0$.
The operator on the Cartesian product is
\begin{equation}\label{eq:square-product-model}
 A_d=A_{\mathbb Z^d}\otimes I+I\otimes A_{4,5},
 \qquad \mathcal X_d=\mathbb Z^d\square\G_{4,5}.
\end{equation}
Here $I$ denotes the identity on the appropriate factor.
For real sets, $S+T=\{s+t:s\in S,\ t\in T\}$.

We fix $\mathcal X_d$, write $X=V(\mathcal X_d)$ and $Q=5+2d$, and choose a
root $o$. We let $d_X(x,y)$ be graph distance, abbreviated to $d(x,y)$ in local
arguments, and set
\[
 N(x)=\{y:y\sim x\},\qquad
 B_R(x)=\{y:d(x,y)\le R\},\qquad S_R(x)=\{y:d(x,y)=R\}.
\]
A shortest path has length $d(x,y)$; $N_{xy}$ denotes the number of such
paths. We write $|F|$ for the size of a finite vertex set and $|J|$ for the
Lebesgue measure of an energy set. The symbols $F\Subset X$ and
$I\Subset U\subset\mathbb R$ mean, respectively, that $F$ is finite and
that $\overline I$ is a compact subset of $U$.

We let $\delta_x$ be the standard basis of $\ell^2(X)$, $\one_J$ an indicator
or the corresponding spectral projection, and $P_C$ the coordinate
projection onto $\ell^2(C)$ for $C\subset X$. Complements of vertex sets are
taken in $X$; $\operatorname{dist}(t,S)=\inf_{s\in S}|t-s|$ for real sets.
Constants denoted $C$ may vary between estimates; subscripts record their
allowed dependence. Inner products are conjugate-linear in the first variable.
We use $\|\cdot\|$ for Hilbert-space or operator norms, $\Tr$ for the
finite-dimensional trace, and $\diag(v)$ for multiplication by the
coordinates of $v$. Logarithms are natural, with
$\log^\pm t=\max\{\pm\log t,0\}$ for $t>0$.

For the bounded-disorder model, we let $\rho$ be a probability density and
assume, for constants $0<\rho_-\le\rho_+<\infty$, that
\begin{equation}\label{eq:density-upper}
 \supp\rho\subset[-1,1],\qquad 0\le\rho\le\rho_+<\infty,
 \qquad \int_{\mathbb R}\rho(u)\,\dd u=1,
\end{equation}
and, for some $b\in(0,1]$,
\begin{equation}\label{eq:density-lower}
 \rho(u)\ge\rho_->0\quad\text{for almost every }u\in[-b,b].
\end{equation}
We use the probability space $\Omega=[-1,1]^X$ with product law
$\mathbb P=\bigotimes_{x\in X}\rho(u)\,\dd u$.
The variables $\omega_x$ are the coordinate functions, and
$\mathbb E$ denotes expectation. The variables are independent
and identically distributed (iid); a.s. means with probability one, and
a.e. refers to Lebesgue measure in a real variable. For $C\subset X$, we write
$\Omega_C=[-1,1]^C$, $\mathbb P_C$ for its product law, and $\mathbb E_C$
for integration in those coordinates with the others fixed. A subscript
$\ne x$ means $C=X\setminus\{x\}$. An event depending on finitely many
coordinates is called a \emph{cylinder event}. For a disorder amplitude
$\lambda>0$, put
\begin{equation}\label{eq:model}
 H_{\lambda,\omega}=A_d+V_\omega,
 \qquad (V_\omega\psi)(x)=v_x\psi(x),\qquad v_x=\lambda\omega_x.
\end{equation}
We call $A_d$ the free operator and also write $A=A_d$, $H=H_{\lambda,\omega}$, and
$H^{\square,d}_{\lambda,\omega}=H_{\lambda,\omega}$ when recording
its dependence on $d$. The density of $v_x$ is
\begin{equation}\label{eq:scaled-density}
 \rho_\lambda(v)=\lambda^{-1}\rho(v/\lambda),\qquad
 \|\rho_\lambda\|_\infty\le\rho_+/\lambda.
\end{equation}
For a Borel set $J\subset\mathbb R$, we define
\begin{equation*}
 \mu_x^\omega(J)=\langle\delta_x,\one_J(H)\delta_x\rangle.
\end{equation*}
We suppress $\omega$ in $\mu_x=\mu_x^\omega$ and write
$\mu_x=\mu_x^{\ac}+\mu_x^{\sing}$ for its Lebesgue decomposition.
We let $P_{\ac}(H)$ project onto the subspace of vectors with a.c. spectral
measures, and set $P_{\sing}(H)=I-P_{\ac}(H)$. The singular part includes
pure-point and singular-continuous measures. We define $\sigma_{\ac}(H)$
as the spectrum of $H$ restricted to $P_{\ac}(H)\ell^2(X)$
\cite{Teschl,CarmonaLacroix}. For other operators below, subscripts on
$\mu_x$ specify the operator and its parameters in the same definition.

For a compact real set $S$, we let $\mathcal P(S)$ be its probability
measures. We define logarithmic energy and capacity by
\begin{equation}\label{eq:equilibrium-definition}
 \mathcal I(\vartheta)=\iint\log\frac1{|E-F|}\,
       \dd\vartheta(E)\,\dd\vartheta(F),\qquad
 \lcap S=\exp\!\left(-\inf_{\vartheta\in\mathcal P(S)}\mathcal I(\vartheta)\right).
\end{equation}
When $\lcap S>0$, the unique minimizing probability measure is the
equilibrium measure $\omega_S$. For a finite union of nondegenerate
intervals, $\omega_S\ll\dd E$ \cite[Appendix~A]{SimonCapacity}.
For $B>2$, the arcsine probability measure is
\begin{equation}\label{eq:arcsine}
 \dd\nu_B(E)=\frac{\one_{(-B,B)}(E)}{\pi\sqrt{B^2-E^2}}\,\dd E,
 \qquad a_B=\log(B/2).
\end{equation}
It satisfies $\omega_{[-B,B]}=\nu_B$ and $\lcap[-B,B]=B/2$
\cite{SaffTotik,SaffSurvey}.

\subsection{Resolvents and tunneling}\label{sec:resolvent}

A graph automorphism is an adjacency-preserving permutation of $X$.
For a group $\Gamma$ of such permutations and $g\in\Gamma$, we define
\begin{equation}\label{eq:environment-action}
 (g\cdot\omega)_x=\omega_{g^{-1}x},\qquad
 U_gf(\omega)=f(g^{-1}\cdot\omega).
\end{equation}
The spatial action $\mathcal U_g\delta_x=\delta_{gx}$ is unitary on
$\ell^2(X)$. Covariance means
$H_{g\cdot\omega}=\mathcal U_gH_\omega\mathcal U_g^{-1}$.
The operators $U_g$ are unitary on $L^2(\Omega)$. We write
$L^2_0(\Omega)=\{f:\mathbb Ef=0\}$ and
$\Var(f)=\mathbb E|f-\mathbb Ef|^2$. We require a finite symmetric
set $T=T^{-1}$ of automorphisms, each moving the root by one edge, and a
constant $C_{\rm P}<\infty$ such that
\begin{equation}\label{eq:Poincare}
 \Var(f)\le C_{\rm P}\sum_{t\in T}\|f-U_tf\|_2^2.
\end{equation}
For $\mathcal X_d$ the five translations within its hyperbolic factor
suffice; \Cref{lem:product-geometry} proves this estimate and the
selected-set bound below for both single-site laws.

Throughout \Cref{sec:analytic-framework} the graph $\mathcal X_d$ is fixed.
The balls and spheres are defined in \Cref{sec:model}; $\nu_B$ and $a_B$
are defined in \eqref{eq:arcsine}.

For $C\subseteq X$, we define the Dirichlet restriction
$H_C=P_CHP_C|_{\ell^2(C)}$; write $I_C$ for the identity on $\ell^2(C)$. For $x\notin C$, we put
\begin{equation}\label{eq:dirichlet-data}
 (H_C\psi)(u)=\sum_{\substack{w\in C\\w\sim u}}\psi(w)+v_u\psi(u),
 \qquad
 b_x^C=\sum_{\substack{u\sim x\\u\in C}}\delta_u.
\end{equation}
For $\ImPart z>0$, we set $G_{uv}(z)=\langle\delta_u,(H-z)^{-1}\delta_v\rangle.$
We define the one-site self-energy by
\begin{equation}\label{eq:one-site-self-energy}
 \Sigma_x(z)=z+
 \left\langle b_x^{X\setminus\{x\}},
 (H_{X\setminus\{x\}}-z)^{-1}b_x^{X\setminus\{x\}}\right\rangle.
\end{equation}
For distinct $x,y$, we write $C_{xy}=X\setminus\{x,y\}$ and
$A_{xy}=\langle\delta_x,A\delta_y\rangle$. We define
\begin{equation}\label{eq:two-site-self-energy}
 \sigma_u^{(xy)}(z)=z+\left\langle b_u^{C_{xy}},
 (H_{C_{xy}}-z)^{-1}b_u^{C_{xy}}\right\rangle,
 \qquad u\in\{x,y\},
\end{equation}
and
\begin{equation}\label{eq:tau}
 \tau_{xy}(z)=-A_{xy}
 +\left\langle b_x^{C_{xy}},(H_{C_{xy}}-z)^{-1}b_y^{C_{xy}}\right\rangle.
\end{equation}
The two-site formula \eqref{eq:two-site} expresses the resolvent in terms
of these coefficients.

For matrices and operators, we put $\ImPart M=(M-M^*)/(2i)$. We let
$R(z)=(H-z)^{-1}$ for $z\in\C\setminus\R$, and set
\begin{equation*}
 G_{CD}(z)=\one_C R(z)\one_D:\ell^2(D)\to\ell^2(C).
\end{equation*}
Here $C,D\subseteq X$. Unless stated otherwise,
$z=E+i\eta$ with $E\in\R$ and $\eta>0$. For the Dirichlet restriction
$H_C$ and neighbor vector $b_x^C$ defined in \eref{eq:dirichlet-data}, we
abbreviate $b_x^C$ to $b_x$ when the set $C$ is clear and write
$G^C_{uv}(z)=\langle\delta_u,(H_C-z)^{-1}\delta_v\rangle$
for $u,v\in C$. Thus, $\|H_C\|\le\|H\|$.

For a scalar resolvent quantity $F$, we write
$F(E+i0):=\lim_{\eta\downarrow0}F(E+i\eta)$ whenever this limit exists.
The Poisson kernel is $P_a(t):=a/[\pi(t^2+a^2)]$ for $a>0$.
The spectral theorem \cite{Teschl,CarmonaLacroix} gives
\begin{equation}\label{eq:resolvent-spectral}
 G_{xx}(z)=\int_\R\frac{d\mu_x(s)}{s-z},\qquad
 \ImPart G_{xx}(E+i\eta)=\int_\R\frac{\eta}{(s-E)^2+\eta^2}\,d\mu_x(s).
\end{equation}
Dividing the imaginary-part formula in \eref{eq:resolvent-spectral}
by $\pi$ and letting $\eta\downarrow0$ recovers the
a.c.\ density; see \cite[Theorem~3.27]{Teschl}:
\begin{equation*}
 \frac{d\mu_x^{\ac}}{dE}(E)
 =
 \frac1\pi\ImPart G_{xx}(E+i0)
\end{equation*}
for almost every $E\in\mathbb R$.
The same limit holds along every nontangential approach
$|\Re z-E|\le C\Im z$, with $C>0$ fixed, by the boundary-value theorem
\cite[Chapter~II]{Duren}. An analytic function with positive
imaginary part in the upper half-plane, such as $G_{xx}$, is called a
Herglotz function.
\subsubsection{Grushin problems and effective Hamiltonians}\label{sec:grushin-effective}

For a finite set $F\subset X$, we let $\iota_F:\ell^2(F)\to\ell^2(X)$
extend a vector by zero, and let $\pi_F=\iota_F^*$ restrict it to $F$.
Thus $\pi_F\iota_F=I_F$ and $\iota_F\pi_F=P_F$.
With $C=X\setminus F$, we introduce the Grushin operator
\begin{equation*}
 \mathcal P_F(z)=
 \begin{pmatrix}
   H-z & \iota_F\\
   \pi_F & 0
 \end{pmatrix}:
 \ell^2(X)\oplus\ell^2(F)\longrightarrow
 \ell^2(X)\oplus\ell^2(F).
\end{equation*}
For $\ImPart z>0$ this operator is invertible.  We write
\begin{equation*}
 \mathcal P_F(z)^{-1}=
 \begin{pmatrix}
  E_F(z)&E_{F,+}(z)\\
  E_{F,-}(z)&E_{F,-+}(z)
 \end{pmatrix}.
\end{equation*}
The finite-dimensional block $E_{F,-+}$ is the effective Hamiltonian.

For $C,D\subseteq X$, we write $A_{CD}=\one_C A\one_D:\ell^2(D)\to\ell^2(C),$ and $A_C=A_{CC}.$

\begin{lemma}\label{lem:effective-resolvents}
We let $F\subset X$ be finite and $C=X\setminus F$.  Then
\begin{equation}\label{eq:grushin-effective}
 E_{F,-+}(z)
 =zI_F-H_F+A_{FC}(H_C-z)^{-1}A_{CF},
 \qquad
 G_{FF}(z)=-E_{F,-+}(z)^{-1}.
\end{equation}
For one retained vertex $x$,
\begin{equation}\label{eq:one-site}
 E_{\{x\},-+}(z)=\Sigma_x(z)-v_x,
 \qquad
 G_{xx}(z)=\frac1{v_x-\Sigma_x(z)},
\end{equation}

For two retained vertices $x\ne y$,
\begin{equation}\label{eq:two-site}
 E_{\{x,y\},-+}(z)=
 \begin{pmatrix}
  \sigma_x^{(xy)}(z)-v_x & \tau_{xy}(z)\\
  \tau_{xy}(z) & \sigma_y^{(xy)}(z)-v_y
 \end{pmatrix},
 \qquad
 G_{\{x,y\},\{x,y\}}(z)=-E_{\{x,y\},-+}(z)^{-1}.
\end{equation}
Hence,
\begin{equation}\label{eq:pair-identities}
 \frac{G_{xy}}{G_{xx}}
 =\frac{\tau_{xy}}{v_y-\sigma_y^{(xy)}},
 \qquad
 \Sigma_y
 =\sigma_y^{(xy)}+
 \frac{\tau_{xy}^2}{v_x-\sigma_x^{(xy)}}.
\end{equation}
\end{lemma}

\begin{proof}
We use the standard Grushin inversion; see Sj\"ostrand--Zworski \cite[equation~(1.1)]{SjostrandZworski}.
For $\mathcal P_F(z)(u,u_-)=(f,f_+)$, the constraint fixes $u_F=f_+$, and
\[
 \begin{aligned}
 u_C&=(H_C-z)^{-1}(f_C-A_{CF}f_+),\\
 u_-&=f_F-A_{FC}(H_C-z)^{-1}f_C
      +\bigl[zI_F-H_F+A_{FC}(H_C-z)^{-1}A_{CF}\bigr]f_+.
 \end{aligned}
\]
The coefficient of $f_+$ is $E_{F,-+}$, and Grushin inversion gives
$G_{FF}=-E_{F,-+}^{-1}$, proving \eqref{eq:grushin-effective}.
Taking $F=\{x\}$ and $F=\{x,y\}$ gives \eqref{eq:one-site} and
\eqref{eq:two-site} from \eqref{eq:one-site-self-energy}--\eqref{eq:tau}.
Inverting the two-site matrix gives the first identity in
\eqref{eq:pair-identities}; eliminating its $x$-component gives
\[
 \Sigma_y-v_y=\sigma_y^{(xy)}-v_y
             -\frac{\tau_{xy}^2}{\sigma_x^{(xy)}-v_x},
\]
which is the second identity.
\end{proof}

The first identity in \eqref{eq:pair-identities} transmits a resonance at
$y$ to the root through $\tau_{xy}$; the second describes the change in the
one-site effective Hamiltonian when $x$ is restored.

We write
\begin{equation}\label{eq:responses}
 W_x(z):=\ImPart G_{xx}(z),\qquad Y_x(z):=\ImPart\Sigma_x(z).
\end{equation}
For existing boundary values we suppress $+i0$ and write
\[
 G_{xy}(E):=G_{xy}(E+i0),\quad
 \Sigma_x(E):=\Sigma_x(E+i0),\quad
 \tau_{xy}(E):=\tau_{xy}(E+i0),
\]
and similarly for $W_x(E)$ and $Y_x(E)$.

\begin{lemma}\label{lem:ward-identities}
We let $z=E+i\eta$ with $\eta>0$.  Then, for every vertex $x$,
\begin{equation}\label{eq:gamma-eta}
 \gamma_x(E,\eta):=\sum_y|G_{xy}(z)|^2
 =\frac{\ImPart G_{xx}(z)}{\eta},
\end{equation}
and
\begin{equation}\label{eq:kappa-definition}
 \kappa_x(E,\eta):=
 \frac{\gamma_x(E,\eta)}{|G_{xx}(z)|^2}
 =\sum_y\left|\frac{G_{xy}(z)}{G_{xx}(z)}\right|^2.
\end{equation}
Moreover,
\begin{equation}\label{eq:kappa-eta}
 \kappa_x(E,\eta)
 =\frac{\ImPart\Sigma_x(z)}{\eta}
 =1+\bigl\|(H_{X\setminus\{x\}}-z)^{-1}b_x\bigr\|^2,
\end{equation}
and hence, with \eref{eq:responses},
\begin{equation}\label{eq:ward}
 Y_x=\frac{W_x}{|G_{xx}|^2}=\eta\kappa_x(E,\eta).
\end{equation}
\end{lemma}

\begin{proof}
For $R_K(z)=(K-z)^{-1}$ with $K$ selfadjoint, the resolvent identity gives
\[
 R_K(z)-R_K(z)^*=(z-\bar z)R_K(z)^*R_K(z)=2i\eta R_K(z)^*R_K(z).
\]
Taking its matrix element at $f$ gives
\begin{equation}\label{eq:Ward}
 \ImPart\langle f,(K-z)^{-1}f\rangle
 =\eta\|(K-z)^{-1}f\|^2.
\end{equation}
We use the identity in this form; see, for example, \cite{Teschl}.  We apply it first with $K=H$ and
$f=\delta_x$ to obtain \eqref{eq:gamma-eta} and
\eqref{eq:kappa-definition}.  We next use \eqref{eq:one-site} and apply the
same identity with $K=H_{X\setminus\{x\}}$ and $f=b_x$ to get
\[
 \ImPart\Sigma_x(z)
 =\eta\Bigl(1+\|(H_{X\setminus\{x\}}-z)^{-1}b_x\|^2\Bigr).
\]
Together with $\ImPart G_{xx}=|G_{xx}|^2\ImPart\Sigma_x$, this gives
\eqref{eq:kappa-eta} and \eqref{eq:ward}.
\end{proof}

\subsubsection{The logarithmic tunneling estimate}\label{sec:log}

Paths on a graph with loops interfere. This distinguishes our setting from tree-like structures. We control the energy average of
$\log|\tau_{xy}|$ by mapping the spectral complement to a disk. The leading
Taylor coefficient counts shortest paths, all with the same sign.
Jensen's formula turns this coefficient into the logarithmic bound.
For the analytic background, see \cite{Duren,Garnett}; for the relation
between interval capacity, conformal maps, and equilibrium measure,
see \cite{SaffTotik,SaffSurvey}. We fix $B>2$ and let
$\D=\{w\in\C:|w|<1\}$. We use the map
\begin{equation}\label{eq:joukowski-map}
 z(w)=\frac B2\,(w+w^{-1}),\qquad 0<|w|<1,
\end{equation}
from the punctured disk onto $\C\setminus[-B,B]$. On the boundary,
$z(e^{i\theta})=B\cos\theta$, so normalized circle measure becomes
the equilibrium measure $\nu_B$ in \eref{eq:arcsine}. We first show
that the transformed tunneling amplitude has enough Hardy regularity
to pass Jensen's formula to the boundary. For a finite measure $\mu$,
write $\|\mu\|_{\TV}=|\mu|(\R)$.

We write $dm(\theta)=d\theta/(2\pi)$ for normalized circle measure. For $p>0$,
the Hardy class $H^p(\D)$ consists of analytic functions with
\begin{equation*}
 \|f\|_{H^p}^p:=\sup_{0<r<1}\int_0^{2\pi}|f(re^{i\theta})|^p\,dm(\theta)<\infty.
\end{equation*}
The boundary-value theorem \cite[Sections~2.1--2.3]{Duren} gives
\[
 f^*(e^{i\theta})=\lim_{r\uparrow1}f(re^{i\theta})\quad\text{a.e.},
 \qquad \|f\|_{H^p}^p=\int_0^{2\pi}|f^*(e^{i\theta})|^p\,dm(\theta).
\]
The same boundary limit holds along paths satisfying
$|w-e^{i\theta}|\le C(1-|w|)$.
We also write $f(e^{i\theta})=f^*(e^{i\theta})$. The following estimate
places the transformed Cauchy integrals in $H^{1/4}$.

\begin{lemma}\label{lem:hardy}
We let $B>0$, let $\mu$ be a finite signed or complex measure supported on
$[-B,B]$, let $c=B/2$, and let $z(w)=c(w+w^{-1})$ as in
\eref{eq:joukowski-map}. For fixed $d\in\C$, the analytic function
\[
 F(w)=d+\int\frac{d\mu(t)}{z(w)-t},
\]
extended at $w=0$, belongs to $H^{1/4}(\D)$. In detail,
\[
 \sup_{0<r<1}
 \int_0^{2\pi}|F(re^{i\theta})|^{1/4}\,\frac{d\theta}{2\pi}
 \le C\!\left(B,|d|,\|\mu\|_{\TV}\right).
\]
\end{lemma}
\begin{proof}
We set
\begin{equation}\label{eq:hardy-kernel}
 K_\mu(w):=\int\frac{1-w^2}{1-2(t/B)w+w^2}\,d\mu(t),
 \qquad F(w)-d=\frac{w}{c(1-w^2)}K_\mu(w).
\end{equation}
For $t/B=\cos\vartheta$ and $\zeta\in\D$,
\[
 \frac{1-w^2}{1-2(t/B)w+w^2}
 =\frac12\left(\frac{1+we^{i\vartheta}}{1-we^{i\vartheta}}
              +\frac{1+we^{-i\vartheta}}{1-we^{-i\vartheta}}\right),
 \qquad \RePart\frac{1+\zeta}{1-\zeta}
 =\frac{1-|\zeta|^2}{|1-\zeta|^2}>0.
\]
Thus $\RePart K_\mu>0$ when $\mu\ge0$ is nonzero. Its analytic principal
square root satisfies $|K_\mu|^{1/2}\le\sqrt2\RePart K_\mu^{1/2}$, so the
mean-value property for harmonic functions gives
\[
 \int_0^{2\pi}|K_\mu(re^{i\theta})|^{1/2}\,dm(\theta)
 \le\sqrt2\RePart K_\mu(0)^{1/2}=\sqrt{2\mu(\R)}.
\]
Decomposing a complex measure into four positive measures gives the same
bound up to an absolute constant with $\mu(\R)$ replaced by $\|\mu\|_{\TV}$.
Since $w/(1-w^2)\in H^{1/2}$, \eqref{eq:hardy-kernel} and
Cauchy--Schwarz yield
\[
 \begin{aligned}
 \int|F(re^{i\theta})-d|^{1/4}\,dm(\theta)
 &\le c^{-1/4}
 \left(\int\left|\frac{re^{i\theta}}{1-r^2e^{2i\theta}}\right|^{1/2}
 dm(\theta)\right)^{1/2}\\
 &\quad\times\left(\int|K_\mu(re^{i\theta})|^{1/2}\,dm(\theta)\right)^{1/2}
 \le C_B\|\mu\|_{\TV}^{1/4}.
 \end{aligned}
\]
Adding $|d|^{1/4}$ proves the claim.
\end{proof}

\begin{lemma}\label{lem:boundary-jensen}
We let $p>0$ and $f\in H^p(\D)$ with $f(0)\ne0$. Then
\[
 \log|f(e^{i\theta})|\in L^1(d\theta),\qquad
 \int_0^{2\pi}\log|f(e^{i\theta})|\,\frac{d\theta}{2\pi}
 \ge \log|f(0)|.
\]
In particular, every nonzero $H^p$ function has nonzero boundary values almost everywhere.
\end{lemma}

\begin{proof}
The canonical factorization \cite[Theorem~2.8]{Duren} gives $f=IO$ and
$\log|f^*|\in L^1(dm)$. Here $I$ is \emph{inner}, meaning
\[
 |I(w)|\le1\quad(w\in\D),\qquad |I^*(e^{i\theta})|=1\quad\text{a.e.},
\]
and $O$ is \emph{outer}: it is zero-free and its logarithm is recovered
from its boundary modulus by
\begin{equation}\label{eq:outer-poisson}
 \log|O(w)|=\int_0^{2\pi}\frac{1-|w|^2}{|e^{i\theta}-w|^2}
             \log|f^*(e^{i\theta})|\,dm(\theta).
\end{equation}
At $w=0$, \eqref{eq:outer-poisson} gives
\[
 \int\log|f^*|\,dm=\log|O(0)|
 \ge\log|I(0)O(0)|=\log|f(0)|.
\]
For a nonzero function with a zero of order $k$ at the origin,
$f(w)=w^kg(w)$, where $g\in H^p$ and $g(0)\ne0$. Since
$|f^*|=|g^*|$, the same logarithmic integrability proves nonvanishing a.e.
\end{proof}

\begin{proposition}\label{prop:log}
We let $H=A+V$ on a connected graph of degree at most $Q$, where $A$ has
unit edge weights, $V$ is real, and $\norm H\le B$ for some $B>2$.
If $r=d(x,y)$ and $N_{xy}$ is the number of shortest paths between
$x$ and $y$, then there is a constant $C_{B,Q}<\infty$ such that
\begin{align}
 \int\log|\tau_{xy}(E)|\,d\nu_B(E)
 &\ge\log N_{xy}-(r-1)\log(B/2),\label{eq:log-lower}\\
 \int\log^+|\tau_{xy}(E)|\,d\nu_B(E)
 &\le C_{B,Q}.\label{eq:log-positive}
\end{align}
Thus,
\begin{equation}\label{eq:log-cost}
 \int-\log\min\{|\tau_{xy}(E)|,1\}\,d\nu_B(E)
 \le(r-1)\log(B/2)-\log N_{xy}+C_{B,Q}.
\end{equation}
The constant is uniform in the graph, potential, and vertices.
\end{proposition}
\begin{proof}
We let $C=X\setminus\{x,y\}$ and
\[
 \mu_{xy}(J)=\langle b_x^C,\one_J(H_C)b_y^C\rangle.
\]
By Cauchy--Schwarz,
\[
 \|\mu_{xy}\|_{\TV}\le\|b_x^C\|\,\|b_y^C\|\le Q,
 \qquad \supp\mu_{xy}\subset[-B,B],
\]
and \eref{eq:tau} gives
\[
 \tau_{xy}(z)=-A_{xy}-\int_{[-B,B]}\frac{d\mu_{xy}(t)}{z-t}.
\]
With $c=B/2$ and
$\mathcal T(w)=\tau_{xy}(c(w+w^{-1}))$, \Cref{lem:hardy} gives a
uniform $H^{1/4}$ bound.  Hence Fatou's lemma and
$\log^+s\le4s^{1/4}$ yield
\[
 \int_0^{2\pi}\log^+|\mathcal T(e^{i\theta})|
 \,\frac{d\theta}{2\pi}\le C_{B,Q}.
\]

For $|z|>B$, since $\|H_C\|\le B$,
\[
 (H_C-z)^{-1}
 =
 -\frac1z\sum_{n=0}^\infty \left(\frac{H_C}{z}\right)^n .
\]
Hence, by \eqref{eq:tau},
\begin{equation}\label{eq:tau-neumann}
 \tau_{xy}(z)
 =
 -A_{xy}
 -
 \sum_{n=0}^\infty
 z^{-(n+1)}
 \left\langle b_x^C,H_C^n b_y^C\right\rangle .
\end{equation}
We let $r=d(x,y)$.  If $r\ge2$, then $A_{xy}=0$.  Since
$b_x^C$ and $b_y^C$ are supported on the neighbors of $x$ and $y$,
respectively,
\[
 \left\langle b_x^C,H_C^n b_y^C\right\rangle=0,
 \qquad 0\le n<r-2.
\]
Indeed, each off-diagonal factor in $H_C=A_C+V_C$ traverses one edge,
whereas a diagonal factor from $V_C$ stays at the same vertex.
For $n=r-2$, every contributing factor follows a shortest path.
Consequently the surviving terms are in one-to-one correspondence
with the shortest paths
\[
 x=u_0\sim u_1\sim\cdots\sim u_{r-1}\sim u_r=y.
\]
Since the edge weights are one,
\[
 \left\langle b_x^C,H_C^{r-2}b_y^C\right\rangle=N_{xy}.
\]
Substitution into \eqref{eq:tau-neumann} gives
\begin{equation*}
 \tau_{xy}(z)
 =
 -N_{xy}z^{-(r-1)}+O(z^{-r}).
\end{equation*}
For $r=1$, the same formula follows from
$A_{xy}=N_{xy}=1$.
Hence
\[
 \mathcal T(w)
 =-N_{xy}c^{-(r-1)}w^{r-1}+O(w^r).
\]
The function
\[
 f(w)=w^{-(r-1)}\mathcal T(w)
\]
belongs to $H^{1/4}$ and satisfies
$f(0)=-N_{xy}c^{-(r-1)}$.  By \Cref{lem:boundary-jensen},
\[
 \int_0^{2\pi}\log|\mathcal T(e^{i\theta})|
 \,\frac{d\theta}{2\pi}
 \ge \log N_{xy}-(r-1)\log c.
\]

For $w=e^{i\theta}$, $z(w)=B\cos\theta$. The two half-circles approach
opposite sides of the slit, and
$\tau_{xy}(E-i0)=\overline{\tau_{xy}(E+i0)}$, so their moduli agree.
Therefore the change of variables $E=B\cos\theta$,
\[
 \frac1{2\pi}\int_0^{2\pi}\Phi(B\cos\theta)\,d\theta
 =\int_{-B}^B\frac{\Phi(E)}{\pi\sqrt{B^2-E^2}}\,dE,
\]
gives \eref{eq:log-positive} and \eref{eq:log-lower}. Finally,
\[
 -\log\min\{|\tau_{xy}|,1\}
 =\log^+|\tau_{xy}|-\log|\tau_{xy}|,
\]
which proves \eref{eq:log-cost}.
\end{proof}

\subsubsection{Logarithmic averages}
\label{sec:finite-band-log}
For one interval, the map in \Cref{prop:log} turns the boundary average
into arcsine measure. For several intervals we use a covering map from the
disk and obtain equilibrium measure. We apply the estimate to the full
Green function, whose spectral measure lies in the chosen set, then use
\eqref{eq:pair-identities} to bound the effective hopping.

Two disk facts explain the change of variables. If $h:\D\to\D$ is
analytic with $h(0)=0$, then Hardy subordination \cite[Section~1.5]{Duren} means
\begin{equation}\label{eq:hardy-subordination}
 g\in H^p(\D)\quad\Longrightarrow\quad
 g\circ h\in H^p(\D),\qquad \|g\circ h\|_{H^p}\le\|g\|_{H^p}.
\end{equation}
Indeed, $|g|^p$ has a harmonic majorant $u\ge|g|^p$ with
$u(0)=\|g\|_{H^p}^p$ \cite[Section~2.6]{Duren}; since $u\circ h$ is harmonic,
\[
 \int|g(h(re^{i\theta}))|^p\,dm(\theta)
 \le\int u(h(re^{i\theta}))\,dm(\theta)
 =u(h(0))=\|g\|_{H^p}^p.
\]
For $\alpha\in\D$, we put
\[
 b_\alpha(w)=
 \begin{cases}
 w,&\alpha=0,\\[2pt]
 \displaystyle\frac{|\alpha|}{\alpha}
              \frac{\alpha-w}{1-\overline\alpha w},&\alpha\ne0.
 \end{cases}
 \qquad
 1-|b_\alpha(w)|^2
 =\frac{(1-|\alpha|^2)(1-|w|^2)}{|1-\overline\alpha w|^2}.
\]
For $\beta\in\R$, a \emph{Blaschke product} is a finite or infinite product
\begin{equation*}
 B(w)=e^{i\beta}\prod_j b_{\alpha_j}(w),\qquad
 \sum_j(1-|\alpha_j|)<\infty,\qquad |B^*|=1\quad\text{a.e.}
\end{equation*}
It converges locally uniformly and has precisely the prescribed zeros,
with multiplicity. If these zeros are also zeros of $g\in H^p$, division
removes them while preserving the boundary modulus
\cite[Sections~2.2--2.4]{Duren}:
\begin{equation}\label{eq:blaschke-division}
 g/B\text{ analytic in }\D
 \quad\Longrightarrow\quad
 g/B\in H^p(\D),\qquad \|g/B\|_{H^p}=\|g\|_{H^p}.
\end{equation}

\begin{lemma}\label{lem:finite-band-jensen}
We let $S$ be a finite union of nondegenerate compact real intervals,
and let $\mu$ be a finite real signed measure supported on $S$.
Suppose its Cauchy transform satisfies
\begin{equation}\label{eq:finite-band-transform}
 f(z)=\int_S\frac{d\mu(t)}{t-z}
      =a z^{-q}+O(z^{-q-1}),\qquad a\ne0,\quad q\in\N.
\end{equation}
Then $\log|f(E+i0)|\in L^1(\omega_S)$ and
\begin{equation*}
 \int_S\log|f(E+i0)|\,d\omega_S(E)
 \ge\log|a|-q\log\lcap S.
\end{equation*}
\end{lemma}
\begin{proof}
We write $S=\bigsqcup_{j=1}^k I_j$ and $\widehat\C=\C\cup\{\infty\}$.
We choose the universal covering map, locally one-to-one, normalized by
\begin{equation}\label{eq:finite-band-cover}
 \phi:\D\longrightarrow\widehat\C\setminus S,\qquad
 \phi(0)=\infty,\qquad \phi(w)=cw^{-1}+O(1),\quad c>0.
\end{equation}
Write $I_j=[m_j-b_j,m_j+b_j]$ and set
\[
 \begin{aligned}
 J_j(\zeta)&=m_j+\tfrac{b_j}{2}(\zeta+\zeta^{-1}),&
 \psi_j&=J_j^{-1},\quad\psi_j(\infty)=0,\\
 f_j(z)&=\int_{I_j}\frac{d\mu(t)}{t-z},&
 g_j&=f_j\circ J_j,\qquad h_j=\psi_j\circ\phi.
 \end{aligned}
\]
Lemma~\ref{lem:hardy}, after translating $I_j$, gives $g_j\in H^{1/4}$.
Since $h_j:\D\to\D$ and $h_j(0)=0$, \eqref{eq:hardy-subordination} yields
\begin{equation}\label{eq:finite-band-hardy}
 f\circ\phi=\sum_jg_j\circ h_j,\qquad
 \|f\circ\phi\|_{H^{1/4}}^{1/4}
 \le\sum_j\|g_j\|_{H^{1/4}}^{1/4}<\infty.
\end{equation}
The Green function with pole at infinity satisfies, as $|z|\to\infty$,
\[
 g_S(z)=\int_S\log|z-t|\,d\omega_S(t)-\log\lcap S,
 \qquad g_S(z)=\log|z|-\log\lcap S+o(1).
\]
It is harmonic off $S$ and $\infty$ and vanishes on $S$; here $\omega_S$
and $\lcap S$ are defined in \eqref{eq:equilibrium-definition}.
By \cite[Theorems~4.3--4.4]{ChristiansenSimonZinchenko}, the pole set
$\phi^{-1}(\infty)$ satisfies the Blaschke condition. Let $\mathcal B$
be the corresponding Blaschke product, with multiplicities, normalized by
$\mathcal B'(0)>0$. Then
\begin{equation}\label{eq:finite-band-blaschke}
 |\mathcal B(w)|=e^{-g_S(\phi(w))},\qquad
 \mathcal B(w)=\frac{\lcap S}{c}w+O(w^2),\qquad |\mathcal B^*|=1\quad\text{a.e.}
\end{equation}
By \eqref{eq:finite-band-transform} and \eqref{eq:finite-band-cover},
$f\circ\phi$ has a zero of order $q$ at every pole of $\phi$.
Equations \eqref{eq:blaschke-division}, \eqref{eq:finite-band-hardy} and
\eqref{eq:finite-band-blaschke} therefore give
\[
 F:=\frac{f\circ\phi}{\mathcal B^q}\in H^{1/4},\qquad
 F(0)=\frac{a c^{-q}}{(\lcap S/c)^q}=a(\lcap S)^{-q},\qquad
 |F^*|=|(f\circ\phi)^*|.
\]
Finally, the covering pushes circle measure to equilibrium measure
\cite[Corollary~4.6]{ChristiansenSimonZinchenko}: for nonnegative Borel $\Phi$,
\begin{equation}\label{eq:finite-band-pushforward}
 \int_0^{2\pi}\Phi(\phi(e^{i\theta}))\,dm(\theta)
 =\int_S\Phi(E)\,d\omega_S(E).
\end{equation}
By \cite[Section~2 and Theorem~3.1]{ChristiansenSimonZinchenko},
$\phi$ extends analytically across the circle outside a set of arc length
zero, and its derivative is nonzero away from the countably many preimages
of the interval endpoints. Radial paths therefore map to nontangential
approaches to the interval interiors. Since $\omega_S\ll dE$,
\eqref{eq:finite-band-pushforward} and the boundary-value theorem
\cite[Chapter~II]{Duren} identify $(f\circ\phi)^*$ with
$f(\phi(e^{i\theta})\pm i0)$ almost everywhere. As $\mu$ is real,
$f(E-i0)=\overline{f(E+i0)}$, so their moduli agree.
Lemma~\ref{lem:boundary-jensen} and \eqref{eq:finite-band-pushforward} give
\[
 \int_S\log|f(E+i0)|\,d\omega_S(E)
 =\int\log|F^*|\,dm
 \ge\log|F(0)|=\log|a|-q\log\lcap S.
\]
The same lemma gives $\log|F^*|\in L^1(dm)$, hence the required
$L^1(\omega_S)$ integrability.
\end{proof}

We fix a deterministic full-measure set of energies on which all boundary
identities used later are valid. The finite modifications are as follows.
\begin{definition}
Let $D,P\subset X$ be finite and disjoint.  Deleting the vertices in
$D$ and pinning the potential to zero on $P$ gives the operator
\begin{equation*}
 H^{D,P}
 =
 A_{X\setminus D}
 +
 \diag\!\bigl(v_x^{D,P}\bigr)_{x\in X\setminus D},
 \qquad
 v_x^{D,P}
 =
 \begin{cases}
  0, & x\in P,\\
  v_x, & x\in X\setminus(D\cup P).
 \end{cases}
\end{equation*}
Here deletion means the Dirichlet restriction to $X\setminus D$, while
zero-pinning leaves the graph unchanged and changes only the indicated
site potentials.

A \emph{finite resampling} on a finite set $R\subset X$ replaces
$(v_x)_{x\in R}$ by independent variables with the same single-site
law, independent of the unchanged potentials.
\end{definition}
The superscript $D,P$ on a kernel or effective coefficient means that it
is computed from $H^{D,P}$. For $u,v\ne x$, we write $G^{(x)}_{uv}=G^{X\setminus\{x\}}_{uv}$ for a
one-vertex deletion.

\begin{lemma}\label{lem:regular-energies}
There is a deterministic Borel set $\Ereg\subset\R$ of full Lebesgue
measure such that, for every $E\in\Ereg$, almost surely the following
holds simultaneously for all finite disjoint $D,P\subset X$ and all
retained vertices.

All boundary values at $E+i0$ of the scalar Green functions and
effective coefficients occurring in
\eqref{eq:one-site}--\eqref{eq:pair-identities}, for $H^{D,P}$ and the
finite-deletion operators entering their Grushin problems, exist and
are finite. Moreover,
\[
 G^{D,P}_{xx}(E)\ne0,
\]
and, for distinct retained vertices $x,y$,
\[
 \tau^{D,P}_{xy}(E)\ne0
 \quad\text{if $x,y$ lie in the same component of $X\setminus D$},
 \qquad
 \tau^{D,P}_{xy}\equiv0
 \quad\text{otherwise}.
\]
All of
\eqref{eq:one-site}, \eqref{eq:two-site}, and
\eqref{eq:pair-identities} remain valid at $z=E+i0$.

The same set $\Ereg$ may be chosen so that these conclusions hold
simultaneously for any prescribed countable family of finite
resamplings.  In addition, for each fixed $x\in X$, after replacing
$v_x$ by a real parameter $s$, they hold for every $E\in\Ereg$, for
almost every exterior configuration and Lebesgue-almost every
$s\in\R$.
\end{lemma}\begin{proof}
All operators obtained from the bounded model by finite deletion and
zero-pinning have spectrum in a fixed compact interval
$[-B_0,B_0]$.  Every scalar resolvent entry is the Cauchy transform of
a finite complex measure supported in this interval, and every scalar
entry of the effective Hamiltonians in
\eqref{eq:grushin-effective} is an affine function of $z$ plus such a transform.
Hence the boundary values exist and are finite for almost every energy;
see, for example, \cite[Section~3.3]{Duren}.

For a retained vertex,
\[
 G^{D,P}_{xx}(z)=-z^{-1}+O(z^{-2}),
\]
so this analytic function is not identically zero.  Lemmas
\ref{lem:hardy} and \ref{lem:boundary-jensen} therefore imply
\[
 G^{D,P}_{xx}(E+i0)\ne0
\]
for almost every $E\in[-B_0,B_0]$.  Outside this interval the function is
analytic in $E$, and its zeros are discrete.

If retained vertices $x,y$ lie in the same connected component of
$X\setminus D$, we put $r_D=d_{X\setminus D}(x,y)$ and let $N^D_{xy}$
count shortest paths in that graph. The expansion in the proof of
\Cref{prop:log} gives
\[
 \tau^{D,P}_{xy}(z)
 =
 -N^D_{xy}z^{-(r_D-1)}+O(z^{-r_D}),
 \qquad N^D_{xy}\ge1,
\]
so the same argument gives
$\tau^{D,P}_{xy}(E+i0)\ne0$ for almost every $E$.
For vertices in different components,
$\tau^{D,P}_{xy}\equiv0$.

There are only countably many finite sets $D,P$, vertices, pairs, and
scalar coefficients under consideration.  Fubini's theorem followed by
a countable intersection therefore produces a deterministic Borel set
$\Ereg$ of full measure on which all these conclusions hold almost
surely.  Since the relevant denominators are reciprocals of diagonal
Green functions for finite deletions, they are nonzero on $\Ereg$;
hence \eqref{eq:one-site}, \eqref{eq:two-site}, and
\eqref{eq:pair-identities} pass to the boundary.

Every finite resampling has the same product-law marginal distribution.
Thus the same deterministic set $\Ereg$ works for each member of any
prescribed countable family of resamplings, and a further countable
intersection gives simultaneous validity.

Finally, fix $x$ and regard $v_x=s$ as a deterministic parameter.
On each interval $|s|\le n$, the same argument applies with a compact
spectral enclosure depending on $n$. Fubini in $(E,s,\omega_{\ne x})$,
followed by a countable intersection over $n\in\N$ and $x\in X$, gives
the stated conclusion for almost every exterior configuration and
Lebesgue-almost every $s\in\R$.
\end{proof}

We use the full-measure set $\Ereg$ from \Cref{lem:regular-energies} and put
\begin{equation*}
 \ell_{xy}(E)=
 \begin{cases}
 -\log\min\{|\tau_{xy}(E+i0)|,1\},&E\in\Ereg,\\
 +\infty,&E\notin\Ereg,
 \end{cases}
\end{equation*}
We let $\mathscr R\subset\mathbb N$ be any deterministic unbounded set
of radii. For nonempty deterministic sets $F_R\subseteq S_R(o)$,
$R\in\mathscr R$, we put $M_R=|F_R|$ and define
\begin{equation}\label{eq:selected-cost}
 \mathcal L_{F,R}(E)=\frac1{RM_R}\sum_{x\in F_R}\E\ell_{ox}(E),
 \qquad
 \mathcal L_F(E)=\liminf_{\substack{R\to\infty\\R\in\mathscr R}}\mathcal L_{F,R}(E).
\end{equation}
Every limit over $R$ involving these sets is taken through
$\mathscr R$; the selected-endpoint arguments therefore apply along any
deterministic unbounded set of radii.

Henceforth $+i0$ is suppressed at energies in $\Ereg$. We may shrink
$\Ereg$ by further deterministic null sets when applying the
countably many differentiation statements below.

\subsection{Resonant environments and resampling}\label{sec:averaging}
Transitivity makes $\overline\mu_\lambda=\mathbb E\mu_x$ independent
of $x$. Fix the potentials on $X\setminus\{x\}$ and let $H_{x,v}$ denote the
operator obtained by setting the potential at $x$ equal to $v$.  Its
local spectral measure at $x$ is
\begin{equation*}
 \mu_{x,v}(J)
 :=
 \langle\delta_x,\one_J(H_{x,v})\delta_x\rangle,
 \qquad J\subset\R\ \text{Borel}.
\end{equation*}
Conditioning on the exterior of $x$, the one-site identity gives
\[
 G_{xx,v}(E+i\eta)
 =
 \frac{1}{v-\Sigma_x(E+i\eta)}.
\]
Writing
\[
 \Sigma_x(E+i\eta)=S_\eta(E)+iT_\eta(E),
 \qquad T_\eta(E)>0,
\]
we obtain
\begin{equation*}
 \int_{\mathbb R}\frac1\pi\ImPart G_{xx,v}(E+i\eta)\,dv
 =
 \int_{\mathbb R}
 \frac{T_\eta(E)}
 {\pi\bigl((v-S_\eta(E))^2+T_\eta(E)^2\bigr)}\,dv
 =1.
\end{equation*}
By the spectral theorem,
\[
 \frac1\pi\ImPart G_{xx,v}(E+i\eta)
 =
 \int_{\mathbb R}
 \frac{\eta}{\pi((E-t)^2+\eta^2)}\,d\mu_{x,v}(t).
\]
We define $\nu(J)=\int_\R\mu_{x,v}(J)\,dv$. Tonelli and
\eqref{eq:resolvent-spectral}, together with the one-site integral, give
\[
 (P_\eta*\nu)(E)=1,\qquad
 \int_\R\frac{d\nu(t)}{1+t^2}=\pi,
\]
where the second identity follows by taking $E=0$ and $\eta=1$.
In particular, $\nu$ is locally finite. For $\varphi\in C_c(\R)$,
\[
 (P_\eta*\varphi)(t)\longrightarrow\varphi(t),\qquad
 |(P_\eta*\varphi)(t)|\le\frac{C_\varphi}{1+t^2},
 \qquad 0<\eta\le1.
\]
The bound follows from the $L^\infty$ contraction on a fixed compact
set and the $t^{-2}$ decay of the kernel outside that set.
Dominated convergence and $P_\eta*\nu=1$ therefore give
\[
 \begin{aligned}
 \int_\R\varphi\,d\nu
 &=\lim_{\eta\downarrow0}\int_\R(P_\eta*\varphi)(t)\,d\nu(t)\\
 &=\lim_{\eta\downarrow0}\int_\R\varphi(E)(P_\eta*\nu)(E)\,dE
 =\int_\R\varphi(E)\,dE.
 \end{aligned}
\]
Thus $\nu=dE$, and for every Borel set $J$,
\begin{equation}\label{eq:rank-one-average}
 \int_{\R}\mu_{x,v}(J)\,dv=|J|,\qquad
 \overline\mu_\lambda(J)\le\|\rho_\lambda\|_\infty|J|.
\end{equation}
This shows that $\overline\mu_\lambda$ is absolutely continuous. We define the local density of states by
\begin{equation}\label{eq:averaged-density}
 \overline\mu_\lambda(dE)=n_\lambda(E)\,dE,
 \qquad 0\le n_\lambda(E)\le\rho_+/\lambda
 \quad\text{almost everywhere}.
\end{equation}

Lemma~\ref{lem:alignment} turns $n_\lambda(E)>0$ into positive one-site
mass near $\Sigma_x(E)$. Proposition~\ref{prop:positive-dos} proves the
needed density positivity. Lemma~\ref{lem:joint} then imposes this local
event near a distant vertex while preserving the exponential tunneling rate.

\subsubsection{Resonances}
A resonance at $x$ means that $v_x$ is close to $\Sigma_x(E)$.
Weighted rank-one averaging identifies the conditional density at this
level \cite{SimonRankOne,SimonAveraging,Marx,LiawTreil}.

The following weighted form of rank-one spectral averaging is a direct
consequence of the Poisson-transform identity for rank-one perturbations;
compare \cite[Proposition~3.1]{Marx}.

\begin{lemma}\label{lem:weighted-contact}
We fix the exterior of a vertex $x$ and vary only its diagonal value $v$.
We write $\Sigma(z)=\Sigma_x(z)$ and $G_v(z)=\frac1{v-\Sigma(z)}$
and we let $\mu_v$ be the spectral measure at $\delta_x$. For
$f\in L^1(\R)$, we set
\[
 \nu_f(J)=\int_\R f(v)\mu_v(J)\,dv.
\]
Then $\nu_f=q_f(E)\,dE$ with $\|q_f\|_{L^1}\le\|f\|_{L^1}$. We put
\[
 \mathcal R_\Sigma=
 \{E\in\R:\Sigma(E+i0)\text{ exists and belongs to }\R\}.
\]
On $\mathcal R_\Sigma$ one has
\begin{equation}\label{eq:contact-formula}
 q_f(E)=f(\Sigma(E+i0))
 \qquad\text{for almost every }E\in \mathcal R_\Sigma.
\end{equation}
Moreover,
\begin{equation}\label{eq:contact-pushforward}
 \bigl|\{E\in \mathcal R_\Sigma:\Sigma(E+i0)\in J\}\bigr|\le |J|
 \qquad\text{for every Borel }J\subset\R.
\end{equation}
Thus, whenever $\Sigma_x(E)$ is real almost surely, $n_\lambda(E)=\E\rho_\lambda(\Sigma_x(E))$ for almost every $E$.
\end{lemma}
\begin{proof}
For a Lebesgue-null set $N$, \eqref{eq:rank-one-average} implies
$\mu_v(N)=0$ for a.e. $v$. Since $\mu_v(\R)=1$,
\begin{equation}\label{eq:contact-L1-contraction}
 \begin{aligned}
 |\nu_f|(N)&\le\int|f(v)|\mu_v(N)\,dv=0,\\
 \|q_f-q_g\|_1&\le\int|f(v)-g(v)|\mu_v(\R)\,dv=\|f-g\|_1.
 \end{aligned}
\end{equation}
This defines $q_f\in L^1$ for every $f\in L^1$ and proves its norm bound.

We first take $f\in C_c(\R)$, the continuous compactly supported functions.
We put $S_\eta(E)=\RePart\Sigma(E+i\eta)$, $T_\eta(E)=\ImPart\Sigma(E+i\eta)>0$
and $S(E)=\Sigma(E+i0)$ on $\mathcal R_\Sigma$. By \eqref{eq:resolvent-spectral}, Fubini,
and $G_v=(v-\Sigma)^{-1}$,
\begin{equation}\label{eq:contact-poisson}
 \begin{aligned}
 (P_\eta*q_f)(E)
 &=\int_\R f(v)\,\frac1\pi\ImPart G_v(E+i\eta)\,dv\\
 &=\int_\R\frac{T_\eta(E)f(v)}{\pi((v-S_\eta(E))^2+T_\eta(E)^2)}\,dv
   =(P_{T_\eta(E)}*f)(S_\eta(E)).
 \end{aligned}
\end{equation}
On $\mathcal R_\Sigma$, $S_\eta(E)\to S(E)$ and $T_\eta(E)\to0$. Uniform convergence of
Poisson integrals of $f\in C_c$ gives
\[
 |(P_{T_\eta(E)}*f)(S_\eta(E))-f(S(E))|
 \le\|P_{T_\eta(E)}*f-f\|_\infty+|f(S_\eta(E))-f(S(E))|
 \longrightarrow0.
\]
The left side of \eqref{eq:contact-poisson} tends to $q_f(E)$ at its
Lebesgue points, proving \eqref{eq:contact-formula} for continuous $f$.
For $f\ge0$, integration and \eqref{eq:contact-L1-contraction} yield
\begin{equation}\label{eq:contact-test-domination}
 \int_{\mathcal R_\Sigma}f(S(E))\,dE=\int_{\mathcal R_\Sigma}q_f(E)\,dE
 \le\int_\R f(v)\,dv.
\end{equation}
Domination of measures in \eqref{eq:contact-test-domination} gives
\eqref{eq:contact-pushforward} for all Borel sets.

For general $f\in L^1$, we choose $f_j\in C_c(\R)$ with
$\|f_j-f\|_1\to0$. The pushforward bound makes $f\circ S$ independent of
the chosen representative, up to an energy-null set. Together with
\eqref{eq:contact-L1-contraction}, it gives
\[
 \begin{aligned}
 \|q_f-f\circ S\|_{L^1(\mathcal R_\Sigma)}
 &\le\|q_f-q_{f_j}\|_1+\|(f_j-f)\circ S\|_{L^1(\mathcal R_\Sigma)}\le2\|f_j-f\|_1\longrightarrow0.
 \end{aligned}
\]
This proves \eqref{eq:contact-formula} for $L^1$ weights. For each exterior
configuration write $q^{\omega_{\ne x}}_{\rho_\lambda}$ for the conditional
density. Conditioning in \eqref{eq:averaged-density} gives
\[
 n_\lambda(E)=\E_{\ne x}q^{\omega_{\ne x}}_{\rho_\lambda}(E)
 =\E_{\ne x}\rho_\lambda(\Sigma_x(E))
\]
for a.e. energy at which $\Sigma_x(E)$ is real almost surely, by Fubini and
\eqref{eq:contact-formula}.
\end{proof}

For the conversion of positive averaged density into a one-site resonance
probability, compare \cite[Lemma~3.4 and Appendix~A]{AW}.

\begin{lemma}\label{lem:alignment}
For almost every $E\in\Ereg$ satisfying
\[
 n_\lambda(E)>0,
 \qquad
 \PP\{W_o(E)=0\}=1,
\]
one has, for every $x\in X$,
\begin{equation}\label{eq:DOS-alignment}
 \Sigma_x(E)\in\R\quad\text{a.s.},
 \qquad
 n_\lambda(E)=\E\rho_\lambda(\Sigma_x(E)).
\end{equation}
Moreover, there exist $a,t_0,q>0$, depending on $E$ but not on $x$,
such that the exterior-measurable event
\begin{equation}\label{eq:B}
 B_x^0
 =
 \left\{
 \frac1{2t}
 \int_{\Sigma_x(E)-t}^{\Sigma_x(E)+t}
 \rho_\lambda(v)\,dv
 \ge a
 \quad\text{for every }0<t<t_0
 \right\}
\end{equation}
satisfies $\PP(B_x^0)=q>0.$
\end{lemma}

\begin{proof}
By transitivity,
\[
 \PP\{W_x(E)=0\}=1
\]
for every $x$.  Since $G_{xx}(E)$ is finite and nonzero by
\Cref{lem:regular-energies}, \eqref{eq:one-site} gives
\[
 \Sigma_x(E)
 =
 v_x-G_{xx}(E)^{-1}\in\R
 \qquad\text{a.s.}
\]
The contact formula \eqref{eq:contact-formula} and
\eqref{eq:averaged-density} therefore give
\[
 n_\lambda(E)
 =
 \E\rho_\lambda(\Sigma_x(E))
\]
for almost every such $E$.

Let $N_\rho$ be the null set of non-Lebesgue points of $\rho_\lambda$.
For each fixed exterior configuration, \eqref{eq:contact-pushforward}
with $\Sigma=\Sigma_x$ gives
\[
 \bigl|
 \{E\in\mathcal R_{\Sigma_x}:
   \Sigma_x(E)\in N_\rho\}
 \bigr|
 =0.
\]
Tonelli's theorem, followed by an intersection over the countable
vertex set, shows that for almost every $E$ under consideration, $\Sigma_x(E)\notin N_\rho$ a.s. for every $x.$
Set $X_x:=\rho_\lambda(\Sigma_x(E)).$
By \eqref{eq:DOS-alignment}, $\E X_x=n_\lambda(E)>0.$
Hence there is $a>0$ such that $\PP\{X_x>2a\}>0.$
For
\[
 A_t(s):=
 \frac1{2t}\int_{s-t}^{s+t}\rho_\lambda(v)\,dv,
\]
Lebesgue differentiation gives
\[
 A_t(\Sigma_x(E))
 \longrightarrow
 \rho_\lambda(\Sigma_x(E))=X_x
 \qquad (t\downarrow0)
\]
almost surely.  Consequently,
\[
 \{X_x>2a\}
 \subset
 \bigcup_{m\ge1}
 \left\{
 A_t(\Sigma_x(E))\ge a
 \text{ for every }0<t<1/m
 \right\}
\]
up to a null set.  Since the set on the left has positive probability,
there is some $m$ for which
\[
 q:=
 \PP\left\{
 A_t(\Sigma_x(E))\ge a
 \text{ for every }0<t<1/m
 \right\}>0.
\]
We set $t_0=1/m$.

Since $t\mapsto A_t(s)$ is continuous for $t>0$, it is enough in
\eqref{eq:B} to impose the inequalities for rational
$t\in(0,t_0)$.  Thus $B_x^0$ is measurable with respect to the
exterior of $x$, and hence independent of $v_x$.  Transitivity makes
$a,t_0,q$ independent of $x$.
\end{proof}

\subsubsection{Positivity criteria for averaged DOS}\label{sec:dos-positive}
To prove $n_\lambda>0$, shift all potentials in one finite block together.
If its effective matrix has nonzero imaginary part, the averaged imaginary
resolvent is positive. If it is Hermitian, a finitely supported approximate
eigenvector puts an eigenvalue in the shift interval. We compare
\cite{Wegner,HislopMuller}.

\begin{proposition}\label{prop:positive-dos}
Assume \eref{eq:density-upper}--\eref{eq:density-lower}.
Let $I\subset\R$ be bounded. Suppose there is a finite set
$\Lambda\subset X$, $|\Lambda|=m$, such that for every $E\in I$
there is a unit vector $\psi_E\in\ell^2(\Lambda)$, extended by zero to
$X$, satisfying
\begin{equation*}
 A_{\Lambda^c\Lambda}\psi_E=0,\qquad
 \|(A-E)\psi_E\|\le\frac{\lambda b}{16}.
\end{equation*}
Then $n_\lambda(E)>0$ for almost every $E\in I$.
\end{proposition}

\begin{proof}
By \eref{eq:averaged-density} and \eref{eq:resolvent-spectral},
\begin{equation}\label{eq:dos-boundary}
 n_\lambda(E)
 =
 \lim_{\eta\downarrow0}\frac1\pi\E W_x(E+i\eta)
 \qquad\text{for almost every }E.
\end{equation}
Fix such an $E\in I$ for which \Cref{lem:regular-energies} also applies
and put $C=\Lambda^c$.

We first use an elementary finite-dimensional observation.  If
$M_\eta\to M$, $\ImPart M_\eta<0$, and $\ImPart M\le0$, write the
eigenvalues of $M_\eta$ as
$\alpha_j(\eta)-i\beta_j(\eta)$, $\beta_j(\eta)>0$.  Schur
triangularization gives
\[
 \ImPart\Tr(M_\eta+sI)^{-1}
 =
 \sum_{j=1}^m
 \frac{\beta_j(\eta)}
 {(s+\alpha_j(\eta))^2+\beta_j(\eta)^2}\ge0.
\]
Since
\[
 \int_{-a}^a
 \frac{\beta\,ds}{(s+\alpha)^2+\beta^2}
 =
 \arctan\frac{a+\alpha}{\beta}
 -
 \arctan\frac{-a+\alpha}{\beta},
\]
one obtains
\begin{equation}\label{eq:matrix-shift-positive}
 \liminf_{\eta\downarrow0}
 \int_{-a}^a\ImPart\Tr(M_\eta+sI)^{-1}\,ds>0
\end{equation}
whenever either $\ImPart M\ne0$, or $M=M^*$ and
$\sigma(M)\cap(-a/2,a/2)\ne\varnothing$.  Indeed, in the first case
$\Tr(\ImPart M)<0$, so some limiting eigenvalue has negative imaginary
part; in the second case the contribution of an eigenvalue converging
to $(-a/2,a/2)$ tends to $\pi$.

Choose $x_0\in\Lambda$ and write
$v_{x_0}=s$ and $v_x=s+u_x$ for $x\ne x_0$, with $u_{x_0}=0$.
This change of variables has Jacobian one.  Restrict to
\begin{equation}\label{eq:rectangle}
 |s|\le\frac{\lambda b}{2},\qquad
 |u_x|\le\frac{\lambda b}{8}\quad(x\ne x_0).
\end{equation}
Then $|v_x|<\lambda b$ for every $x\in\Lambda$, and hence
\[
 \prod_{x\in\Lambda}\rho_\lambda(v_x)
 \ge\left(\frac{\rho_-}{\lambda}\right)^m
 \qquad\text{a.e.}
\]

For fixed $u$, set
\begin{equation}\label{eq:common-shift-matrix}
 M_{u,\eta}
 =
 A_\Lambda+\diag u-E-i\eta
 -
 A_{\Lambda C}(H_C-E-i\eta)^{-1}A_{C\Lambda}.
\end{equation}
By \eref{eq:grushin-effective},
$G_{\Lambda\Lambda}(E+i\eta)=(M_{u,\eta}+sI)^{-1}$.
The Ward identity gives $\ImPart M_{u,\eta}<0$, while its boundary value
$M_u$ satisfies $\ImPart M_u\le0$.

If $\ImPart M_u\ne0$, \eqref{eq:matrix-shift-positive} applies.
Suppose $M_u=M_u^*$.  Since
$A_{C\Lambda}\psi_E=0$, the exterior term in
\eqref{eq:common-shift-matrix} annihilates $\psi_E$, and
\[
 \|M_u\psi_E\|
 \le
 \|(A-E)\psi_E\|+\|\diag u\,\psi_E\|
 \le\frac{\lambda b}{16}+\frac{\lambda b}{8}
 =\frac{3\lambda b}{16}
 <\frac{\lambda b}{4}.
\]
Hence
\begin{equation*}
 \operatorname{dist}(0,\sigma(M_u))
 \le \|M_u\psi_E\|
 <\frac{\lambda b}{4}.
\end{equation*}
Taking
$a=\lambda b/2$ in \eqref{eq:matrix-shift-positive} therefore yields
\begin{equation}\label{eq:block-shift-positive}
 \liminf_{\eta\downarrow0}
 \int_{-\lambda b/2}^{\lambda b/2}
 \ImPart\Tr(M_{u,\eta}+sI)^{-1}\,ds>0
\end{equation}
for every $u$ in \eqref{eq:rectangle} and almost every exterior
configuration.

Let $\mathcal U$ denote the $u$-rectangle and put
\[
 J_\eta(u,\omega_C)
 =
 \int_{-\lambda b/2}^{\lambda b/2}
 \ImPart\Tr(M_{u,\eta}+sI)^{-1}\,ds\ge0.
\]
Disintegration, the preceding change of variables, and the lower bound
on the joint density give
\[
 \E\ImPart\Tr G_{\Lambda\Lambda}(E+i\eta)
 \ge
 \left(\frac{\rho_-}{\lambda}\right)^m
 \E_C\int_{\mathcal U}J_\eta(u,\omega_C)\,du.
\]
With $J_*=\liminf_{k\to\infty}J_{1/k}$,
\eqref{eq:block-shift-positive} gives $J_*>0$ almost everywhere.
Thus Fatou's lemma and transitivity, together with
\eqref{eq:dos-boundary}, imply
\[
 \begin{aligned}
 m\pi n_\lambda(E)
 &=
 \lim_{k\to\infty}
 \E\ImPart\Tr G_{\Lambda\Lambda}(E+i/k)\ge
 \left(\frac{\rho_-}{\lambda}\right)^m
 \E_C\int_{\mathcal U}J_*(u,\omega_C)\,du
 >0.
 \end{aligned}
\]
Hence $n_\lambda(E)>0$ for almost every $E\in I$.
\end{proof}
This gives the following two corollaries:
\begin{corollary}\label{cor:free-criterion}
Assume \eref{eq:density-lower} in addition to the model hypotheses.
Then
\[
 n_\lambda(E)>0
 \qquad\text{for almost every }E\in\sigma(A).
\]
\end{corollary}

\begin{proof}
We put $I=\sigma(A)$, which is compact, fix \(\lambda>0\), and set
$\varepsilon=\frac{\lambda b}{32}.$
For each \(E_0\in I\), Weyl's criterion and finite-support approximation
give a finitely supported unit vector \(\psi_{E_0}\) such that $\|(A-E_0)\psi_{E_0}\|<\varepsilon.$
Whenever \(|E-E_0|<\varepsilon\),
\begin{equation}\label{eq:free-criterion-local-quasimode}
 \|(A-E)\psi_{E_0}\|
 \le\|(A-E_0)\psi_{E_0}\|+|E-E_0|
 <2\varepsilon=\frac{\lambda b}{16}.
\end{equation}
We choose \(E_1,\dots,E_N\in I\) with
\[
 I\subset\bigcup_{j=1}^N(E_j-\varepsilon,E_j+\varepsilon),
 \qquad
 S=\bigcup_{j=1}^N\supp\psi_{E_j}.
\]
We let \(\Lambda\) contain \(S\) and all neighbors of \(S\).  Then for every
\(j\),
\begin{equation}\label{eq:free-criterion-common-block}
 \supp\psi_{E_j}\subset\Lambda,
 \qquad
 \supp(A\psi_{E_j})\subset\Lambda,
 \qquad
 A_{\Lambda^c\Lambda}\psi_{E_j}=0.
\end{equation}
For \(E\in I\), we choose \(j\) with \(|E-E_j|<\varepsilon\) and put
\(\psi_E=\psi_{E_j}\).  Equations
\eqref{eq:free-criterion-local-quasimode}--\eqref{eq:free-criterion-common-block}
are exactly the hypotheses of \Cref{prop:positive-dos}; hence
\(n_\lambda(E)>0\) for almost every \(E\in I\).
\end{proof}

\begin{corollary}\label{cor:shifted-dos}
Suppose the independent potentials have a bounded, compactly supported
common density $\varrho$, and
$\varrho(v)\ge c_0>0$ for almost every $v\in[v_0-a,v_0+a]$, where
$a>0$. We write $\mathbb E\mu_x(dE)=n(E)\,dE$. Then $n(E)>0$ almost
everywhere on $\sigma(A)+v_0+(-a,a)$.
\end{corollary}
\begin{proof}
We fix $t\in(v_0-a,v_0+a)$.  We choose $\zeta>0$ so small that
\begin{equation}\label{eq:shifted-dos-core}
 [t-\zeta,t+\zeta]\subset(v_0-a,v_0+a).
\end{equation}
For $E_0\in\sigma(A)$, Weyl's criterion gives a finitely supported unit
vector $\psi$ such that
\begin{equation*}
 \|(A-E_0)\psi\|<\frac{\zeta}{32}.
\end{equation*}
We let $S=\supp\psi$ and choose a finite $\Lambda$ containing $S$ and every
neighbor of $S$.  Every edge incident to $S$ then has both endpoints in
$\Lambda$, hence
\begin{equation}\label{eq:shifted-dos-decoupling}
 A_{\Lambda^c\Lambda}\psi=0.
\end{equation}
For every $E$ with $|E-E_0|<\zeta/32$, the same vector satisfies
\begin{align}
 \|(A-E)\psi\|
 &\le \|(A-E_0)\psi\|+|E-E_0|\|\psi\| <\frac{\zeta}{16}.\label{eq:shifted-dos-nearby}
\end{align}

We choose one vertex $x_0\in\Lambda$ and parametrize a rectangular set of
potentials by
\begin{equation}\label{eq:shifted-dos-rectangle}
 v_x=t+s+u_x,\qquad
 u_{x_0}=0,\quad |s|\le\frac\zeta2,
 \quad |u_x|\le\frac\zeta8\ (x\ne x_0).
\end{equation}
By \eqref{eq:shifted-dos-core}, every coordinate in this rectangle lies
in the region on which $\varrho\ge c_0$.  Thus the joint density of the
coordinates in $\Lambda$ is bounded below there by
\begin{equation}\label{eq:shifted-dos-density-floor}
 \prod_{x\in\Lambda}\varrho(v_x)\ge c_0^{|\Lambda|}>0.
\end{equation}
At spectral parameter $E+t+i\eta$, subtracting the common scalar $tI$
from the retained block gives
\[
 A_\Lambda-E-i\eta+sI+\diag(u),
\]
with the same exterior contribution in the Grushin effective Hamiltonian
as in the proof of \Cref{prop:positive-dos}, evaluated at $E+t+i\eta$.
The vector $\psi$ has zero coupling to the exterior by
\eqref{eq:shifted-dos-decoupling}.

If the boundary value of this finite effective Hamiltonian has nonzero
imaginary part, the shift integral from \Cref{prop:positive-dos} is
strictly positive.  In the
Hermitian case the diagonal perturbation costs at most $\zeta/8$, and
\eqref{eq:shifted-dos-nearby} gives
\begin{equation*}
 \|(A-E+\diag(u))\psi\|
 \le\frac\zeta{16}+\frac\zeta8
 =\frac{3\zeta}{16}<\frac\zeta4.
\end{equation*}
Thus an eigenvalue lies within $\zeta/4$ of the origin, while the common
shift $s$ ranges over $[-\zeta/2,\zeta/2]$; the same shift integral has a
positive lower limit.  Integrating over the rectangle
\eqref{eq:shifted-dos-rectangle}, using the lower density bound
\eqref{eq:shifted-dos-density-floor}, and then applying Fatou's lemma as
in \Cref{prop:positive-dos}, we obtain
\[
 n(E+t)>0
 \quad\text{for almost every }E\in(E_0-\zeta/32,E_0+\zeta/32).
\]
Equivalently, $n(E')>0$ almost everywhere in a neighborhood of
$E_0+t$.  The sets obtained as $E_0\in\sigma(A)$ and
$t\in(v_0-a,v_0+a)$ vary form an open cover of
$\sigma(A)+v_0+(-a,a)$.  Passing to a countable subcover and taking the
union of the corresponding null exceptional sets proves the claim.
\end{proof}

\subsubsection{Finite resampling and decoupling}

Good tunneling and a resonant exterior are correlated events. We
approximate the exterior event by a cylinder event, then resample its
coordinates. The logarithmic change is $o(R)$ for $R=d(o,x)$, by the
one-site rank-one resolvent formula \cite{SimonRankOne,LiawTreil}.

\begin{lemma}\label{lem:finite-resampling}
We fix $E\in\Ereg$ and distinct vertices $o,x\in X$. We let
\[
 F\subset X\setminus\{o,x\},\qquad |F|=m\ge1,
\]
and let $\omega'$ be obtained from $\omega$ by resampling the potentials on $F$. Then
\begin{equation}\label{eq:finite-resampling}
 \PP\left\{\left|\log|\tau_{ox}(E,\omega')|
 -\log|\tau_{ox}(E,\omega)|\right|>s\right\}
 \leq24\rho_+m\exp\!\left(-\frac{s}{2m}\right).
\end{equation}
The bound depends only on $m$ and $\rho_+$.
\end{lemma}
\begin{proof}
We first record the one-site estimate. We let $U,U'$ be independent with
density $\rho_\lambda$ and let
\[
 f(v)=\frac{a_1v+a_0}{c_1v+c_0}\not\equiv0
\]
be a quotient of affine functions, with nonzero denominator. Then
\begin{equation}\label{eq:one-resampling}
 \PP\!\left\{\left|\log|f(U)|-\log|f(U')|\right|>s\right\}
 \le 24\rho_+e^{-s/2}.
\end{equation}
Indeed, for one nonconstant affine factor the logarithmic difference is
$\Delta_z=\log|U-z|-\log|U'-z|$. If $|z|>2\lambda$, then
$|\Delta_z|\le\log3$. If $|z|\le2\lambda$, then
\[
 |\Delta_z|>s
 \quad\Longrightarrow\quad
 \min\{|U-z|,|U'-z|\}<3\lambda e^{-s},
\]
and the density bound \eref{eq:scaled-density} gives
$\PP\{|\Delta_z|>s\}\le12\rho_+e^{-s}$. Since $f$ has at most two
nonconstant affine factors, a union bound gives
\eref{eq:one-resampling}.
For the omitted small values of $s$, we use the trivial probability bound
one. Indeed, normalization of a density on $[-1,1]$ gives
$\rho_+\ge1/2$, and therefore
\[
 12\rho_+e^{-s}\ge12\rho_+/3\ge2
 \qquad(0\le s\le\log3).
\]
Thus the displayed affine-factor bound holds for all $s\ge0$, including
when the root lies outside $|z|\le2\lambda$. Constant affine factors
contribute zero to the logarithmic difference.

We fix one resampled vertex $u\notin\{o,x\}$ and hold all other
potentials fixed. We let $R_v$ be the resolvent on
$C_{ox}=X\setminus\{o,x\}$ with the potential at $u$ set to $v$.
The rank-one identity \cite[Lemma~3.1 and its proof]{LiawTreil}
gives, for $\ImPart z>0$,
\begin{equation}\label{eq:rank-one-resolvent}
 R_v=R_0-\frac{v\,R_0|\delta_u\rangle\langle\delta_u|R_0}
 {1+v\langle\delta_u,R_0\delta_u\rangle}.
\end{equation}
By \Cref{lem:regular-energies}, \eref{eq:rank-one-resolvent}
also holds at the boundary for almost every $v$. Substitution in
\eref{eq:tau} shows that
$\tau_{ox}(E;v_u)=(a_1v_u+a_0)/(c_1v_u+c_0)$, with coefficients
independent of $v_u$. The denominator in the rank-one formula has
constant term one and is therefore a nonzero polynomial.

In fact, with all quantities on the right computed at $v_u=0$, we set
\[
 a=\tau_{ox}(E;0),\quad
 c=\langle\delta_u,R_0\delta_u\rangle,\quad
 b=\langle b_o^{C_{ox}},R_0\delta_u\rangle
   \langle\delta_u,R_0b_x^{C_{ox}}\rangle.
\]
Then the fractional-linear dependence is exactly
\begin{equation}\label{eq:tau-fractional-linear}
 \tau_{ox}(E;v_u)=a-\frac{v_ub}{1+v_uc}
 =\frac{a+(ac-b)v_u}{1+cv_u}.
\end{equation}
The numerator in \eqref{eq:tau-fractional-linear} is also a nonzero
polynomial for almost every fixed exterior. We let
$\mathcal F_{\ne u}=\sigma(v_w:w\ne u)$ and let $Z_u$ be the
exterior event that both of its coefficients vanish. On
$Z_u$, the amplitude is zero for all $v_u$ except possibly the single
pole, which has conditional probability zero. Hence
\[
 0=\PP\{\tau_{ox}(E)=0\}
  =\E_{\ne u}\PP\{\tau_{ox}(E)=0\mid\mathcal F_{\ne u}\}
  \ge\PP_{\ne u}(Z_u).
\]
The first equality is the nonvanishing statement in
\Cref{lem:regular-energies} on the original connected graph.
The possible pole has conditional probability zero, so
\eqref{eq:one-resampling} applies.

We now order the $m$ sites of $F$. We let $\omega^{(j)}$ use the first
$j$ resampled values and retain the remaining old values, and write
$\tau(\xi)=\tau_{ox}(E,\xi)$. Thus $\omega^{(0)}=\omega$ and
$\omega^{(m)}=\omega'$, and the total change is
\[
 \log|\tau(\omega')|-\log|\tau(\omega)|
 =\sum_{j=1}^m
 \left(\log|\tau(\omega^{(j)})|-\log|\tau(\omega^{(j-1)})|\right).
\]
Each step satisfies \eref{eq:one-resampling}, conditionally on
all unchanged coordinates. If the sum exceeds
$s$ in modulus, one of its $m$ terms exceeds $s/m$. A union bound
gives \eref{eq:finite-resampling}.
\end{proof}

Thus, for every fixed resampling block and every
$\varepsilon>0$,
\[
 \PP\{\text{logarithmic change}>\varepsilon R\}
 \le24\rho_+m e^{-\varepsilon R/(2m)}\longrightarrow0.
\]
Thus finite local resampling changes $\log|\tau_{ox}|$ by $o(R)$ with
probability tending to one.

The following lemma combines a tunneling event with a favorable exterior event.
A family $(B_x)_{x\in X}$ is \emph{covariant} if
$\one_{B_{\phi x}}(\phi\cdot\omega)=\one_{B_x}(\omega)$ for every
graph automorphism $\phi$. We approximate $B_x$ by a cylinder event and
resample its finitely many coordinates.

\begin{lemma}\label{lem:joint}
We let $F_R\subseteq S_R(o)$ be nonempty deterministic endpoint sets for $R\in\mathscr R$ tending to infinity, put $M_R=|F_R|$, and suppose
\[
 \frac{\log M_R}{R}\longrightarrow\gamma>0.
\]
We fix $E\in\Ereg$ with $\mathcal L_F(E)<\gamma$, where
$\mathcal L_F$ is defined in \eref{eq:selected-cost}. We let
$(B_x)_{x\in X}$ be a covariant family of events, measurable with
respect to all potentials except $v_x$, with $\PP(B_x)=q>0$.
For every
\[
 \mathcal L_F(E)<\beta_0<\beta_1<\gamma
\]
there are $c_E>0$ and infinitely many radii $R$ such that
\begin{equation}\label{eq:joint-good}
 \frac1{M_R}\sum_{x\in F_R}
 \PP\{B_x,\ |\tau_{ox}(E)|\ge e^{-\beta_1R}\}\ge c_E.
\end{equation}
The local approximation uses a fixed-radius resampling block whose size is
independent of $R$ and $x$.
\end{lemma}
\begin{proof}
We choose $\mathcal L_F(E)<\alpha<\beta_0$. Along an infinite sequence of
radii, $\mathcal L_{F,R}(E)\le\alpha$. For those radii we set
\[
 T_{x,R}=\{|\tau_{ox}(E)|\ge e^{-\beta_0R}\},\qquad
 \delta=1-\alpha/\beta_0>0.
\]
Since $\ell_{ox}(E)>\beta_0R$ on $T_{x,R}^c$, Markov's inequality gives
\begin{equation}\label{eq:many-tunneling}
 \frac1{M_R}\sum_{x\in F_R}\PP(T_{x,R})\ge\delta.
\end{equation}

We choose $\varepsilon>0$ with
$\delta(q-\varepsilon)-\varepsilon>\delta q/2$.  We let
\[
 \mathscr F_r=\sigma(v_y:y\in B_r(o)\setminus\{o\}).
\]
The sigma-algebras $\mathscr F_r$ increase to the exterior sigma-algebra
of $o$.  Martingale convergence gives
\[
 \E\bigl[\one_{B_o}\mid\mathscr F_r\bigr]\longrightarrow\one_{B_o}
 \quad\text{in }L^1.
\]
We choose $r_0\ge1$ so that the $L^1$ distance is less than
$\varepsilon/2$, and set
\[
 C_o=\left\{\E[\one_{B_o}\mid\mathscr F_{r_0}]>\frac12\right\}.
\]
On $B_o\triangle C_o$ the pointwise difference between
$\E[\one_{B_o}\mid\mathscr F_{r_0}]$ and $\one_{B_o}$ is at least
$1/2$.  Hence
\[
 \PP(B_o\triangle C_o)
 \le2\E\left|\E[\one_{B_o}\mid\mathscr F_{r_0}]-\one_{B_o}\right|
 <\varepsilon,
 \qquad
 \PP(C_o)\ge q-\varepsilon.
\]
The event $C_o$ is measurable with respect to
$B_{r_0}(o)\setminus\{o\}$.
Transporting by automorphisms gives events $C_x$ with these bounds.
For $R>r_0$ and $x\in F_R\subset S_R(o)$, the supporting coordinates of
$C_x$ exclude both endpoints. Resample those coordinates to obtain
$\omega'$. Then
\[
 \PP\bigl(C_x(\omega')\cap T_{x,R}(\omega)\bigr)
 =\PP(C_x)\PP(T_{x,R}).
\]
If $T_{x,R}(\omega)$ holds but
$|\tau_{ox}(E,\omega')|<e^{-\beta_1R}$, the logarithmic change is larger
than $(\beta_1-\beta_0)R$. With
$m=|B_{r_0}(o)\setminus\{o\}|$, \Cref{lem:finite-resampling} bounds this
failure by
\[
 \varepsilon_R=24\rho_+m
 \exp\!\left(-\frac{(\beta_1-\beta_0)R}{2m}\right).
\]
Using equality in law of $\omega$ and $\omega'$ and replacing $C_x$ by
$B_x$ therefore gives
\[
 \PP\{B_x,|\tau_{ox}(E)|\ge e^{-\beta_1R}\}
 \ge(q-\varepsilon)\PP(T_{x,R})-\varepsilon-\varepsilon_R.
\]
We average over $F_R$ and use \eref{eq:many-tunneling}. For all sufficiently
large radii in the selected sequence, the right side is at least
$\delta q/3$, proving \eref{eq:joint-good}.
\end{proof}

The cylinder approximation, and hence the resampling block, is chosen
before $R$ and remains fixed as $R\to\infty$.

Combining \Cref{lem:alignment} and \Cref{lem:joint}, favorable exterior
environments can therefore be imposed while preserving the exponential
tunneling scale used in \Cref{sec:ac}.

\subsection{Absolutely continuous spectrum}\label{sec:ac}

We combine \eqref{eq:Poincare} with a translation bound to count distant
resonances. If the boundary imaginary part vanished, these resonances would
force root responses too large for the uniform logarithmic bounds below.
Conditional averaging then removes singular spectrum.

We formulate the resonance step for selected parts of spheres. We let $F_R\subseteq S_R(o)$ be nonempty deterministic
sets for $R\in\mathscr R$ tending to infinity, with the convention in
\eqref{eq:selected-cost}. We put $M_R=|F_R|$ and assume
\begin{equation}\label{eq:selected-growth}
 \frac{\log M_R}{R}\longrightarrow\gamma>0.
\end{equation}
For representatives $g_xo=x$, suppose that
\begin{equation}\label{eq:selected-mixing}
 \left\|\sum_{x\in F_R}U_{g_x}f\right\|_2
 \le C_F(1+R)^{s_F}\sqrt{M_R}\,\|f\|_2,
 \qquad f\in L^2_0(\Omega).
\end{equation}

\begin{proposition}\label{prop:selected-ac-point}
Assume the bounded iid model \eqref{eq:model} with
\eqref{eq:density-upper}, the Poincar\'e inequality \eref{eq:Poincare}, and
\eref{eq:selected-growth}--\eref{eq:selected-mixing}. For almost every
$E\in\Ereg$ satisfying
\[
 n_\lambda(E)>0,\qquad \mathcal L_F(E)<\gamma,
\]
one has
\begin{equation}\label{eq:selected-positive-boundary}
 \PP\{\ImPart G_{oo}(E+i0)>0\}=1.
\end{equation}
\end{proposition}

To apply \eqref{eq:Poincare}, we bound $\log W_x-\log W_y$ for $x\sim y$.
We fix $E_0<\infty$, $|E|\le E_0$ and $0<\eta\le1$. Constants may depend on
$Q,E_0,\lambda,\rho_+$, but are uniform in $E,\eta,x,y$.

One-site and two-site spectral averaging give the a priori fractional-moment estimate \cite{AizenmanMolchanov,ASFH}
\begin{equation}\label{eq:apriori}
 \E|G^C_{uv}(E+i\eta)|^\theta\le C_{\theta,\lambda},
 \qquad 0<\theta<1,
\end{equation}
for $u,v\in C$ and every finite or cofinite $C$ used below \cite[Lemma~3.1]{Tautenhahn}. To match the adjacency convention, note that Tautenhahn's compressed Laplacian on our $Q$-regular graph is $-\Delta_C=QI-A_C$, and
\[
 A_C+\lambda\omega-(E+i\eta)
 =-\bigl((-\Delta_C)-\lambda\omega-(Q-E-i\eta)\bigr).
\]
Reflection of the potential preserves its density bound, and conjugation gives the same absolute resolvent entries in either half-plane. Thus the cited estimate applies here. The bound \eqref{eq:apriori} is uniform in $u,v$.

The one-site formula \eref{eq:one-site}, applied to $H_C$ with
the neighbor vectors \eref{eq:dirichlet-data}, reads
\[
 \frac1{G^C_{uu}}
 =v_u-z-\sum_{\substack{a\sim u,\ b\sim u\\a,b\in C\setminus\{u\}}}
 G^{C\setminus\{u\}}_{ab}.
\]
The sum has at most $Q^2$ terms. Since the potential is bounded, \eqref{eq:apriori} and $|\sum_j a_j|^\theta\le\sum_j|a_j|^\theta$ imply
\begin{equation}\label{eq:inverse-log}
 \E|G^C_{uu}|^{-\theta}\le C_{\theta,Q,\lambda,E_0}.
\end{equation}

Indeed, \eqref{eq:apriori} gives
\[
 \begin{aligned}
 \E|G^C_{uu}|^{-\theta}
 &\le\lambda^\theta+(E_0+1)^\theta
   +\sum_{a,b}\E|G^{C\setminus\{u\}}_{ab}|^\theta \le\lambda^\theta+(E_0+1)^\theta+Q^2C_{\theta,\lambda}.
 \end{aligned}
\]
The bound is uniform in $C$ and $\eta$.

Applying \eqref{eq:apriori} with $u=v$, combining it with
\eref{eq:inverse-log}, and using
\[
 (\log t)^2\le C_\theta(t^\theta+t^{-\theta}),\qquad t>0,
\]
we obtain
\begin{equation}\label{eq:diagonal-log}
 \E\bigl|\log|G^C_{uu}|\bigr|^2\le C.
\end{equation}

To compare $W_x$ and $W_y$, we keep the finite star around one vertex and
write the relevant numerators as vector-valued polynomials in the potentials
on that star. Each variable appears with degree at most one; such a
polynomial is called \emph{multiaffine}. The determinant cancels in the
ratio $W_x/W_y$. Dividing each cofactor
polynomial by its largest coefficient then gives a bound uniform in the
exterior variables and in $\eta$. The next lemma supplies this bound.

The next estimate is a multiaffine sublevel bound; compare
\cite{CarberyWright}. Its constant is uniform in the coefficients and
the dimension of the target space.

\begin{lemma}\label{lem:poly}
We let $t_1,\ldots,t_m$ be independent variables on $[-1,1]$, each with a density bounded by $M_\rho<\infty$. If a nonzero complex multiaffine polynomial $p(t)$ has largest coefficient modulus one, then
\[
 \E(\log^-|p(t)|)^2\le C_{m,M_\rho}.
\]
More generally, let $P$ be a nonzero multiaffine polynomial with values in a finite-dimensional Hilbert space, and let $c_*$ be its largest coefficient norm. The same bound holds with $|p|$ replaced by $\norm P/c_*$. The constant is independent of the coefficients and the dimension of the
Hilbert space.
\end{lemma}
\begin{proof}
We induct on $m$; for $m=0$, the normalized polynomial is a constant of
modulus one. We write $t'=(t_1,\ldots,t_{m-1})$ and
\[
 p(t)=a(t')t_m+b(t'),\qquad M(t')=\max\{|a(t')|,|b(t')|\}.
\]
We fix $t'$ with $M(t')>0$. If $|a|<M/2$, then $|b|=M$ and
$|p|\ge |b|-|a|\ge M/2$ on $[-1,1]$. Otherwise, for
$0<\varepsilon<1/2$,
\[
 \PP\{|p|<\varepsilon M\mid t'\}\le C\varepsilon:
\]
the relevant disk has radius $\varepsilon M/|a|\le2\varepsilon$ and meets the real axis in an interval of length at most $4\varepsilon$. We set $L=\log^-(|p|/M)$. Increasing the constant also covers
$\varepsilon\ge1/2$, so
\[
 \PP\{L>s\mid t'\}\le\min\{1,C_0e^{-s}\},\qquad s\ge0.
\]
The layer-cake formula then gives
\[
 \E(L^2\mid t')
 =2\int_0^\infty s\,\PP(L>s\mid t')\,ds
 \le2\int_0^\infty s\min\{1,C_0e^{-s}\}\,ds<\infty.
\]
The last integral depends only on the density bound.
One of $a,b$, say $q$, has a coefficient of modulus one, and none has
a larger coefficient. The induction hypothesis gives $q\ne0$ almost surely
and bounds $\|\log^-|q|\|_2$. Since $M\ge|q|$, the omitted event
$M=0$ is null and $\|\log^-M\|_2$ has the same bound. We now use
\[
 \log^-|p|\le\log^-M+\log^-(|p|/M)
\]
and the $L^2$ triangle inequality:
\[
 \|\log^-|p|\|_2
 \le\|\log^-M\|_2+\|\log^-(|p|/M)\|_2
 \le C_{m-1,M_\rho}^{1/2}+C^{1/2}.
\]
This closes the induction and proves the scalar assertion.

For the vector assertion, we write
\[
 P(t)=\sum_{J\subseteq\{1,\ldots,m\}}c_J\prod_{j\in J}t_j,
 \qquad \|c_{J_*}\|=c_*,\qquad e=c_{J_*}/c_*.
\]
The scalar polynomial $p(t)=\langle e,P(t)\rangle/c_*$ has a
coefficient equal to one, while all its other coefficients have modulus
at most one. Also,
\[
 |p(t)|\le\|P(t)\|/c_*,\qquad
 \log^-\bigl(\|P(t)\|/c_*\bigr)\le\log^-|p(t)|.
\]
The scalar bound also bounds the vector expression, uniformly in dimension.
\end{proof}

There is also a deterministic upper bound in the bounded-variable case:
if $P(t)=\sum_Jc_J\prod_{j\in J}t_j$ and $|t_j|\le1$, then
\[
 \|P(t)\|\le\sum_J\|c_J\|\le2^m c_*.
\]
Thus the positive logarithm of the normalized polynomial is at most
$m\log2$, while \Cref{lem:poly} bounds the negative logarithm.

\begin{lemma}\label{lem:edge}
We let $H=A+\diag(v_u)$ have unit adjacency weights on a graph of degree at most $Q$, where $v_u=\lambda t_u$ and the real variables $t_u$ are independent with common law $\vartheta$. Its domain is
\begin{equation}\label{eq:multiplication-domain}
 \mathcal D(H)=\{\psi\in\ell^2(X):\sum_u|v_u\psi(u)|^2<\infty\}.
\end{equation}
We write
$G_{xy}(z)=\langle\delta_x,(H-z)^{-1}\delta_y\rangle$ and $W_x(z)=\ImPart G_{xx}(z)$. Assume $\E\log^2(1+|t_u|)<\infty$, the uniform bound \eqref{eq:apriori} for some $0<\theta<1$, and the following estimate: every nonzero Hilbert-space-valued multiaffine polynomial $P$ in $m\le Q+1$ variables, with largest coefficient norm $c_*$, satisfies
\begin{equation}\label{eq:poly-input}
 \E\left(\log^-\frac{\|P(t)\|}{c_*}\right)^2\le C_m,
\end{equation}
where $C_m$ is independent of the coefficients and target dimension. Then, for every $E_0<\infty$, there is $C<\infty$ such that
\begin{equation}\label{eq:edge}
 \sup_{\substack{|E|\le E_0,\ 0<\eta\le1\\x\sim y}}
 \E\left|\log W_x(E+i\eta)-\log W_y(E+i\eta)\right|^2\le C.
\end{equation}
The bounded law \eqref{eq:density-upper} in the model
\eqref{eq:model} satisfies these hypotheses by \Cref{lem:poly} and \eqref{eq:apriori}. 
\end{lemma}
\begin{proof}
We fix $x\sim y$ and retain the star $F=\{x\}\cup N(x)$. We put $C=F^c$
and $m=|F|\le Q+1$. Keeping all neighbors of $x$ preserves its unit couplings. We condition on the potentials in $C$ and set
\begin{equation}\label{eq:star-data}
 \mathcal D_F(z)=zI-A_F+A_{FC}(H_C-z)^{-1}A_{CF},
 \qquad \mathcal B_F=\ImPart\mathcal D_F.
\end{equation}
Then
\[
 \mathcal B_F=\eta I+A_{FC}\ImPart(H_C-z)^{-1}A_{CF}>0.
\]
The Grushin formula \eqref{eq:grushin-effective} gives
\begin{equation}\label{eq:star-reduction}
 E_{F,-+}(z)=\mathcal D_F(z)-\diag(v_F),\qquad
 \mathcal M(v_F):=-E_{F,-+}(z)=\diag(v_F)-\mathcal D_F(z),
\end{equation}
and hence
\[
 G_F(z)\coloneq \one_F(H-z)^{-1}\one_F=\mathcal M(v_F)^{-1}.
\]
Only $v_u$, $u\in F$, remain random. We number the center $x$ by $0$ and
its neighbors by $1,\ldots,m-1$. All neighbors of the center lie in $F$, so the
exterior correction has zero row and column there. Thus,
\begin{equation}\label{eq:star-row}
 (\mathcal D_F)_{00}=z,\quad
 (\mathcal D_F)_{0j}=(\mathcal D_F)_{j0}=-1,\quad
 (\mathcal B_F)_{00}=\eta,\quad
 (\mathcal B_F)_{0j}=0\quad(j\ne0).
\end{equation}
For an invertible matrix $M$, we write $M^{-*}=(M^*)^{-1}$.
The imaginary part of its inverse is computed as follows:
\begin{equation}\label{eq:inverse-imaginary-part}
 \begin{aligned}
 M^{-1}-M^{-*}
 &=M^{-*}(M^*-M)M^{-1},\\
 \ImPart(M^{-1})
 &=\frac{M^{-1}-M^{-*}}{2i}
 =-M^{-*}(\ImPart M)M^{-1}.
 \end{aligned}
\end{equation}
Applying \eref{eq:inverse-imaginary-part} to
\eref{eq:star-reduction} gives

\[
 \ImPart G_F=\mathcal M(v_F)^{-*}\mathcal B_F\mathcal M(v_F)^{-1}.
\]
We let $\adj(M)$ denote the transposed cofactor matrix, so that $M^{-1}=\adj(M)/\det M$. For each standard basis vector $e_u\in\C^F$, we put
\begin{equation}\label{eq:star-polynomial}
 P_u(v_F)=\mathcal B_F^{1/2}\adj(\mathcal M(v_F))e_u.
\end{equation}
Taking the $u$th diagonal entry of
\eref{eq:inverse-imaginary-part}, using \eref{eq:star-polynomial}, gives
\begin{equation}\label{eq:star-poly}
 W_u=\frac{\|P_u(v_F)\|^2}{|\det\mathcal M(v_F)|^2}.
\end{equation}
Each $P_u$ is multiaffine because the random variables occupy distinct diagonal entries of $\mathcal M$.

For the center cofactor column, we write $N=F\setminus\{x\}$, let $\mathbf 1\in\C^N$ have all entries
one, and put
\[
 \mathcal N=\diag(v_N)-(\mathcal D_F)_{NN},\qquad
 \mathcal M=
 \begin{pmatrix}
 v_0-z&\mathbf 1^{\mathsf T}\\
 \mathbf 1&\mathcal N
 \end{pmatrix}.
\]
Using the unit entries in \eqref{eq:star-row}, the adjugate identity gives
\[
 \adj(\mathcal M)e_0
 =\begin{pmatrix}
       \det\mathcal N\\
       -\adj(\mathcal N)\mathbf 1
   \end{pmatrix}.
\]
For example, multiplying this column by $\mathcal M$ gives zero in
the lower block, since
$\mathcal N\adj(\mathcal N)=(\det\mathcal N)I$. Its top entry is
\[
 (v_0-z)\det\mathcal N
      -\mathbf 1^{\mathsf T}\adj(\mathcal N)\mathbf 1
 =\det\mathcal M.
\]
Polynomial continuation covers singular $\mathcal N$. The center column
is independent of $v_0$.

The coefficient of $\prod_{j=1}^{m-1}v_j$ in $P_0$ is $\mathcal B_F^{1/2}e_0$. For $1\le j\le m-1$, the coefficient of $\prod_{k\ne j,\ k\ge1}v_k$ is
\begin{equation}\label{eq:cofactor}
 \mathcal B_F^{1/2}\bigl(-(\mathcal D_F)_{jj}e_0-e_j\bigr).
\end{equation}
We expand the column $\adj(\mathcal M)e_0$ in cofactors. In its $e_0$
component, we choose every diagonal variable except $v_j$ in the determinant
on the neighbors. The remaining factor is $-(\mathcal D_F)_{jj}$.
In its $e_j$ component, the same monomial leaves the signed coupling from
the center to $j$, equal to $-1$ by \eqref{eq:star-row}. Every other
component omits a different diagonal variable, so the coefficient of this
monomial is zero. Multiplying by $\mathcal B_F^{1/2}$ gives
\eqref{eq:cofactor}.

The orthogonality in \eqref{eq:star-row} gives
\begin{gather*}
 \left\|\mathcal B_F^{1/2}
 \bigl(-(\mathcal D_F)_{jj}e_0-e_j\bigr)\right\|^2
 =\eta |(\mathcal D_F)_{jj}|^2+(\mathcal B_F)_{jj}
 \ge(\mathcal B_F)_{jj},\text{ and } \|\mathcal B_F^{1/2}e_0\|^2=\eta.
\end{gather*}
We substitute $v_u=\lambda t_u$ and let $c_*$ be the largest coefficient
norm of $P_0(\lambda t)$. The two coefficients just computed acquire
factors $\lambda^{m-1}$ and $\lambda^{m-2}$, respectively. They control
every diagonal entry of $\mathcal B_F$. Since at least one of those
entries is at least $m^{-1}\Tr\mathcal B_F$, we obtain
\begin{equation}\label{eq:coefficient-lower}
 c_*\ge c_{m,\lambda}\sqrt{\Tr\mathcal B_F},
 \qquad c_{m,\lambda}>0.
\end{equation}
Here $m\ge2$, and we may take $c_{m,\lambda}=m^{-1/2}\min\{\lambda^{m-1},\lambda^{m-2}\}$.

Indeed the individual coefficient bounds are
\[
 c_*^2\ge\lambda^{2m-2}(\mathcal B_F)_{00},\qquad
 c_*^2\ge\lambda^{2m-4}(\mathcal B_F)_{jj}\quad(1\le j<m).
\]
Adding after dividing by the smaller of the two powers gives
\[
 \Tr\mathcal B_F
 \le\frac{m c_*^2}{\min\{\lambda^{2m-2},\lambda^{2m-4}\}},
\]
which is \eref{eq:coefficient-lower}. This comparison remains useful even
when $\mathcal B_F$ itself becomes small as $\eta\downarrow0$.

Every cofactor has degree at most $m-1$, and each coefficient is a sum of at most $(m-1)!$ products of entries of $\mathcal D_F$. Thus, for $y\in F$ and real $t_u$,
\begin{equation}\label{eq:star-numerator-upper}
 \|P_y(\lambda t)\|
 \le C_{m,\lambda}(1+\|\mathcal D_F\|)^{m-1}
       \sqrt{\Tr\mathcal B_F}\prod_{u\in F}(1+|t_u|).
\end{equation}
In \eqref{eq:star-poly}, the denominator $|\det\mathcal M|^2$ is the same
for $W_x$ and $W_y$,
so it cancels in their ratio. Combining \eref{eq:star-numerator-upper}
with \eqref{eq:coefficient-lower}, we bound $\log^+(W_y/W_x)$ by a constant times
\[
 1+\log(1+\|\mathcal D_F\|)+\sum_{u\in F}\log(1+|t_u|)
       +\log^-\frac{\|P_0(\lambda t)\|}{c_*}.
\]

The cancellation can be written as a ratio before taking logarithms:
\begin{equation}\label{eq:star-response-ratio}
 \begin{aligned}
 \frac{W_y}{W_x}
 &=\frac{\|P_y(\lambda t)\|^2}{\|P_0(\lambda t)\|^2}\le C_{m,\lambda}(1+\|\mathcal D_F\|)^{2(m-1)}
   \prod_{u\in F}(1+|t_u|)^2
   \left(\frac{c_*}{\|P_0(\lambda t)\|}\right)^2.
 \end{aligned}
\end{equation}
The cancellation of $\Tr\mathcal B_F$ makes the logarithmic comparison
uniform in $\eta$, even as both responses tend to zero.

We let $\E_F$ denote conditional expectation over the potentials in $F$.
Taking logarithms in \eref{eq:star-response-ratio}, then using
\eqref{eq:poly-input} and the logarithmic moment of $t_u$, gives
\[
 \E_F\left(\log^+\frac{W_y}{W_x}\right)^2
 \le C_{m,\lambda,\vartheta}
       \bigl(1+\log^2(1+\|\mathcal D_F\|)\bigr).
\]
To average this conditional bound, we expand an entry of
\eref{eq:star-data}:
\[
 (\mathcal D_F)_{ij}
 =z\delta_{ij}-(A_F)_{ij}
   +\sum_{a,b\in C}A_{ia}G^C_{ab}(z)A_{bj}.
\]
There are at most $Q^2$ nonzero terms in the sum. Since
$\|\mathcal D_F\|\le\sum_{i,j}|(\mathcal D_F)_{ij}|$, subadditivity
for $0<\theta<1$ and \eref{eq:apriori} give
\[
 \begin{aligned}
 \E\|\mathcal D_F\|^\theta
 &\le\sum_{i,j\in F}\E|(\mathcal D_F)_{ij}|^\theta\le m^2\bigl((E_0+1)^\theta+1+Q^2C_{\theta,\lambda}\bigr).
 \end{aligned}
\]
Applying $(\log(1+t))^2\le C_\theta(1+t^\theta)$ yields
\[
 \E\log^2(1+\|\mathcal D_F\|)\le C.
\]
Averaging over the exterior bounds $\E(\log^+(W_y/W_x))^2$. Repeating
the argument with the star centered at $y$ bounds
$\E(\log^+(W_x/W_y))^2$ and proves \eqref{eq:edge}.

For positive $W_x,W_y$, the final combination uses the identity
\[
 |\log W_y-\log W_x|^2
 =\left(\log^+\frac{W_y}{W_x}\right)^2
  +\left(\log^+\frac{W_x}{W_y}\right)^2.
\]
Exactly one of the two positive logarithms can be nonzero.

\end{proof}

Following \cite[Sections~4.1--4.2]{AWtree}, we bound the centered logarithms.
The translation estimate \eqref{eq:selected-mixing} will then control
second moments of resonance counts.

\begin{proposition}\label{prop:tightness}
For a covariant selfadjoint operator $H=A+\diag(v)$, assume \eqref{eq:Poincare}, the edge bound \eqref{eq:edge}, and the uniform diagonal logarithmic bound \eqref{eq:diagonal-log}. Then
\begin{equation}\label{eq:geometric-response}
 m_\eta(E)=\exp(\E\log W_o(E+i\eta))
\end{equation}
is well defined and positive. For every $E_0<\infty$, there is $C<\infty$ such that, for $|E|\le E_0$, $0<\eta\le1$, and every $x\in X$,
\begin{align}
 \E\left|\log\frac{W_x(E+i\eta)}{m_\eta(E)}\right|^2&\le C,\label{eq:W-tight}\\
 \E\left|\log\frac{Y_x(E+i\eta)}{m_\eta(E)}\right|^2&\le C.\label{eq:Y-tight}
\end{align}
Thus, for $T>0$,
\begin{equation}\label{eq:tail}
 \PP\left\{\left|\log\frac{W_o}{m_\eta}\right|>T\right\}\le\frac C{T^2}.
\end{equation}
\end{proposition}
\begin{proof}
The Ward identity \eqref{eq:Ward} and the resolvent norm bound give
\[
 \eta|G_{oo}|^2\le W_o\le\eta^{-1},\qquad
 |\log W_o|\le|\log\eta|+2|\log|G_{oo}||.
\]
Thus \eqref{eq:diagonal-log} gives $\log W_o\in L^2$ for each fixed
$\eta>0$, and \eqref{eq:geometric-response} is well defined.
For $t\in T$, we choose a shortest path $o=x_0,\ldots,x_{d_t}=to$.
The edge bound \eqref{eq:edge} and Cauchy--Schwarz give
\[
 \begin{aligned}
 \E|\log W_o-\log W_{to}|^2
 &=\E\left|\sum_{j=1}^{d_t}(\log W_{x_{j-1}}-\log W_{x_j})\right|^2\\
 &\le d_t\sum_{j=1}^{d_t}\E|\log W_{x_{j-1}}-\log W_{x_j}|^2
 \le C_{\rm edge}d_t^2.
 \end{aligned}
\]
Since $\log m_\eta=\E\log W_o$, covariance and \eqref{eq:Poincare} imply
\[
 \E\left|\log\frac{W_o}{m_\eta}\right|^2
 =\Var(\log W_o)
 \le C_{\rm P}C_{\rm edge}\sum_{t\in T}d(o,to)^2.
\]
The sum is finite, proving \eqref{eq:W-tight}; covariance transfers it
to every $x$. Equation~\eqref{eq:ward} gives
$\log Y_x=\log W_x-2\log|G_{xx}|$, hence
\[
 \E\left|\log\frac{Y_x}{m_\eta}\right|^2
 \le2\E\left|\log\frac{W_x}{m_\eta}\right|^2+8\E|\log|G_{xx}||^2
 \le C
\]
by \eqref{eq:W-tight} and \eqref{eq:diagonal-log}. This is
\eqref{eq:Y-tight}; Chebyshev's inequality gives \eqref{eq:tail}.
\end{proof}

The bounds \eqref{eq:W-tight}--\eqref{eq:Y-tight} imply that $W_o(E+i0)$
either vanishes almost surely or is positive almost surely. In the positive
case they also bound $m_\eta$ away from zero.

\begin{lemma}\label{lem:boundary-dichotomy}
We fix a regular energy $E$. Under \eqref{eq:W-tight},
\[
 \PP\{W_o(E+i0)>0\}\in\{0,1\}.
\]
If this probability is one, then $m_\eta(E)$ is bounded above and bounded away from zero for all sufficiently small $\eta$.
\end{lemma}
\begin{proof}
We let $W_0=W_o(E+i0)$, which is finite almost surely by
\Cref{lem:regular-energies}.
Suppose $\PP\{W_0>0\}>0$. If $m_{\eta_j}\to0$ along a sequence
$\eta_j\downarrow0$, then
\[
 \log\frac{W_o(E+i\eta_j)}{m_{\eta_j}}\longrightarrow+\infty
 \quad\text{on }\{W_0>0\}.
\]
On an event of positive probability, the squared logarithm tends to
infinity, whereas Fatou and \eqref{eq:W-tight} give
\[
 \E\liminf_j\left|\log\frac{W_o(E+i\eta_j)}{m_{\eta_j}}\right|^2
 \le\liminf_j\E\left|\log\frac{W_o(E+i\eta_j)}{m_{\eta_j}}\right|^2
 \le C.
\]
This contradicts the finite integral, because a nonnegative function that is infinite on a
set of positive probability has infinite integral. If $m_{\eta_j}\to\infty$, the logarithm tends
to $-\infty$ and gives the same contradiction. Thus $m_\eta$ stays between
positive constants for small $\eta$.

On $\{W_0=0\}$, it now follows that $\log(W_o(E+i\eta)/m_\eta)\to-\infty$. Another application of Fatou's lemma and \eqref{eq:W-tight} shows that this event has probability zero. Therefore positive probability of $W_0>0$ implies probability one.
\end{proof}

\subsubsection{Distant resonances}
We compare the root response before and after deleting a distant vertex $x$.
The comparison keeps a subtractive term for possible cancellation and
gives a lower bound when $x$ is resonant.

We fix $x\ne o$. Deleting $x$ and its incident edges gives the Dirichlet resolvent $R^x(z)=(H_{X\setminus\{x\}}-z)^{-1}$. We set $G_{oo}^{(x)}(z)=\langle\delta_o,R^x(z)\delta_o\rangle$ and
\begin{equation*}
 W_o^x(z)=\ImPart G_{oo}^{(x)}(z),\qquad
 a_o^{(ox)}(z)=\bigl(G_{oo}^{(x)}(z)\bigr)^{-1}.
\end{equation*}
These cavity quantities are computed after deleting $x$ and depend only
on its exterior. Equations \eqref{eq:one-site} and
\eqref{eq:two-site-self-energy} give $a_o^{(ox)}=v_o-\sigma_o^{(ox)}$.

\begin{lemma}\label{lem:cavity}
We fix $E_0<\infty$. Uniformly for $|E|\le E_0$, $x\ne o$, and $0<\eta\le1$,
\begin{equation}\label{eq:cavity-log}
 \E\left(\log^+\frac{W_o^x}{W_o}\right)^2\le C,
 \qquad \E|a_o^{(ox)}|^\theta\le C\quad(0<\theta<1).
\end{equation}
Thus, for every $\varepsilon>0$, there are $c>0$, $C_0<\infty$ and $A_0<\infty$ such that
\begin{equation}\label{eq:environment-tail}
 \PP\{Y_x<c m_\eta\}
 +\PP\{W_o^x>C_0m_\eta\}
 +\PP\{|a_o^{(ox)}|>A_0\}<\varepsilon
\end{equation}
throughout the same ranges of $E,x,\eta$.
\end{lemma}
\begin{proof}
In the decomposition $\ell^2(X)=\C\delta_x\oplus\ell^2(X\setminus\{x\})$, we let $b_x=b_x^{X\setminus\{x\}}$ and put
\begin{gather*}
 u=(0,R^x\delta_o),\qquad w=(1,-R^x b_x),\\
 \alpha=\langle b_x,R^x\delta_o\rangle,
 \qquad g=(v_x-\Sigma_x)^{-1}=G_{xx}(z).
\end{gather*}
We write $(H-z)^{-1}\delta_o=(\psi_x,\psi_C)$ with
$C=X\setminus\{x\}$. The two block equations are
\[
 (v_x-z)\psi_x+\langle b_x,\psi_C\rangle=0,
 \qquad b_x\psi_x+(H_C-z)\psi_C=\delta_o.
\]
We solve the second equation first and substitute it into the first:
\[
 \psi_C=R^x\delta_o-\psi_xR^xb_x,
 \qquad
 (v_x-z-\langle b_x,R^xb_x\rangle)\psi_x=-\alpha.
\]
Since $z+\langle b_x,R^xb_x\rangle=\Sigma_x$, this gives
\[
 \psi_x=-\alpha g,
 \qquad\psi_C=R^x\delta_o+\alpha gR^xb_x.
\]
Combining the two components proves
\begin{equation}\label{eq:vector}
 (H-z)^{-1}\delta_o=u-\alpha g w.
\end{equation}
We define the exterior-dependent scalar
\[
 c_x=\Sigma_x+\alpha\langle u,w\rangle/\|u\|^2.
\]
Multiplying \eqref{eq:vector} by $v_x-\Sigma_x$ gives
\[
 (v_x-\Sigma_x)(H-z)^{-1}\delta_o=(v_x-\Sigma_x)u-\alpha w.
\]
The vector $u$ is nonzero because $R^x\delta_o\ne0$.
Taking its inner product with the last display makes the scalar
normalization explicit:
\[
 \begin{aligned}
 \left\langle u,(v_x-\Sigma_x)u-\alpha w\right\rangle
 &=(v_x-\Sigma_x)\|u\|^2-\alpha\langle u,w\rangle=\|u\|^2(v_x-c_x).
 \end{aligned}
\]
Cauchy--Schwarz and division by $\|u\|>0$ therefore yield
\[
 |v_x-\Sigma_x|\,\|(H-z)^{-1}\delta_o\|
 \ge\|u\|\,|v_x-c_x|.
\]
The complex center $c_x$ is exterior-measurable. A disk of radius $a$
around it meets the real axis in length at most $2a$.
Applying \eref{eq:Ward} to $H$ and to $H_{X\setminus\{x\}}$
gives $W_o=\eta\|(H-z)^{-1}\delta_o\|^2$ and
$W_o^x=\eta\|u\|^2$, so
\begin{gather*}
 W_o\ge W_o^x\frac{|v_x-c_x|^2}{|v_x-\Sigma_x|^2},\\
 \log^+(W_o^x/W_o)
 \le2\log^+|v_x-\Sigma_x|+2\log^-|v_x-c_x|.
\end{gather*}
The first logarithm has uniformly bounded second moment by \eqref{eq:inverse-log}, because \eref{eq:one-site} gives $v_x-\Sigma_x=G_{xx}^{-1}$. For the second, conditioning outside $x$ and using
\eref{eq:scaled-density} gives
\[
 \PP\{|v_x-c_x|<e^{-t}\mid(v_y)_{y\ne x}\}
 \le2\norm{\rho_\lambda}_\infty e^{-t},\qquad t\ge0.
\]
With $L_x=\log^-|v_x-c_x|$ and $C_\lambda=2\|\rho_\lambda\|_\infty$,
\[
 \E(L_x^2\mid(v_y)_{y\ne x})
 \le2\int_0^\infty t\min\{1,C_\lambda e^{-t}\}\,dt<\infty.
\]
This proves the first bound in \eqref{eq:cavity-log}, uniformly in the
center $c_x$. Applying \eqref{eq:inverse-log} to $G_{oo}^{(x)}$ proves the second.

For \eqref{eq:environment-tail}, we use \eqref{eq:Y-tight} and
\[
 \log^+(W_o^x/m_\eta)
 \le\log^+(W_o^x/W_o)+\log^+(W_o/m_\eta).
\]
By \eqref{eq:cavity-log}, \eqref{eq:W-tight}, and Markov's inequality,
for $0<c<1<C_0$ and $A_0>0$,
\[
 \PP\{Y_x<cm_\eta\}\le\frac{C}{|\log c|^2},\qquad
 \PP\{W_o^x>C_0m_\eta\}\le\frac{C}{(\log C_0)^2},\qquad
 \PP\{|a_o^{(ox)}|>A_0\}\le C_\theta A_0^{-\theta}.
\]
We decrease $c$ and increase $C_0,A_0$ until each bound is below
$\varepsilon/3$. The constants are uniform in the stated ranges.

\end{proof}

The vector identity \eref{eq:vector} also gives a lower bound on
the response at $o$.
The subtracted term is the response with $x$ removed; it allows for
cancellation when $x$ is restored.

\begin{lemma}\label{lem:amplification}
For every $z=E+i\eta$, $\eta>0$, and every $x\ne o$,
\begin{equation}\label{eq:amplification}
 \sqrt{W_o}\ge
 \frac{|\tau_{ox}|}{|a_o^{(ox)}|}\,|G_{xx}|\sqrt{Y_x}
 -\sqrt{W_o^x}.
\end{equation}
\end{lemma}
\begin{proof}
The reverse triangle inequality in \eqref{eq:vector} gives
\[
 \sqrt\eta\,\|(H-z)^{-1}\delta_o\|
 \ge |\alpha|\,|g|\sqrt\eta\,\|w\|-\sqrt\eta\,\|u\|.
\]
Equations \eqref{eq:Ward}, \eqref{eq:kappa-eta} and \eqref{eq:ward} identify
\[
 \sqrt\eta\,\|(H-z)^{-1}\delta_o\|=\sqrt{W_o},\qquad
 \sqrt\eta\,\|w\|=\sqrt{Y_x},\qquad
 \sqrt\eta\,\|u\|=\sqrt{W_o^x}.
\]
The $x$ coordinate of \eqref{eq:vector} gives $G_{xo}=-\alpha G_{xx}$.
By symmetry and \eqref{eq:pair-identities}, with $x$ and $o$ interchanged,
\[
 |\alpha|=\left|\frac{G_{ox}}{G_{xx}}\right|
 =\left|\frac{\tau_{ox}}{v_o-\sigma_o^{(ox)}}\right|
 =\frac{|\tau_{ox}|}{|a_o^{(ox)}|}.
\]
The denominators are nonzero at $\eta>0$ by the positive imaginary parts
of the diagonal resolvents. We substitute these identities and $g=G_{xx}$
into the first display to obtain \eqref{eq:amplification}.
\end{proof}

To count resonances, we enlarge the site-dependent events to translates of a single diagonal event. This permits direct use of \eqref{eq:selected-mixing}.

\begin{lemma}\label{lem:RD-count}
We fix $z=E+i\eta$ with $\eta>0$, $t>0$, and a finite set
$F_R\subseteq S_R(o)$ of size $M_R$. Assume \eref{eq:selected-mixing}
for this family. We let
\[
 A_x\subseteq\{|G_{xx}(z)|\ge1/(\sqrt2t)\},\qquad x\in F_R,
\]
and put $Z=\sum_{x\in F_R}\one_{A_x}$. There are constants independent of
$R,t,\eta$ such that
\begin{equation}\label{eq:RD-count}
 \E Z^2\le C_1(M_Rt)^2+C_2(1+R)^{2s_F}M_Rt.
\end{equation}
\end{lemma}
\begin{proof}
We set $\mathcal R_x=\{|G_{xx}(z)|\ge1/(\sqrt2t)\}$ and
$N_R=\sum_{x\in F_R}\one_{\mathcal R_x}$. Conditional on all potentials
except $v_o$, the one-site identity gives
$\mathcal R_o=\{|v_o-\Sigma_o(z)|\le\sqrt2t\}$, hence
\[
 p\coloneq\PP(\mathcal R_o)\le2\sqrt2\,\|\rho_\lambda\|_\infty t.
\]
For $f=\one_{\mathcal R_o}-p$, covariance gives
\[
 N_R=M_Rp+\sum_{x\in F_R}U_{g_x}f,
 \qquad \E f=0,\qquad \|f\|_2^2=p(1-p).
\]
The constant and centered terms are orthogonal. By
\eref{eq:selected-mixing},
\[
 \E N_R^2\le(M_Rp)^2+
 C_F^2(1+R)^{2s_F}M_Rp(1-p).
\]
Since $0\le Z\le N_R$, substitution of the bound for $p$ proves
\eref{eq:RD-count}.
\end{proof}

\subsubsection{Positivity of the boundary response}
Assume the boundary imaginary part vanishes almost surely. We will find,
with probability at least $p_*>0$, a distant resonance such that
\[
 \log(W_o/m_\eta)\ge cR-O(1).
\]
This contradicts \eqref{eq:tail}, which bounds the probability of such a
fluctuation by $O(R^{-2})$.

\begin{proof}[Proof of \Cref{prop:selected-ac-point}]
We fix a regular energy $E$ with $n_\lambda(E)>0$ and
$\mathcal L_F(E)<\gamma$; constants may depend on $E$. By
\Cref{lem:boundary-dichotomy}, it suffices to exclude
$\PP\{W_o(E+i0)=0\}=1$. Assume this vanishing.

By \Cref{lem:alignment}, the events $B_x^0$ in \eref{eq:B} are
independent of $v_x$, have common probability $q>0$, and satisfy
\[
 \int_{\Sigma_x(E)-t}^{\Sigma_x(E)+t}\rho_\lambda(v)\,dv\ge2at
 \qquad(0<t<t_0)
\]
on $B_x^0$, with fixed $a,t_0>0$. We choose
\[
 \mathcal L_F(E)<\alpha<\beta_0<\beta_1<\beta_2<\gamma.
\]
We apply \Cref{lem:joint} to the endpoint sets $F_R$ and the events $B_x^0$.
With $u_R=e^{-\beta_1R}$, along an infinite sequence of radii,
\begin{equation}\label{eq:ac-joint}
 \frac1{M_R}\sum_{x\in F_R}
 \PP\{B_x^0,\ |\tau_{ox}(E)|\ge u_R\}\ge c_0
\end{equation}
for some $c_0>0$. We put
\begin{equation}\label{eq:scales}
 t_R=e^{-\beta_2R},\qquad
 M_Rt_R=e^{(\gamma-\beta_2)R+o(R)}\longrightarrow\infty,
 \qquad
 \frac{u_R}{t_R}=e^{(\beta_2-\beta_1)R}\longrightarrow\infty.
\end{equation}
From now on $R$ belongs to this selected sequence.

We fix such an $R$ with $t_R<t_0$. Boundary regularity and covariance give
$Y_x(E+i\eta)\to0$ almost surely at every $x$.  On the event in
\eref{eq:ac-joint}, the indicators of the additional conditions below have
lower limit one as $\eta\downarrow0$.  Fatou's lemma for each $x$, followed
by the finite average over $F_R$, therefore gives
\begin{equation}\label{eq:finite-env}
 \frac1{M_R}\sum_{x\in F_R}\PP\left\{
 \begin{gathered}
 B_x^0,\quad |\tau_{ox}(E+i\eta)|\ge u_R/2,\quad
 Y_x(E+i\eta)\le t_R,\\
 |\RePart\Sigma_x(E+i\eta)-\Sigma_x(E)|\le t_R/2
 \end{gathered}\right\}\ge c_0/2
\end{equation}
for all sufficiently small $\eta>0$.

We apply \eref{eq:environment-tail} with total exceptional probability below
$c_0/4$. This gives $c>0$, $C_0,A_0<\infty$, independent of
$R,\eta,x$. Intersect the event in \eref{eq:finite-env} with
\begin{equation}\label{eq:extra}
 Y_x\ge cm_\eta,\qquad W_o^x\le C_0m_\eta,
 \qquad |a_o^{(ox)}|\le A_0,
\end{equation}
and call the result $F_{x,R,\eta}^0$. Its average probability over
$x\in F_R$ is at least $c_0/4$, and it is measurable with respect to the
exterior of $x$.

We define
\[
 A_{x,R,\eta}=F_{x,R,\eta}^0\cap
 \{|v_x-\RePart\Sigma_x(E+i\eta)|\le t_R\},\qquad
 Z_{R,\eta}=\sum_{x\in F_R}\one_{A_{x,R,\eta}}.
\]
On $F_{x,R,\eta}^0$, the resonance window contains
$[\Sigma_x(E)-t_R/2,\Sigma_x(E)+t_R/2]$. Conditional integration in
$v_x$ and the alignment bound give
\begin{equation*}
 \E Z_{R,\eta}\ge c_1M_Rt_R,
 \qquad c_1=ac_0/4>0.
\end{equation*}
At each counted site, the one-site identity yields
$|G_{xx}(E+i\eta)|\ge(\sqrt2t_R)^{-1}$. Hence \Cref{lem:RD-count} gives
\[
 \E Z_{R,\eta}^2
 \le C_1(M_Rt_R)^2+C_2(1+R)^{2s_F}M_Rt_R.
\]
By \eref{eq:scales},
\begin{equation*}
 \frac{(1+R)^{2s_F}M_Rt_R}{(M_Rt_R)^2}
 =\frac{(1+R)^{2s_F}}{M_Rt_R}
 =(1+R)^{2s_F}e^{-(\gamma-\beta_2)R+o(R)}\longrightarrow0.
\end{equation*}
Consequently, after increasing $R_0$ if necessary,
\begin{align}
 \PP\{Z_{R,\eta}>0\}
 &\ge\frac{(\E Z_{R,\eta})^2}{\E Z_{R,\eta}^2}\notag\\
 &\ge
 \frac{c_1^2(M_Rt_R)^2}
 {C_1(M_Rt_R)^2+C_2(1+R)^{2s_F}M_Rt_R}
 \ge \frac{c_1^2}{2C_1}=:p_*>0
 \label{eq:ac-paley-zygmund}
\end{align}
for every selected $R\ge R_0$ and all sufficiently small $\eta>0$.

On a counted event, \Cref{lem:amplification}, \eref{eq:finite-env}, and
\eref{eq:extra} imply
\begin{equation*}
 \sqrt{\frac{W_o}{m_\eta}}
 \ge A_*\frac{u_R}{t_R}-\sqrt{C_0}
 =A_*e^{(\beta_2-\beta_1)R}-\sqrt{C_0},
 \qquad
 A_*:=\frac{\sqrt c}{2\sqrt2A_0}>0.
\end{equation*}
For all sufficiently large selected $R$,
\begin{equation*}
 A_*e^{(\beta_2-\beta_1)R}-\sqrt{C_0}
 \ge\frac{A_*}{2}e^{(\beta_2-\beta_1)R},
\end{equation*}
and therefore, with
\begin{equation*}
 \chi=2(\beta_2-\beta_1)>0,
 \qquad C_*=-2\log(A_*/2),
\end{equation*}
we have on $\{Z_{R,\eta}>0\}$
\begin{equation}\label{eq:ac-log-event}
 \log\frac{W_o}{m_\eta}\ge\chi R-C_*.
\end{equation}
The uniform estimate \eqref{eq:W-tight} gives
\begin{equation}\label{eq:ac-chebyshev-tail}
 \PP\left\{\log\frac{W_o}{m_\eta}\ge L\right\}
 \le\frac{\E\left|\log(W_o/m_\eta)\right|^2}{L^2}
 \le\frac{C_{\rm tail}}{L^2},\qquad L>0,
\end{equation}
with $C_{\rm tail}$ independent of $R$ and $\eta$. For each sufficiently
large $R$ in the selected sequence, \eqref{eq:ac-paley-zygmund} and
\eqref{eq:ac-log-event} hold for all sufficiently small $\eta$. Hence
\eqref{eq:ac-chebyshev-tail} gives
\[
 p_*\le\frac{C_{\rm tail}}{(\chi R-C_*)^2}.
\]
Letting $R\to\infty$ contradicts $p_*>0$.
Lemma~\ref{lem:boundary-dichotomy} now proves
\eqref{eq:selected-positive-boundary}.
\end{proof}

\subsubsection{Singular spectrum}
We condition on the exterior of one site. Its exceptional energy set is
then fixed during one-site averaging
\cite{SimonRankOne,SimonAveraging,Marx,LiawTreil}. This averaging and the
boundary theorem remove singular mass on $D$, as in
\cite[Proposition~2.2]{AWtree}.

\begin{lemma}\label{lem:no-singular}
We let $H=A+\diag(v_x)$ be selfadjoint on \eqref{eq:multiplication-domain}, with
independent real diagonal variables of common bounded density
$\rho_\lambda$. We let $D$ be a deterministic Borel set. Assume that,
for every vertex $x$ and almost every exterior configuration,
$\Sigma_x(E+i0)$ is finite for almost every $E\in D$, and that
\begin{equation*}
 \int\rho_\lambda(v)\mu_{x,v}(J)\,dv
 \le\|\rho_\lambda\|_\infty |J|
 \quad\text{for every Borel }J.
\end{equation*}
If, for every $x$ and almost every $E\in D$,
\[
 \PP\{W_x(E+i0)>0\}=1,
\]
then almost surely
\[
 \one_D(H)P_{\sing}(H)=0.
\]
\end{lemma}
\begin{proof}
We fix $x$ and condition on the exterior configuration.  At every energy
where the finite boundary value $\Sigma_x(E+i0)=S_x(E)+iY_x(E)$ exists,
$Y_x(E)\ge0$ and the one-site identity gives
\begin{equation}\label{eq:no-singular-one-site-boundary}
 G_{xx,v}(E+i0)=\frac1{v-S_x(E)-iY_x(E)},
 \qquad
 W_{x,v}(E)=\frac{Y_x(E)}{(v-S_x(E))^2+Y_x(E)^2}.
\end{equation}
Since the single-site law is absolutely continuous,
\begin{equation*}
 \int_\R\rho_\lambda(v)\,
 \one_{\{W_{x,v}(E)=0\}}\,dv
 =\one_{\{Y_x(E)=0\}}
\end{equation*}
for almost every such $E$: when $Y_x(E)=0$, the exceptional value
$v=S_x(E)$ has zero single-site probability, while for $Y_x(E)>0$ the
right side of \eqref{eq:no-singular-one-site-boundary} is strictly positive
for every $v\in\R$.

The hypothesis $\PP\{W_x(E+i0)>0\}=1$ and Tonelli's theorem therefore give
\begin{align}
 0
 &=\int_D\PP\{W_x(E+i0)=0\}\,dE\notag\\
 &\ge
 \E_{\ne x}\int_D
 \one_{\{\Sigma_x(E+i0)\text{ finite}\}}
 \one_{\{Y_x(E)=0\}}\,dE.
 \notag
\end{align}
By the exterior regularity hypothesis, for almost every exterior
configuration the set
\begin{equation*}
 N_x=
 \{E\in D:\Sigma_x(E+i0)\text{ is undefined or infinite}\}
 \cup\{E\in D:Y_x(E)=0\}
\end{equation*}
has Lebesgue measure zero.  The conditional spectral averaging bound then
implies
\begin{equation}\label{eq:no-singular-null-mass}
 \int_\R\rho_\lambda(v)\mu_{x,v}(N_x)\,dv
 \le\|\rho_\lambda\|_\infty|N_x|=0,
\end{equation}
so $\mu_{x,v}(N_x)=0$ for almost every $v$.

For $E\in D\setminus N_x$, \eqref{eq:no-singular-one-site-boundary} gives
\begin{equation*}
 0<W_{x,v}(E)
 =\frac{Y_x(E)}{(v-S_x(E))^2+Y_x(E)^2}
 \le\frac1{Y_x(E)}<\infty.
\end{equation*}
Hence
\begin{equation*}
 \left\{E\in D\setminus N_x:
 \lim_{\eta\downarrow0}\ImPart G_{xx,v}(E+i\eta)=+\infty\right\}
 =\varnothing.
\end{equation*}
The boundary-value characterization of the singular part
\cite[Theorem~11.6(ii)]{SimonTrace},
\cite[Theorem~2.2(ii)]{Marx}, together with
\eqref{eq:no-singular-null-mass}, yields
\begin{equation*}
 \mu_x^{\sing}(D)=0
 \qquad\text{almost surely.}
\end{equation*}
Finally,
\begin{equation*}
 \|\one_D(H)P_{\sing}(H)\delta_x\|^2
 =\langle\delta_x,\one_D(H)P_{\sing}(H)\delta_x\rangle
 =\mu_x^{\sing}(D)=0.
\end{equation*}
Intersecting the probability-one events over the countable vertex set and
using density of their span gives $\one_D(H)P_{\sing}(H)=0$.
\end{proof}

For the bounded model, the exterior regularity and conditional averaging required in \Cref{lem:no-singular} follow from \Cref{lem:regular-energies} and \eqref{eq:rank-one-average}.

\begin{proposition}\label{prop:selected-endpoints}
Assume the bounded model hypotheses, the Poincar\'e inequality
\eref{eq:Poincare}, and \eref{eq:selected-growth}--\eref{eq:selected-mixing}.
Suppose $\|H\|\le B$ almost surely for some $B>2$ and $I$ is Borel with
$\overline I\subset(-B,B)$, and assume $n_\lambda(E)>0$ for almost every
$E\in I$. If, for some $\kappa\ge0$,
\begin{equation}\label{eq:selected-multiplicity}
 \liminf_{R\to\infty}\frac1{RM_R}
 \sum_{x\in F_R}\log N_{ox}\ge\kappa,
\end{equation}
then
\begin{equation}\label{eq:selected-integral}
 \int\mathcal L_F(E)\,d\nu_B(E)\le\log(B/2)-\kappa.
\end{equation}
For
\begin{equation*}
 D_F=I\cap\Ereg\cap\{E:\mathcal L_F(E)<\gamma\},
\end{equation*}
one has
\begin{equation}\label{eq:selected-mass}
 \nu_B(D_F)\ge
 \nu_B(I)-\frac{\log(B/2)-\kappa}{\gamma}.
\end{equation}
Almost surely, simultaneously at every vertex,
\[
 \frac{d\mu_x^{\ac}}{dE}(E)>0\quad\text{for almost every }E\in D_F,
 \qquad
 \one_{D_F}(H)P_{\sing}(H)=0.
\]
Thus, there is a nonzero a.c. component on $I$ whenever
\begin{equation}\label{eq:selected-test}
 \gamma\nu_B(I)+\kappa>\log(B/2).
\end{equation}
\end{proposition}
\begin{proof}
For $x\in F_R$, \Cref{prop:log} with $d(o,x)=R$ gives
\[
 \int\ell_{ox}(E)\,d\nu_B(E)
 \le(R-1)\log(B/2)-\log N_{ox}+C_{B,Q}.
\]
Tonelli, averaging over $F_R$, Fatou's lemma along $R\in\mathscr R$, and
\eref{eq:selected-multiplicity} yield \eref{eq:selected-integral}.
Since $\mathcal L_F\ge0$ and $\Ereg^c$ is $\nu_B$-null,
\[
 \gamma\bigl(\nu_B(I)-\nu_B(D_F)\bigr)
 \le\int_I\mathcal L_F(E)\,d\nu_B(E),
\]
which proves \eref{eq:selected-mass}.

By \Cref{prop:selected-ac-point}, the boundary imaginary part is positive
almost surely for almost every $E\in D_F$. Transitivity, Fubini, and
countability give this simultaneously at every vertex, and the Stieltjes
boundary formula gives the stated local a.c. densities. Applying
\Cref{lem:no-singular} gives the vanishing singular projection.
If \eref{eq:selected-test} holds, \eref{eq:selected-mass} gives
$\nu_B(D_F)>0$; because $I\Subset(-B,B)$, this implies $|D_F|>0$, and
the positive local density gives a nonzero a.c. component on $I$.
\end{proof}

\subsubsection{Equilibrium-measure criterion}\label{sec:equilibrium-criterion}
We now replace the interval in \Cref{prop:selected-endpoints} by a finite
union of intervals containing $\sigma(H)$.

\begin{proposition}\label{prop:equilibrium-selected}
Assume the bounded model hypotheses, \eqref{eq:Poincare},
\eqref{eq:selected-growth}--\eqref{eq:selected-mixing}, and the
multiplicity bound \eqref{eq:selected-multiplicity}. We let $S$ be a
deterministic finite union of nondegenerate compact real intervals with
$\sigma(H)\subset S$ almost surely. Then
\begin{equation}\label{eq:equilibrium-integral}
 \int_S\mathcal L_F(E)\,d\omega_S(E)\le\log\lcap S-\kappa.
\end{equation}
If $I\subset S$ is a deterministic Borel set on which
$n_\lambda(E)>0$ for Lebesgue-almost every $E$, there is a deterministic
Borel set $D\subset I$ such that
\begin{equation}\label{eq:equilibrium-mass}
 \omega_S(D)\ge\omega_S(I)-\frac{\log\lcap S-\kappa}{\gamma}.
\end{equation}
Almost surely, simultaneously at every vertex,
\[
 \frac{d\mu_x^{\ac}}{dE}(E)>0\quad\text{for a.e. }E\in D,
 \qquad \one_D(H)P_{\sing}(H)=0.
\]
Thus, $D$ has positive Lebesgue measure whenever the
right-hand side of \eqref{eq:equilibrium-mass} is positive.
\end{proposition}
\begin{proof}
We put $R=d(x,y)$. Since the adjacency and potential are real, the
spectral representation of $G_{xy}$ is the Cauchy transform of a real
signed measure supported on $S$. Finite propagation and the unit
edge weights in \eqref{eq:model} give
\begin{equation}\label{eq:full-green-leading}
 G_{xy}(z)=-N_{xy}z^{-R-1}+O(z^{-R-2}).
\end{equation}
Indeed, all powers $H^j$ with $j<R$ have zero $(x,y)$ entry, and
$(H^R)_{xy}=N_{xy}$. We apply \Cref{lem:finite-band-jensen} to
\eqref{eq:full-green-leading} to obtain
\begin{equation}\label{eq:full-green-equilibrium}
 \int_S\log|G_{xy}(E+i0)|\,d\omega_S(E)
 \ge\log N_{xy}-(R+1)\log\lcap S.
\end{equation}

The a priori estimate \eqref{eq:apriori}, Fatou's lemma, and
$\log^+t\le C_\theta t^\theta$ give a constant $C_+<\infty$ such that
\begin{equation}\label{eq:equilibrium-positive-log}
 \E\log^+|G^C_{uv}(E+i0)|\le C_+
\end{equation}
for almost every $E$ in the bounded range $S$, uniformly in the
vertices and in the full or vertex-deleted graphs used here. The
constant may depend on the fixed law. Since $\omega_S\ll dE$, the
exceptional energy sets are also $\omega_S$-null.
We subtract \eqref{eq:full-green-equilibrium} from the integral of the
positive logarithm, then take expectation and use Tonelli's theorem:
\begin{equation}\label{eq:equilibrium-negative-log}
 \int_S\E\log^-|G_{xy}(E+i0)|\,d\omega_S(E)
 \le(R+1)\log\lcap S-\log N_{xy}+C_+.
\end{equation}

By \eqref{eq:one-site} and \eqref{eq:pair-identities},
\begin{equation*}
 G_{xy}=\tau_{xy}G_{xx}G^{(x)}_{yy},\qquad
 G^{(x)}_{yy}:=\ip{\delta_y}{(H_{X\setminus\{x\}}-z)^{-1}\delta_y}.
\end{equation*}
This is the two-site Grushin identity \eqref{eq:pair-identities} written
with the $x$-deleted one-site Green function.  Its logarithmic consequence,
together
with \eqref{eq:equilibrium-positive-log}, is
\begin{equation}\label{eq:equilibrium-tau-cost}
 \E\log^-|\tau_{xy}|
 \le\E\log^-|G_{xy}|+2C_+.
\end{equation}
We average \eqref{eq:equilibrium-negative-log} and
\eqref{eq:equilibrium-tau-cost} over $F_R$, divide by $R$, and apply
Fatou along the selected radii. The definition
\eqref{eq:selected-cost} and the multiplicity bound
\eqref{eq:selected-multiplicity} give \eqref{eq:equilibrium-integral},
since $(R+1)/R\to1$ and $3C_+/R\to0$.

We define $D=I\cap\Ereg\cap\{\mathcal L_F<\gamma\}$, removing any further
deterministic null set in \Cref{prop:selected-ac-point}.
Nonnegativity and \eqref{eq:equilibrium-integral} give
\[
 \gamma\bigl(\omega_S(I)-\omega_S(D)\bigr)
 \le\int_I\mathcal L_F\,d\omega_S
 \le\log\lcap S-\kappa,
\]
which proves \eqref{eq:equilibrium-mass}. The spectral conclusions follow
from \Cref{prop:selected-ac-point} and \Cref{lem:no-singular}, as in
\Cref{prop:selected-endpoints}. Finally,
\[
 \omega_S\ll dE,
\]
so positive equilibrium measure implies positive Lebesgue measure.
\end{proof}

\section{Geometry and bounded disorder}\label{sec:geometry-bounded}
We use the notation
\[
 \begin{gathered}
 K_\square=\sigma(A_{4,5}),\qquad R_d^\square=r_\square+2d, \qquad
 B_{d,\lambda}^\square=R_d^\square+\lambda, \qquad \text{ and }
 I_d^\square=[-R_d^\square,R_d^\square].
 \end{gathered}
\]
Thus $\sigma(A_d)=K_\square+[-2d,2d]$ and
$\|A_d\|=R_d^\square$.
For a finite vertex set $S$ of a graph $Y$, let
\[
 \partial_E S=\{\{x,y\}\in E(Y):x\in S,\ y\notin S\}.
\]
The graph is nonamenable if
$\inf_{0<|S|<\infty}|\partial_E S|/|S|>0$.
The geometric estimates below show that this is also true for
$\mathcal X_d$.

\subsection{The square lattice and its Euclidean products}\label{sec:geometry}

The graph $\G_{4,5}$ was defined in \Cref{sec:intro} and is illustrated in
\Cref{fig:square-lattice2}. Every vertex has five incident edges. For the
path counting below, we need to label these five edges in a way that is
preserved by translations of the graph.

We obtain such labels from the dual pentagon construction. Start with a
regular right-angled hyperbolic pentagon $P$ and reflect it successively
in its five sides. Since each interior angle of $P$ is $\pi/2$, exactly
four reflected pentagons meet at every corner. Put one vertex at the
center of each reflected pentagon and join two centers when the
pentagons share a side. Each pentagon has five sides, so every vertex of
the dual graph has degree five; each corner is surrounded by four
pentagons, so every face of the dual graph is a square. Hence the dual
$1$-skeleton is precisely $\G_{4,5}$.

Number the sides of $P$ cyclically by $1,\ldots,5$, and let $s_i$ be
reflection in side $i$. Reflections are involutions, and reflections in
two consecutive sides commute because those sides meet at a right angle.
Thus
\begin{equation}\label{eq:square-reflections}
 W_5=\langle s_1,\ldots,s_5\mid s_i^2=1,
       \ s_is_{i+1}=s_{i+1}s_i\ (i\bmod5)\rangle,
 \qquad S_\square=\{s_1,\ldots,s_5\}.
\end{equation}
We write $e$ for the identity of $W_5$ and $|w|$ for the word length with
respect to $S_\square$. Every reflected pentagon is $wP$ for a unique
$w\in W_5$. We identify its center with $w$, and we take the original
pentagon as the root $o=e$.
Crossing side $i$ of $wP$ leads to $ws_iP$; accordingly, the edge
\(
 \{w,ws_i\}
\)
has label $i$. This gives the same five labels at every vertex. Moreover,
left multiplication preserves them, since
\[
 g(ws_i)=(gw)s_i.
\]
Thus $W_5$ acts freely and transitively on the vertices by label-preserving
graph automorphisms.

The square faces are already visible in \eqref{eq:square-reflections}.
For each $i$ (with indices modulo $5$) and every $w\in W_5$,
\[
 w\longrightarrow ws_i\longrightarrow ws_is_{i+1}
   =ws_{i+1}s_i\longleftarrow ws_{i+1}\longleftarrow w
\]
is a four-cycle. These are the five squares incident to $w$. Filling all
such four-cycles gives the simply connected square complex associated
with this right-angled Coxeter group; see
\cite[Sections~5--6]{DavisOkun}. The local pattern at the root is shown in
\Cref{fig:square-edge-labels}.

\begin{figure}[htbp]
\centering
\begin{tikzpicture}[scale=.90, every node/.style={font=\scriptsize}]
  \node[draw,circle,inner sep=1.2pt] (e) at (0,0) {$e$};
  \node[draw,circle,inner sep=1.0pt] (s1) at (90:1.35) {$s_1$};
  \node[draw,circle,inner sep=1.0pt] (s2) at (18:1.35) {$s_2$};
  \node[draw,circle,inner sep=1.0pt] (s3) at (-54:1.35) {$s_3$};
  \node[draw,circle,inner sep=1.0pt] (s4) at (-126:1.35) {$s_4$};
  \node[draw,circle,inner sep=1.0pt] (s5) at (162:1.35) {$s_5$};

  \coordinate (p12) at (54:2.25);
  \coordinate (p23) at (-18:2.25);
  \coordinate (p34) at (-90:2.25);
  \coordinate (p45) at (-162:2.25);
  \coordinate (p51) at (126:2.25);

  \draw (e)--node[pos=.48,fill=white,inner sep=.7pt] {$1$} (s1);
  \draw (e)--node[pos=.48,fill=white,inner sep=.7pt] {$2$} (s2);
  \draw (e)--node[pos=.48,fill=white,inner sep=.7pt] {$3$} (s3);
  \draw (e)--node[pos=.48,fill=white,inner sep=.7pt] {$4$} (s4);
  \draw (e)--node[pos=.48,fill=white,inner sep=.7pt] {$5$} (s5);

  \draw (s1)--node[pos=.55,fill=white,inner sep=.7pt] {$2$} (p12);
  \draw (s2)--node[pos=.55,fill=white,inner sep=.7pt] {$1$} (p12);
  \draw (s2)--(p23)--(s3);
  \draw (s3)--(p34)--(s4);
  \draw (s4)--(p45)--(s5);
  \draw (s5)--(p51)--(s1);

  \fill (p12) circle (1.2pt) node[above right=1pt] {$s_1s_2$};
  \fill (p23) circle (1.2pt) node[right=2pt] {$s_2s_3$};
  \fill (p34) circle (1.2pt) node[below=2pt] {$s_3s_4$};
  \fill (p45) circle (1.2pt) node[left=2pt] {$s_4s_5$};
  \fill (p51) circle (1.2pt) node[above left=1pt] {$s_5s_1$};
\end{tikzpicture}
\caption{The five edge labels at the root. For each consecutive pair
$i,i+1$, the commuting relation $s_is_{i+1}=s_{i+1}s_i$ closes one
four-cycle. Thus exactly five square faces meet at $e$.}
\label{fig:square-edge-labels}
\end{figure}
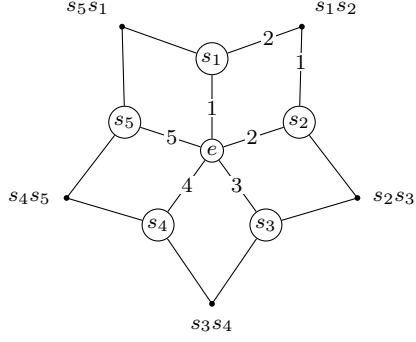

A word
\(
 s_{i_1}\cdots s_{i_n}=w
\)
records the path from $e$ obtained by following the edge labels
$i_1,\ldots,i_n$. Its word length is therefore exactly the graph distance,
\[
 |w|=d_{\G_{4,5}}(e,w).
\]
In particular, a word for $w$ is reduced precisely when its associated
path from $e$ to $w$ is shortest. The graph-product normal form
\cite[Definition~3.7 and Theorem~3.9]{Green} says that two reduced words
represent the same element precisely when one can be obtained from the
other by repeatedly interchanging neighboring letters that commute. Among
the other generators, $s_i$ commutes only with $s_{i-1}$ and $s_{i+1}$.
Consequently,
if a reduced word for $w$ is followed by $s_i$, then the new $s_i$
cancels an earlier $s_i$ exactly when that earlier occurrence can be
moved to the end of the word through letters commuting with $s_i$; the
final pair $s_is_i$ then disappears. These are the only word rules used
in the shortest-path arguments below.

\begin{lemma}\label{bare:lem:radius}
We put $A=A_{4,5}$, $K=\sigma(A)$, and $r=\|A\|$. The spectrum is symmetric and
\begin{equation}\label{bare:eq:radius}
 K=-K,\qquad 4\le r\le\sqrt{22}.
\end{equation}
\end{lemma}
\begin{proof}
All relations in \eqref{eq:square-reflections} have even length. Thus the
graph is bipartite: every edge joins an even-length word to an odd-length
word. The map $f(w)\mapsto(-1)^{|w|}f(w)$ conjugates $A$ to $-A$, proving
$K=-K$. Lifting walks to the universal covering tree injects closed
walks of the five-regular tree into closed walks of $X$. Taking large
even return-moment roots gives $r\ge2\sqrt{5-1}=4$; see also \cite[Chapter~II]{Woess}.

For the upper bound, the exact edge-isoperimetric formula of
\cite[Theorem~4.1]{HJL}, applied to vertex degree $5$ and face degree $4$,
gives
\begin{equation}\label{bare:eq:isoperimetry}
 \inf_{\varnothing\ne S\subset X\text{ finite}}
 \frac{|\partial_E S|}{|S|}
 =(5-2)\sqrt{1-\frac4{(5-2)(4-2)}}=\sqrt3.
\end{equation}
For nonnegative finitely supported $f$, the layer-cake formula and
\eqref{bare:eq:isoperimetry} give
\[
 \sqrt3\,\norm f_2^2\le\sum_{\{x,y\}\in E(X)}|f(x)^2-f(y)^2|.
\]
Applying Cauchy--Schwarz to the last sum yields
\begin{align*}
 3\norm f_2^4
 &\le \left(\sum_{\{x,y\}}(f(x)-f(y))^2\right)
       \left(\sum_{\{x,y\}}(f(x)+f(y))^2\right)=25\norm f_2^4-\ip f{Af}^{\,2}.
\end{align*}
The adjacency kernel is nonnegative, so
$|\ip f{Af}|\le\ip{|f|}{A|f|}$ for arbitrary $f$. Taking the supremum
over unit vectors proves $r\le\sqrt{22}$.
\end{proof}

\subsubsection{Translations of the random environment}\label{sec:group-ac}
We prove \eqref{eq:product-poincare} and \eqref{eq:product-mixing} by
comparing translations on the probability space with translations on the
group, then applying a norm bound for the hyperbolic factor.

We let $\Gamma_d=\mathbb Z^d\times W_5$, acting freely and transitively
on $\mathcal X_d$. The product root is $(0,o)$, also written $o$ locally. For finitely supported coefficients
$a:\Gamma_d\to\mathbb C$, we write
\begin{equation}\label{eq:translation-sums}
 U(a)=\sum_g a(g)U_g,\qquad L(a)=\sum_g a(g)L_g,
 \qquad (L_g\xi)(h)=\xi(g^{-1}h).
\end{equation}
The coordinate action $U$ in \eqref{eq:environment-action} is often called
the Bernoulli-Koopman representation; $L$ is the left regular representation.
Both act by the displayed translations. On $\ell^2(W_5)$, the unitary
$(\mathscr J\xi)(w)=\xi(w^{-1})$ identifies the right-labelled adjacency
with left convolution:
\[
 \mathscr J A_{4,5}\mathscr J^{-1}=\sum_{i=1}^5L_{s_i},
 \qquad \mathscr J\delta_o=\delta_o.
\]
The same identification holds on the hyperbolic factor of $\Gamma_d$.
The argument applies to the laws in \eqref{eq:model} and \eqref{bc:eq:model}.
\begin{lemma}\label{lem:Bernoulli}
We let $(\mathsf S,\mathscr S,\vartheta)$ be a probability space with
separable $L^2(\vartheta)$, let
$\Omega=\mathsf S^X$ carry the product law, and let $\Gamma_d$ act
through left translation on $X\simeq\Gamma_d$. For the operators in
\eref{eq:translation-sums} and every finitely supported
$a:\Gamma_d\to\C$,
\begin{equation}\label{eq:embedding-bound}
 \norm{U(a)}_{L^2_0(\Omega)\to L^2_0(\Omega)}
 \le\norm{L(a)}_{\ell^2(\Gamma_d)\to\ell^2(\Gamma_d)}.
\end{equation}
\end{lemma}
\begin{proof}
We write
\[
 L^2(\mathsf S,\vartheta)=\C\mathbf1\oplus\mathcal H_0.
\]
If $\mathcal H_0=\{0\}$, then $L^2_0(\Omega)=\{0\}$ and the desired estimate
holds trivially. In the remaining case, we choose an orthonormal basis $(e_j)$ of
$\mathcal H_0$. Finite tensor
products
\[
 \xi_{F,\alpha}=\prod_{x\in F}e_{\alpha(x)}(\omega_x),
 \qquad \varnothing\ne F\Subset\Gamma_d,
\]
form an orthonormal basis of $L^2_0(\Omega)$. The action $U_g$ permutes these
tensors. For one orbit, we let $H$ be the stabilizer of a tensor. Since $F$ is
finite and left translation on $\Gamma_d$ is free,
\[
 H\hookrightarrow\operatorname{Sym}(F),\qquad |H|<\infty.
\]
Thus the orbit representation is the quasi-regular representation on
$\ell^2(\Gamma_d/H)$. Torsion only makes $H$ possibly nontrivial. We define
\[
 J_H\delta_{gH}=|H|^{-1/2}\sum_{h\in H}\delta_{gh}
 \in\ell^2(\Gamma_d).
\]
The cosets are disjoint, so $J_H$ is an isometry, and for every $k\in\Gamma_d$,
\[
 L_kJ_H\delta_{gH}=J_H\delta_{kgH}.
\]
Hence each quasi-regular orbit is a subrepresentation of the left regular
representation and therefore
\[
 \|U(a)\|_{\text{orbit}}\le\|L(a)\|.
\]
Taking the orthogonal direct sum over tensor orbits proves
\eqref{eq:embedding-bound}. This explicit decomposition is the Bernoulli
Koopman representation underlying the general discussion in
\cite[Section~3]{SewardKoopman}.
\end{proof}

The Cayley graph of $W_5$ has uniformly thin geodesic triangles: each
side lies within a fixed distance of the other two. This is word
hyperbolicity \cite{BridsonHaefliger}. Its rapid-decay estimate
\cite{deLaHarpe,Jolissaint} states that, for some $C,s<\infty$,
\begin{equation*}
 \left\|\sum_{w\in W_5}a(w)L_w\right\|
 \le C(1+R)^s\|a\|_2
 \quad\text{if }\supp a\subset\{w:|w|\le R\}.
\end{equation*}
Here $\|a\|_2=(\sum_w|a(w)|^2)^{1/2}$ and $|w|=d(o,w)$.
Lemma~\ref{lem:Bernoulli} transfers this estimate to functions of the
random potential.
\begin{lemma}\label{lem:product-geometry}
For every finite $d\ge0$,
\begin{equation}\label{eq:product-radius}
 4\le r_\square\le\sqrt{22}<5,\qquad
 \|A_{\mathcal X_d}\|=r_\square+2d<5+2d.
\end{equation}
On any iid product probability space with separable single-site
$L^2$, for $T_H=\{(0,s_i):1\le i\le5\}$,
\begin{equation}\label{eq:product-poincare}
 \sum_{t\in T_H}\|f-U_tf\|_2^2
 \ge10\left(1-\frac{r_\square}{5}\right)\Var(f).
\end{equation}
We identify vertices with $\Gamma_d$ and write $U_x=U_g$ when $x=go$.
There are $C<\infty$ and $s<\infty$ such that, for every finite
$F\subset B_R((0,o))$ and every centered $f$ (that is, $\mathbb Ef=0$),
\begin{equation}\label{eq:product-mixing}
 \left\|\sum_{x\in F}U_xf\right\|_2
 \le C(1+R)^s(2R+1)^{d/2}\sqrt{|F|}\,\|f\|_2.
\end{equation}
Each translation in \eqref{eq:product-poincare} moves the root by one edge.
\end{lemma}
\begin{proof}
The radius statement is immediate from \eqref{bare:eq:radius} and the
tensor-sum formula:
\[
 \sigma(A_d)=K_\square+[-2d,2d],
 \qquad
 \|A_d\|=r_\square+2d<5+2d.
\]

For the Poincar\'e estimate, we apply \Cref{lem:Bernoulli} to
\[
 a=\frac15\sum_{t\in T_H}\delta_t.
\]
On \(L^2_0(\Omega)\),
\[
 \left\|\frac15\sum_{t\in T_H}U_t\right\|
 \le
 \left\|\frac15\sum_{t\in T_H}L_t\right\|
 =\frac{r_\square}{5}<1.
\]
Therefore, for \(f_0=f-\E f\),
\begin{align}
 \sum_{t\in T_H}\|f-U_tf\|_2^2
 &=10\left\langle f_0,
 \left(I-\frac15\sum_{t\in T_H}U_t\right)f_0\right\rangle\ge10\left(1-\frac{r_\square}{5}\right)\|f_0\|_2^2,
 \notag
\end{align}
which is \eqref{eq:product-poincare}.

For $z\in\mathbb Z^d$, we put $|z|_1=\sum_{i=1}^d|z_i|$.
For $F\subset B_R((0,o))$, we write
\[
 F_z=\{w\in W_5:(z,w)\in F\},
 \qquad
 F=\bigsqcup_{|z|_1\le R}\bigl(\{z\}\times F_z\bigr).
\]
Then rapid decay and Cauchy--Schwarz give
\begin{align}
 \left\|\sum_{(z,w)\in F}L_{(z,w)}\right\|
 &\le \sum_{|z|_1\le R}
       \left\|\sum_{w\in F_z}L_w\right\|\le C(1+R)^s\sum_{|z|_1\le R}|F_z|^{1/2}\notag\\
 &\le C(1+R)^s(2R+1)^{d/2}|F|^{1/2}.
 \notag
\end{align}
Applying \eqref{eq:embedding-bound} to the last convolution operator gives
\eqref{eq:product-mixing}.
\end{proof}

The square faces give more than one shortest path to many vertices.
We count all shortest paths first, and then select vertices carrying
comparable multiplicities. This construction serves both cases of
\Cref{thm:square-bounded}. We seek sets with the two asymptotic rates
\begin{equation}\label{eq:bare-endpoint-rates}
 \frac{\log|F_R|}{R}\longrightarrow\gamma>0,\qquad
 \liminf_{R\to\infty}\frac1{R|F_R|}
      \sum_{x\in F_R}\log N_{ox}\ge\kappa,
 \qquad R\in\mathscr R.
\end{equation}
\begin{lemma}\label{bare:lem:path-growth}
We put
\begin{equation}\label{bare:eq:entropies}
 h=\log\frac{3+\sqrt5}{2},\qquad
 \eta_{\rm geo}=\log(2+\sqrt2).
\end{equation}
Then the number $s_n$ of vertices at distance $n$ and the number $g_n$ of
shortest paths of length $n$, with all endpoints allowed, satisfy
\begin{equation}\label{bare:eq:path-growth}
 \lim_n\frac{\log s_n}{n}=h,
 \qquad g_n=\sum_{d(o,x)=n}N_{ox},\qquad
 \lim_n\frac{\log g_n}{n}=\eta_{\rm geo}.
\end{equation}
Moreover, $1\le N_{ox}\le2^n$ if $d(o,x)=n$.
\end{lemma}
\begin{proof}
The commutation graph has vertices $s_i$ and an edge for each commuting
pair; here it is a five-cycle. For $w\ne e$, we define its descent set
\[
 D(w)=\{s_i:|ws_i|=|w|-1\}.
\]
Each descent labels an edge to the previous sphere, hence a predecessor.
The elements of $D(w)$ commute, so
\[
 |D(w)|\in\{1,2\}.
\]
We put $\operatorname{Link}(s_i)=\{s_{i-1},s_{i+1}\}$, with cyclic indices.
For $s\notin D(w)$,
\begin{equation}\label{bare:eq:descent-update}
 D(ws)=\{s\}\cup\bigl(D(w)\cap\operatorname{Link}(s)\bigr).
\end{equation}
Thus the transition data are
\[
\begin{array}{c|cc|c}
 |D(w)|&\#\{|D(ws)|=1\}&\#\{|D(ws)|=2\}&\#\text{ predecessors}\\ \hline
 1&2&2&1\\
 2&1&2&2
\end{array}
\]
(the last column is the number of possible final letters of a reduced
word for the endpoint).

We let \(u_n,v_n\) count vertices in the two states.  Dividing incoming
edge counts by the number of predecessors gives
\begin{equation}\label{bare:eq:vertex-matrix}
 (u_{n+1},v_{n+1})=(u_n,v_n)
 \begin{pmatrix}2&1\\1&1\end{pmatrix},
 \qquad (u_1,v_1)=(5,0).
\end{equation}
If instead \(p_n,q_n\) count reduced words directly, each incoming edge
contributes one reduced word, and we obtain:
\begin{equation*}
 (p_{n+1},q_{n+1})=(p_n,q_n)
 \begin{pmatrix}2&2\\1&2\end{pmatrix},
 \qquad (p_1,q_1)=(5,0).
\end{equation*}
Hence
\[
 s_n=u_n+v_n,
 \qquad
 g_n=p_n+q_n=\sum_{d(o,x)=n}N_{ox},
\]
For a square matrix $M$, we write
$\operatorname{spr}(M)=\max\{|z|:z\in\sigma(M)\}$ for its spectral radius.
The Perron eigenvalue is the spectral radius of a positive matrix. Here
\[
 \operatorname{spr}\!\begin{pmatrix}2&1\\1&1\end{pmatrix}
   =\frac{3+\sqrt5}{2},
 \qquad
 \operatorname{spr}\!\begin{pmatrix}2&2\\1&2\end{pmatrix}
   =2+\sqrt2.
\]
Therefore
\[
 \lim_{n\to\infty}\frac1n\log s_n=h,
 \qquad
 \lim_{n\to\infty}\frac1n\log g_n=\eta_{\rm geo}.
\]
Finally each backward step has at most two choices, so recursively $1\le N_{ox}\le 2^{d(o,x)}.$
\end{proof}

\begin{lemma}\label{bare:lem:selection}
There are deterministic selected radii and endpoint sets satisfying
\eqref{eq:bare-endpoint-rates} with
\begin{equation}\label{bare:eq:selection}
 0<\eta_{\rm geo}-\log2\le\gamma\le h<1,
 \qquad \gamma+\kappa=\eta_{\rm geo}.
\end{equation}
Hence, for every $0\le m\le1$,
\begin{equation}\label{bare:eq:improved-criterion}
 \gamma m+\kappa\ge\eta_{\rm geo}-h(1-m).
\end{equation}
\end{lemma}
\begin{proof}
We partition the sphere of radius $n$ into sets
\[
 F_{n,j}=\{x:d(o,x)=n,\ 2^j\le N_{ox}<2^{j+1}\},
 \qquad 0\le j\le n.
\]
By \Cref{bare:lem:path-growth}, one bin $j_n$ carries at least $g_n/(n+1)$
of the total path count. We choose the least such maximizing index to fix
all choices deterministically, and put $F_n=F_{n,j_n}$. Then
\begin{equation}\label{bare:eq:bin-bounds}
 \frac{g_n}{2(n+1)}\le |F_n|2^{j_n}\le g_n.
\end{equation}
The two bounded sequences $n^{-1}\log|F_n|$ and $n^{-1}j_n\log2$
have a common convergent subsequence. We denote its limits by
$\gamma,\kappa$. Equation \eqref{bare:eq:bin-bounds} and
\eqref{bare:eq:path-growth} imply $\gamma+\kappa=\eta_{\rm geo}$.
Since $|F_n|\le s_n$ and $j_n\le n$, they also imply
$\gamma\le h$ and $\gamma\ge\eta_{\rm geo}-\log2>0$.
Every selected endpoint has $\log N_{ox}\ge j_n\log2$, proving
\eqref{eq:bare-endpoint-rates}. Finally,
$\gamma m+\kappa=\eta_{\rm geo}-\gamma(1-m)\ge\eta_{\rm geo}-h(1-m)$.
\end{proof}

\begin{lemma}\label{lem:square-free-spectrum}
For every finite $d\ge1$, the free adjacency on $\mathcal X_d$ satisfies
\begin{equation}\label{eq:square-free-spectrum}
 \sigma(A_{\mathcal X_d})=I_d^\square,
 \qquad \|A_{\mathcal X_d}\|=R_d^\square.
\end{equation}
Thus, under \eref{eq:density-upper}--\eref{eq:density-lower},
$n_\lambda(E)>0$ for almost every $E\in I_d^\square$ for every
$\lambda>0$, where $n_\lambda$ denotes the averaged local density for
\eref{eq:square-product-model}.
\end{lemma}
\begin{proof}
We let \(J_i\psi(w)=\psi(ws_i)\) and
\begin{equation}\label{eq:square-ladder-comparison}
 A_{\mathrm{lad}}=J_1+J_2+J_3,
 \qquad
 A_{4,5}=A_{\mathrm{lad}}+J_4+J_5,
 \qquad
 \|A_{4,5}-A_{\mathrm{lad}}\|\le2.
\end{equation}
The homomorphism fixing $s_1,s_2,s_3$ and sending $s_4,s_5$ to $e$
restricts to the identity on the subgroup they generate. Hence
\[
 \langle s_1,s_2,s_3\rangle
 \cong
 \langle s_1,s_3\mid s_1^2=s_3^2=e\rangle
 \times\langle s_2\mid s_2^2=e\rangle.
\]
We let $K_2$ be the graph of two vertices joined by one edge. Each left
coset carries the infinite ladder $\mathbb Z\square K_2$, and the restriction
of $A_{\mathrm{lad}}$ to that coset is unitarily equivalent to
\[
 A_{\mathbb Z}\otimes I+I\otimes A_{K_2}.
\]
Therefore
\begin{equation*}
 \sigma(A_{\mathrm{lad}})=[-2,2]+\{-1,1\}=[-3,3].
\end{equation*}

For selfadjoint \(S,T\),
\[
 \|S-T\|\le a,
 \quad \operatorname{dist}(E,\sigma(T))>a
 \Longrightarrow
 \|(T-E)^{-1}(S-T)\|<1
 \Longrightarrow E\notin\sigma(S).
\]
Applying this in both directions to \eqref{eq:square-ladder-comparison},
with \(\Sigma=\sigma(A_{4,5})\), yields
\begin{equation*}
 [-3,3]=\sigma(A_{\mathrm{lad}})\subset \Sigma+[-2,2].
\end{equation*}
Since \(\pm r_\square\in\Sigma\) and \(r_\square\le5\),
\[
 [-r_\square-2,-r_\square+2]\cup[-3,3]
 \cup[r_\square-2,r_\square+2]
 =[-r_\square-2,r_\square+2].
\]
Together with \(\Sigma\subset[-r_\square,r_\square]\), this gives
\begin{equation*}
 \Sigma+[-2,2]=[-r_\square-2,r_\square+2].
\end{equation*}
For \(d\ge1\), the tensor-sum formula now gives
\[
 \sigma(A_{\mathcal X_d})
 =\Sigma+[-2d,2d]
 =[-r_\square-2d,r_\square+2d],
\]
which is \eqref{eq:square-free-spectrum}.  The positivity of
\(n_\lambda\) then follows from \Cref{cor:free-criterion}.
\end{proof}

\subsection{Proof of the first theorem}\label{sec:bounded-applications}
We put $A=A_{4,5}$, $K=K_\square$ and $r=r_\square$. We bound the arcsine
mass of $K$ by proving that it has at most four components and that its
capacity is bounded below by squared shortest-path counts.

We let $\mu_0$ be the free local spectral probability measure and $F$ its
distribution function:
\[
 \mu_0(S)=\ip{\delta_o}{\one_S(A)\delta_o},\qquad
 F(t)=\mu_0(( -\infty,t]).
\]
Transitivity implies $\supp\mu_0=K$: if a spectral projection has zero
diagonal at $o$, it has zero diagonal everywhere and hence vanishes.

\begin{lemma}\label{bare:lem:labels}
For $t\in(-r,r)\setminus K$,
\begin{equation*}
 F(t)\in\{1/4,1/2,3/4\}.
\end{equation*}
Hence, $K$ is a union of at most four compact intervals,
allowing components to degenerate to points.
\end{lemma}
\begin{proof}
We use the right regular realization
$\mathcal A=\mathscr J^{-1}C_r^*(W_5)\mathscr J$ from
\Cref{sec:group-ac}, with canonical trace
$\tau(T)=\langle\delta_o,T\delta_o\rangle$. Since the largest clique of
the five-cycle commutation graph has two vertices,
\cite[Corollary~4.9]{RaumSkalski} gives
\[
 \tau(P)\in\tfrac14\Z\qquad(P=P^*=P^2\in\mathcal A).
\]
If $t\in(-r,r)\setminus K$, continuous functional calculus gives
$P_t=\one_{(-\infty,t]}(A)\in\mathcal A$, and
$F(t)=\tau(P_t)$.  Since $K=\supp\mu_0$ and $-r<t<r$, both $P_t$ and
$I-P_t$ are nonzero; faithfulness of $\tau$ gives $0<F(t)<1$.  Hence
$F(t)\in\{1/4,1/2,3/4\}$.  Distinct bounded gaps have distinct values of
$F$, so there are at most three bounded gaps and therefore at most four
spectral components.
\end{proof}

\begin{lemma}\label{lem:capacity-geodesics}
We let $A$ be the unit-weight adjacency of an infinite connected
vertex-transitive graph of bounded degree, and let $K=\sigma(A)$.
If $N_{ox}$ counts shortest paths, then
\begin{equation}\label{eq:capacity-geodesics}
 \lcap K\ge\limsup_{n\to\infty}
 \left(\sum_{d(o,x)=n}N_{ox}^2\right)^{1/(2n)}.
\end{equation}
\end{lemma}
\begin{proof}
We let $p_n$ be the monic orthogonal polynomial of degree $n$ for the free
local measure $\mu_0$.  We first note why such a polynomial exists in
every degree.  If
\[
 q(t)=t^n+\sum_{j=0}^{n-1}a_jt^j,
\]
then, for any vertex $x$ with $d(o,x)=n$, finite propagation gives
\[
 \ip{\delta_x}{A^j\delta_o}=0\quad(0\le j<n),
 \qquad
 \ip{\delta_x}{A^n\delta_o}=N_{ox}>0.
\]
Hence
\begin{equation}\label{eq:capacity-leading-row}
 \ip{\delta_x}{q(A)\delta_o}=N_{ox}.
\end{equation}
Thus every monic polynomial of degree $n$ has positive $L^2(\mu_0)$ norm,
so Gram--Schmidt produces $p_n$.

Applying \eqref{eq:capacity-leading-row} to $p_n$ and summing the squared
matrix elements over the distance sphere yields
\begin{align}
 \int_K|p_n(t)|^2\,d\mu_0(t)
 &=\|p_n(A)\delta_o\|_2^2 \ge\sum_{d(o,x)=n}
      |\ip{\delta_x}{p_n(A)\delta_o}|^2 =\sum_{d(o,x)=n}N_{ox}^2.
 \notag
\end{align}
We now let
\[
 t_n(K)=\inf\{\|q\|_{L^\infty(K)}:
                 q\text{ monic of degree }n\}.
\]
The polynomial $p_n$ minimizes the $L^2(\mu_0)$ norm among monic
degree-$n$ polynomials.  Therefore, for every monic $q$ of degree $n$,
\[
 \|p_n\|_{L^2(\mu_0)}
 \le\|q\|_{L^2(\mu_0)}
 \le\|q\|_{L^\infty(K)}\mu_0(K)^{1/2}
 =\|q\|_{L^\infty(K)}.
\]
Taking the infimum over $q$ gives
\begin{equation}\label{eq:capacity-chebyshev-bound}
 \left(\sum_{d(o,x)=n}N_{ox}^2\right)^{1/2}
 \le \|p_n\|_{L^2(\mu_0)}
 \le t_n(K).
\end{equation}
We raise \eqref{eq:capacity-chebyshev-bound} to the power $1/n$ and take
$\limsup$.  The Chebyshev characterization
\[
 \lim_{n\to\infty}t_n(K)^{1/n}=\lcap K
\]
from \cite[Appendix~B]{SimonCapacity} then gives exactly
\eqref{eq:capacity-geodesics}.
\end{proof}

\begin{lemma}\label{lem:square-capacity}
The free spectrum of $\G_{4,5}$ satisfies
\begin{equation}\label{eq:square-capacity}
 \lcap K\ge\sqrt{2+\sqrt[3]{2}+\sqrt[3]{4}}>\frac{11}{5}.
\end{equation}
\end{lemma}
\begin{proof}
We use the descent sets $D(w)$ from \Cref{bare:lem:path-growth} and write
$N_w=N_{ow}$. For $n\ge1$, we put
\[
 U_n=\sum_{\substack{|w|=n\\|D(w)|=1}}N_w^2,\qquad
 V_n=\sum_{\substack{|w|=n\\|D(w)|=2}}N_w^2.
\]
If $D(z)=\{s,t\}$, its two predecessors are $zs,zt$, and
$N_z=N_{zs}+N_{zt}$. We define
\begin{equation}\label{eq:square-mixed-count}
 C_n=\sum_{\substack{|z|=n+1\\D(z)=\{s,t\}}}N_{zs}N_{zt},
\end{equation}
where each vertex $z$ is counted once with its unordered descent pair.
By \eqref{bare:eq:descent-update}, a vertex with one descent has two
one-predecessor successors, while a vertex with two descents has one.
In both cases there are two successors with two predecessors. Summing the squares
and using \eqref{eq:square-mixed-count} gives
\begin{align*}
 U_{n+1}&=2U_n+V_n,\\
 V_{n+1}&=2U_n+2V_n+2C_n.
\end{align*}

For $|w|=n$ and a commuting pair
$\{s,t\}\subset S_\square\setminus D(w)$,
\eqref{bare:eq:descent-update} gives
\[
 |wst|=n+2,\qquad D(wst)=\{s,t\}.
\]
Indeed, consecutive vertices of the five-cycle have no common neighbor.
Conversely, if $|z|=n+2$ and $D(z)=\{s,t\}$, the normal form in
\Cref{sec:geometry} gives $w=zst$ with $|w|=n$ and
$s,t\notin D(w)$. Thus these pairs are in bijection with the
two-predecessor vertices at distance $n+2$, and
\[
 C_{n+1}=\sum_{|w|=n}
 \sum_{\substack{\{s,t\}\subset S_\square\setminus D(w)\\st=ts}}
 N_{ws}N_{wt}.
\]
For fixed $w$, write $a=N_w$. If $|D(w)|=1$, the commutation graph on
$S_\square\setminus D(w)$ is a path on four generators. By
\eqref{bare:eq:descent-update}, the successor path counts in this order are
\[
 a+a',\quad a,\quad a,\quad a+a'',
\]
where $a',a''$ count paths to the other predecessors of the two end
successors. The inner sum is therefore $3a^2+a(a'+a'')$.
If $|D(w)|=2$, the complement is a path on three generators, with counts
\[
 a+a',\quad a,\quad a+a'',
\]
and the inner sum is $2a^2+a(a'+a'')$.
Summing over $w$, each product in \eqref{eq:square-mixed-count} occurs
twice in the mixed terms, once for each predecessor. Hence
\[
 C_{n+1}=3U_n+2V_n+2C_n.
\]
Equivalently,
\begin{equation*}
 \begin{pmatrix}U_{n+1}\\V_{n+1}\\C_{n+1}\end{pmatrix}
 =M\begin{pmatrix}U_n\\V_n\\C_n\end{pmatrix},
 \qquad
 M=\begin{pmatrix}2&1&0\\2&2&2\\3&2&2\end{pmatrix},
 \qquad
 (U_1,V_1,C_1)=(5,0,5).
\end{equation*}
The first two iterates are
\begin{equation*}
 (U_2,V_2,C_2)=(10,20,25),\qquad
 (U_3,V_3,C_3)=(40,110,120),
\end{equation*}
and
\begin{equation*}
 M^2=
 \begin{pmatrix}6&4&2\\14&10&8\\16&11&8\end{pmatrix}>0.
\end{equation*}
Hence $M$ is primitive.  Perron--Frobenius gives constants
$c_-,c_+>0$ such that, for all sufficiently large $n$,
\begin{equation*}
 c_-\operatorname{spr}(M)^{n-1}
 \le U_n+V_n
 \le c_+\operatorname{spr}(M)^{n-1},
\end{equation*}
and therefore
\begin{equation*}
 \left(\sum_{d(o,x)=n}N_{ox}^2\right)^{1/n}
 =(U_n+V_n)^{1/n}\longrightarrow\operatorname{spr}(M).
\end{equation*}
Moreover
\[
 \det(zI-M)=z^3-6z^2+6z-2,
 \qquad
 \operatorname{spr}(M)=2+\sqrt[3]{2}+\sqrt[3]{4}.
\]
Hence \Cref{lem:capacity-geodesics} yields
\[
 \lcap K\ge\sqrt{\operatorname{spr}(M)}
 =\sqrt{2+\sqrt[3]{2}+\sqrt[3]{4}}.
\]
Finally,
\[
 M\begin{pmatrix}20\\57\\61\end{pmatrix}
 =\begin{pmatrix}97\\276\\296\end{pmatrix}
 >\frac{121}{25}\begin{pmatrix}20\\57\\61\end{pmatrix},
\]
so Perron comparison gives \(\operatorname{spr}(M)>121/25\), hence
\(\lcap K>11/5\).
\end{proof}

\begin{lemma}\label{lem:capacity-arcsine}
For $0<\lambda\le10^{-4}$ and $B=r+\lambda$,
\begin{equation}\label{eq:capacity-arcsine-mass}
 \nu_B(K)>\frac{153}{250}.
\end{equation}
\end{lemma}
\begin{proof}
Suppose first that $E\subset[-r,r]$ is a union of $k\ge2$
nondegenerate intervals and contains both endpoints.  The scaled form of
\cite[Theorem~3]{DubininKarp} gives
\begin{equation}\label{eq:capacity-upper-arcsine}
 \lcap E\le\frac r2
 \left(\sin\frac{\pi\nu_r(E)}2\right)^{1/(k-1)}.
\end{equation}
Since \Cref{bare:lem:labels} gives $k\le4$, and the sine lies in $[0,1]$,
we may replace the exponent in \eqref{eq:capacity-upper-arcsine} by $1/3$.
For $k=1$ the same bound is immediate, and degenerate components follow by
approximation; see \cite[Appendix~A]{SimonCapacity}.

We apply this to $E=K$.  By \eqref{eq:square-capacity} and
\eqref{bare:eq:radius},
\begin{equation}\label{eq:capacity-unshifted-mass}
 \nu_r(K)>
 \frac2\pi\arcsin\left((22/25)^{3/2}\right)>0.618.
\end{equation}
Indeed, \eqref{eq:capacity-upper-arcsine} implies
\[
 \sin\frac{\pi\nu_r(K)}2
 \ge \left(\frac{2\lcap K}{r}\right)^3
 >\left(\frac{22}{5\sqrt{22}}\right)^3
 =(22/25)^{3/2}.
\]

It remains to replace $\nu_r$ by $\nu_B$.  On $[-r,r]$ the density of
$\nu_B$ is bounded above by that of $\nu_r$, and hence
\begin{equation*}
 \nu_B(K)\ge\nu_r(K)-\bigl(1-\nu_B([-r,r])\bigr).
\end{equation*}
Moreover,
\begin{align*}
 1-\nu_B([-r,r])
 &=\frac2\pi\arccos\frac r{r+\lambda}
 =\frac2\pi\arctan\sqrt{2\lambda/r+(\lambda/r)^2}\\
 &\le\frac2\pi\sqrt{2\lambda/r+(\lambda/r)^2}<0.0046,
\end{align*}
where we used $r\ge4$ and $\lambda\le10^{-4}$.  Combining this with
\eqref{eq:capacity-unshifted-mass} gives
\[
 \nu_B(K)>0.618-0.0046=0.6134>\frac{153}{250}.
\]
\end{proof}

\begin{lemma}\label{lem:capacity-elementary}
With $h,\eta_{\rm geo}$ from \eqref{bare:eq:entropies},
$4\le r\le\sqrt{22}$, and $0<\lambda\le10^{-4}$,
\begin{equation}\label{eq:capacity-elementary}
 h<\frac{77}{80},\qquad
 \eta_{\rm geo}>\frac{491}{400},\qquad
 \log\frac{r+\lambda}{2}<\frac{341}{400}.
\end{equation}
\end{lemma}
\begin{proof}
These bounds follow directly from the explicit constants:
\[
 h=\log\frac{3+\sqrt5}{2}=0.96242\ldots<\frac{77}{80},
 \qquad
 \eta_{\rm geo}=\log(2+\sqrt2)=1.22794\ldots>\frac{491}{400}.
\]
Also $r\le\sqrt{22}$ and $\lambda\le10^{-4}$, so
\[
 \log\frac{r+\lambda}{2}
 \le\log\frac{\sqrt{22}+10^{-4}}2
 =0.85239\ldots<\frac{341}{400}.
\]
\end{proof}

The following proofs also give the quantitative bound
\begin{equation}\label{eq:bounded-energy-mass}
 \nu_{B_{d,\lambda}^\square}(D_{d,\lambda})>
 \begin{cases}
 31/20000,&d=0,\\
 1/100,&d\ge1,
 \end{cases}
\end{equation}
in the respective disorder ranges of \Cref{thm:square-bounded}.

\begin{proof}[Proof of \Cref{thm:square-bounded} for $d=0$]
We take $B=r+\lambda$ and $I=K$.  The bound $\|H\|\le B$ follows from
\eqref{eq:model} and \eqref{eq:density-upper}, and $K\Subset(-B,B)$ since
$\lambda>0$.  By \Cref{cor:free-criterion}, $n_\lambda>0$ almost
everywhere on $K$.

We use the selected endpoints from \Cref{bare:lem:selection} and the
translation estimates in \Cref{lem:product-geometry}.  Writing
$m=\nu_B(K)$, \eqref{eq:selected-mass} and
\eqref{bare:eq:improved-criterion} give
\begin{equation}\label{eq:capacity-bare-conclusion}
 \nu_B(D_{0,\lambda})
 \ge\frac{\gamma m+\kappa-\log(B/2)}{\gamma}
 \ge\frac{\eta_{\rm geo}-h(1-m)-\log(B/2)}{\gamma}.
\end{equation}
By \eqref{eq:capacity-arcsine-mass} and
\eqref{eq:capacity-elementary}, the numerator is bounded below by
\[
 \frac{491}{400}
 -\frac{77}{80}\frac{97}{250}
 -\frac{341}{400}
 =\frac{31}{20000}.
\]
Since $0<\gamma<1$ by \eqref{bare:eq:selection},
\eqref{eq:capacity-bare-conclusion} proves
\eqref{eq:bounded-energy-mass} for $d=0$.  The assertions about the local
a.c. density and the singular projection follow from
\Cref{prop:selected-endpoints}.
\end{proof}

For $d\ge1$, \Cref{lem:square-free-spectrum} fills the free gaps, while
interleaving hyperbolic and Euclidean steps gives additional shortest-path
multiplicity.

\begin{proof}[Proof of \Cref{thm:square-bounded} for $d\ge1$]
We fix $d\ge1$ and $0<\lambda\le1/5$, and put
\[
 C=r_\square+\lambda,\qquad B=C+2d,\qquad I=I_d^\square.
\]
Then $\|H\|\le B$ by \eqref{eq:model}, while
\Cref{lem:square-free-spectrum} and \Cref{cor:free-criterion} give
$n_\lambda>0$ almost everywhere on $I$.  Its arcsine mass is
\begin{equation}\label{eq:square-product-mass}
 \mu_d:=\nu_B(I)
 =\frac2\pi\arcsin\frac{r_\square+2d}{r_\square+2d+\lambda}
 >\frac56.
\end{equation}
Indeed, \eqref{bare:eq:radius} gives $r_\square+2d\ge6$, so the fraction
inside the inverse sine is at least $30/31$, and
$30/31>\sin(5\pi/12)$.

We let $F_n\subset S_n(o)$, $n\in\mathscr R$, be the sets from
\Cref{bare:lem:selection}, with rates $\gamma,\kappa$ satisfying
\eqref{bare:eq:selection}.  We choose
\begin{equation}\label{eq:square-product-optimum}
 q_0=\frac{C}{C+2d},\qquad q_i=\frac2{C+2d}\quad(1\le i\le d),
 \qquad \sum_{i=0}^dq_i=1.
\end{equation}
For $n\in\mathscr R$, choose $n_i=\lfloor q_i n/q_0\rfloor$, put
$R=n+\sum_{i=1}^dn_i$, and set
\[
 \mathcal F_R=\{(n_1,\ldots,n_d)\}\times F_n
 \subset S_R((0,o)).
\]
Since distance in the Cartesian product is additive,
$n/R\to q_0$ and $n_i/R\to q_i$.  Interleaving a hyperbolic geodesic
with the positive Euclidean coordinate steps gives
\begin{equation}\label{eq:square-product-rates}
 N_{(0,o),((n_1,\ldots,n_d),x)}
 \ge N_{ox}\frac{R!}{n!n_1!\cdots n_d!}.
\end{equation}
Different factor paths and different interleavings give distinct shortest
paths.  Stirling's formula, \eqref{eq:square-product-rates}, and
\eqref{eq:bare-endpoint-rates} therefore give
\begin{equation}\label{eq:square-product-entropy}
 \gamma_d=q_0\gamma,\qquad
 \kappa_d=q_0\kappa+\mathsf H(q),\qquad
 \mathsf H(q)=-\sum_{i=0}^dq_i\log q_i.
\end{equation}
The estimates in \Cref{lem:product-geometry} apply to these sets; for
fixed $d$, the Euclidean factor in \eqref{eq:product-mixing} is only
polynomial.  Thus the hypotheses of \Cref{prop:selected-endpoints} hold
along these radii.

For the proportions in \eqref{eq:square-product-optimum},
\begin{equation}\label{eq:square-product-cancellation}
 \mathsf H(q)=\log(B/2)-q_0\log(C/2).
\end{equation}
Substituting \eqref{eq:square-product-entropy} and
\eqref{eq:square-product-cancellation} into \eqref{eq:selected-mass}, and
then using \eqref{bare:eq:improved-criterion}, gives
\begin{equation}\label{eq:square-product-final-mass}
 \nu_B(D_{d,\lambda})
 \ge\frac{\gamma\mu_d+\kappa-\log(C/2)}{\gamma}
 \ge\frac{\eta_{\rm geo}-h(1-\mu_d)-\log(C/2)}{\gamma}.
\end{equation}
Now
\begin{equation*}
 \eta_{\rm geo}>\frac65,\qquad h<1,\qquad \log(C/2)<1.
\end{equation*}
The first two inequalities follow from the explicit values in
\eqref{bare:eq:entropies}.  For the last one,
$C/2\le(\sqrt{22}+1/5)/2<2.45<e$.  Together with
\eqref{eq:square-product-mass}, these bounds give
\[
 \eta_{\rm geo}-h(1-\mu_d)-\log(C/2)
 >\frac65-\frac16-1=\frac1{30}.
\]
Since $0<\gamma<1$, \eqref{eq:square-product-final-mass} is greater than
$1/30$, and hence in particular greater than $1/100$.  This proves
\eqref{eq:bounded-energy-mass} for $d\ge1$.  The density and singular
projection conclusions follow from \Cref{prop:selected-endpoints}, and
\Cref{lem:square-free-spectrum} gives $D_{d,\lambda}\subset\sigma(A_d)$.
\end{proof}

The full-support lower bound makes the averaged density positive on a
spectral enclosure for $H$.  We apply \Cref{prop:equilibrium-selected} to
that enclosure and use the selected-path rates.

We use the endpoint growth rate $\gamma$ from
\Cref{bare:lem:selection}.  The proof gives
\begin{equation}\label{eq:full-support-mass}
 \omega_{\Sigma_{d,\lambda}}(D_{d,\lambda})
 \ge\frac{\log(2+\sqrt2)-\log((r_\square+\lambda)/2)}{\gamma}>0.
\end{equation}
Here $\omega_{\Sigma_{d,\lambda}}$ is the equilibrium measure defined in
\Cref{sec:model}.

\begin{proof}[Proof of \Cref{thm:full-support-bounded}]
We put
\[
 K_d=\sigma(A_d),\qquad
 \Sigma=\Sigma_{d,\lambda}=K_d+[-\lambda,\lambda].
\]
If $\operatorname{dist}(E,K_d)>\lambda$, then
\[
 H_{\lambda,\omega}-E
 =(A_d-E)\left[I+(A_d-E)^{-1}V_\omega\right],
 \qquad
 \|(A_d-E)^{-1}V_\omega\|<1,
\]
so
\begin{equation*}
 \sigma(H_{\lambda,\omega})\subset\Sigma\qquad\text{a.s.}
\end{equation*}
Only gaps of $K_d$ longer than $2\lambda$ survive this thickening, so
$\Sigma$ is a finite union of compact intervals.  By
\eqref{eq:product-radius} and monotonicity of capacity,
\begin{equation*}
 \lcap\Sigma\le\frac{r_\square+\lambda+2d}{2}.
\end{equation*}
Because $b=1$, \Cref{cor:shifted-dos} gives
\[
 n_\lambda(E)>0
 \quad\text{for a.e. }E\in K_d+(-\lambda,\lambda).
\]
The finite complement of this open set in $\Sigma$ is
$\omega_\Sigma$-null, so \Cref{prop:equilibrium-selected} applies with
$I=S=\Sigma$.

For $d=0$, $\gamma+\kappa=\eta_{\rm geo}$, and therefore
\begin{align*}
 \omega_{\Sigma_{0,\lambda}}(D_{0,\lambda})
 &\ge 1-\frac{\log\lcap\Sigma_{0,\lambda}-\kappa}{\gamma}\\
 &\ge\frac{\gamma+\kappa-\log((r_\square+\lambda)/2)}{\gamma}
 =\frac{\eta_{\rm geo}-\log((r_\square+\lambda)/2)}{\gamma}.
\end{align*}
For $d\ge1$, set
\[
 C=r_\square+\lambda,\qquad B=C+2d,
 \qquad q_0=\frac CB,\qquad q_i=\frac2B.
\]
The product endpoint construction gives
\[
 \gamma_d=q_0\gamma,
 \qquad
 \kappa_d=q_0\kappa+\mathsf H(q),
 \qquad
 \mathsf H(q)=\log(B/2)-q_0\log(C/2).
\]
Hence
\begin{align*}
 \omega_{\Sigma_{d,\lambda}}(D_{d,\lambda})
 &\ge1-\frac{\log(B/2)-\kappa_d}{\gamma_d}\\
 &=\frac{\gamma+\kappa-\log(C/2)}{\gamma}
 =\frac{\eta_{\rm geo}-\log((r_\square+\lambda)/2)}{\gamma}.
\end{align*}
This is \eqref{eq:full-support-mass}.  Finally,
\[
 \eta_{\rm geo}-\log\frac{r_\square+\lambda}{2}
 =\log\frac{2(2+\sqrt2)}{r_\square+\lambda},
\]
which is positive exactly under \eqref{eq:full-support-window}.  The
spectral conclusions follow from \Cref{prop:equilibrium-selected}.
\end{proof}

\section{Stability theorem}\label{sec:interval-stability}
We fix a finite $d\ge0$. Throughout this section,
\[
 X:=V(\mathcal X_d),\qquad A=A_d,\qquad
 Q=5+2d,\qquad r=\|A\|=r_\square+2d.
\]
We identify $X$ with $\Gamma_d=\mathbb Z^d\times W_5$ as in
\eqref{eq:square-reflections}; products of vertices below are group products.
\subsection{Finite-ball tools for the interval theorem}\label{sec:bounded-cauchy}
For an interval result we need a bound uniform in energy. Cauchy averaging
on a finite ball gives an imaginary potential. A strict random-walk entropy
bound then yields purely a.c. spectrum. Section~\ref{sec:law-stability}
compares the chosen law with the Cauchy law.

\begin{lemma}\label{bc:lem:finite-jensen}
We let $\Lambda$ be finite.  We fix arbitrary bounded real potentials $w$ on
$\Lambda$ and $Z$ on $\Lambda^c$, and replace $w_x$ on $\Lambda$ by
$w_x+C_x$, where the $C_x$ are independent with density $c_\eps$.
We let $H^{\rm hyb}$ have these potentials on $\Lambda$ and $Z$ outside,
and write $G^{\rm hyb}_{ox}(z)=\langle\delta_o,(H^{\rm hyb}-z)^{-1}\delta_x\rangle$.
We let $\mathbb E_{\Lambda,C}$ integrate just the $C_x$, $x\in\Lambda$.
For $z=E+i\eta$, $\eta>0$, we put
\[
 R_\Lambda^w(z)=
 \bigl(A+\diag(w\one_\Lambda+Z\one_{\Lambda^c})-i\eps P_\Lambda-z\bigr)^{-1}.
\]
Then
\begin{equation}\label{bc:eq:finite-jensen}
 \E_{\Lambda,C}\log|G^{\rm hyb}_{ox}(z)|
 \ge\log|R_\Lambda^w(z)_{ox}|.
\end{equation}
The variables $w$ and $Z$ are fixed during this integration.
\end{lemma}
\begin{proof}
We let $\mathcal J\subset\Lambda$ be the sites already replaced by
$-i\eps$, and leave the next variable $v$ at $j\notin\mathcal J$ free.
With all remaining real variables fixed, we write
$K_{\mathcal J}(v)=H_0+vP_{\{j\}}-i\eps P_{\mathcal J}$, where $H_0$
is selfadjoint. For $\ImPart v\le0$ and $z=E+i\eta$,
\[
 \ImPart\langle h,(K_{\mathcal J}(v)-z)h\rangle
 =-\eps\|P_{\mathcal J}h\|^2+(\ImPart v)|h(j)|^2-\eta\|h\|^2
 \le-\eta\|h\|^2.
\]
The same argument for the adjoint proves invertibility, with
\[
 \|(K_{\mathcal J}(v)-z)^{-1}\|\le\eta^{-1}.
\]
Thus $f(v)=\langle\delta_o,(K_{\mathcal J}(v)-z)^{-1}\delta_x\rangle$
is analytic in the lower half-plane. The Cayley map satisfies
\[
 v(\zeta)=-i\eps\frac{1+\zeta}{1-\zeta},\qquad
 v(0)=-i\eps,\qquad v_*(dm)=c_\eps(v)\,dv,
 \qquad f\circ v\in H^\infty(\D).
\]
Lemma~\ref{lem:boundary-jensen} gives
\begin{equation}\label{bc:eq:one-site-jensen-step}
 \int_\R\log|f(v)|c_\eps(v)\,dv\ge\log|f(-i\eps)|.
\end{equation}
A zero center value gives the trivial lower bound $-\infty$.
At every step $\log|f|\le\log(\eta^{-1})$, so Tonelli applied to
$\log(\eta^{-1})-\log|f|\ge0$ permits successive extended integrals.
Choose an ordering $\Lambda=\{j_1,\ldots,j_N\}$ and let $G^{(k)}_{ox}$
be the resolvent entry after the first $k$ variables have been replaced by
$-i\eps$, while the remaining $N-k$ variables are still real Cauchy
variables.  Applying \eqref{bc:eq:one-site-jensen-step} conditionally gives
\begin{equation*}
 \E_{C_{j_{k+1}}}\log|G^{(k)}_{ox}(z)|
 \ge\log|G^{(k+1)}_{ox}(z)|,
 \qquad 0\le k<N.
\end{equation*}
Successive integration yields the monotone chain
\begin{align}
 \E_{C_{j_1},\ldots,C_{j_N}}\log|G^{(0)}_{ox}(z)|
 &\ge \E_{C_{j_2},\ldots,C_{j_N}}\log|G^{(1)}_{ox}(z)| \ge\cdots\ge\log|G^{(N)}_{ox}(z)|
 =\log|R_\Lambda^w(z)_{ox}|,
 \notag
\end{align}
which is \eqref{bc:eq:finite-jensen}.
\end{proof}

\subsubsection{Entropy and finite-block criteria}
Simple random walk chooses each of the $Q$ neighbors with probability
$1/Q$ at each step. For its $n$-step endpoint law and entropy, we set
$0\log0=0$ and define
\begin{equation}\label{bc:eq:walks}
 p_n(o,x)=Q^{-n}(A^n)_{ox},\qquad
 H_n=-\sum_xp_n(o,x)\log p_n(o,x),\qquad
 h_{\rm rw}=\lim_{n\to\infty}\frac{H_n}{n}=\inf_{n\ge1}\frac{H_n}{n}.
\end{equation}
For a vertex set $F$, we write $p_n(o,F)=\sum_{x\in F}p_n(o,x)$.
The entropy limit follows from $H_{m+n}\le H_m+H_n$. Counting one walk
and using $\|A/Q\|=r/Q$ give
\begin{equation}\label{bc:eq:prob-bounds}
 Q^{-n}\le p_n(o,x)\le(r/Q)^n\quad\text{on its support},
 \qquad h_{\rm rw}\ge\log(Q/r)>0.
\end{equation}

\begin{proposition}\label{bc:prop:block}
We let $H$ be a covariant iid bounded-disorder model on $\mathcal X_d$
whose common single-site law has a bounded density, and let $I$ be a
nondegenerate compact interval on which its averaged local
spectral density is positive almost everywhere. We put
$T_E(x,y)=\mathbb E\log^-|G_{xy}(E+i0)|$ for $x\ne y$ and $T_E(x,x)=0$.
Assume that, for constants $C_{\rm gl},C_+<\infty$,
\begin{align}
 T_E(x,y)&\le T_E(x,z)+T_E(z,y)+C_{\rm gl},\label{bc:eq:generic-glue}\\
 \E\log^+|G^C_{xy}(E+i0)|&\le C_+\label{bc:eq:generic-plus}
\end{align}
for almost every $E\in I$, all vertices, and the full and one-vertex
deleted graphs used below.
Suppose that for some odd $m$ there are deterministic nonnegative numbers
$b_x$, $x\in S:=\supp p_m(o,\cdot)$, such that
\begin{equation*}
 T_E(o,x)+C_{\rm gl}\le b_x\quad \text{ for }x\in S,\qquad
 \sum_{x\in S}p_m(o,x)b_x<m h_{\rm rw}
\end{equation*}
for almost every $E\in I$.  Then, almost surely,
\[
 I\subset\sigma_{\ac}(H),\qquad \one_I(H)P_{\sing}(H)=0,
\]
and the local a.c. density is positive for almost every $E\in I$,
simultaneously at every vertex.
\end{proposition}
\begin{proof}
We put
\[
 \bar b=\sum_{x\in S}p_m(o,x)b_x<mh_{\rm rw}
\]
and choose $\alpha>0$ so that
\begin{equation*}
 \bar b\le m(h_{\rm rw}-2\alpha).
\end{equation*}
We define the cost of reaching $y$ by block increments as
\[
 L_0(y)=\inf\left\{\sum_{j=1}^k b_{x_j}:
          k\ge1,\ x_j\in S,\ x_1\cdots x_k=y\right\},\qquad L_0(o)=0.
\]  If $y=x_1\cdots x_k$ with $x_j\in S$, covariance and
\eqref{bc:eq:generic-glue} give inductively
\begin{align}
 T_E(o,y)
 &\le\sum_{j=1}^kT_E(o,x_j)+(k-1)C_{\rm gl}\notag\\
 &\le\sum_{j=1}^k(b_{x_j}-C_{\rm gl})+(k-1)C_{\rm gl}
 =\sum_{j=1}^kb_{x_j}-C_{\rm gl}
 \le\sum_{j=1}^kb_{x_j}.
 \notag
\end{align}
Taking the infimum over all block representations yields
\begin{equation}\label{eq:block-concatenated-cost}
 T_E(o,y)\le L_0(y)
 \qquad\text{for a.e. }E\in I.
\end{equation}

We take $n=km$.  We let $X_1,\dots,X_k$ be independent block increments with
common law $p_m(o,\cdot)$.  Their product has law $p_n(o,\cdot)$, and the
particular block representation gives
$L_0(X_1\cdots X_k)\le\sum_j b_{X_j}$.  Hence
\begin{equation*}
 \sum_y p_n(o,y)L_0(y)
 \le k\bar b\le n(h_{\rm rw}-2\alpha).
\end{equation*}
Since $H_n\ge nh_{\rm rw}$,
\begin{equation}\label{bc:eq:entropy-surplus}
 \sum_y p_n(o,y)
 \bigl[-\log p_n(o,y)-L_0(y)\bigr]
 \ge2\alpha n.
\end{equation}
We define
\[
 \mathcal G_n:=
 \{y\in\supp p_n(o,\cdot):L_0(y)\le-\log p_n(o,y)-\alpha n\}.
\]
Since $L_0\ge0$ and $-\log p_n(o,y)\le n\log Q$, the surplus
\eqref{bc:eq:entropy-surplus} gives, with $q_n=p_n(o,\mathcal G_n)$,
\[
 2\alpha n\le\alpha n(1-q_n)+n\log Q\,q_n.
\]
Thus
\begin{equation*}
 p_n(o,\mathcal G_n)\ge c_0:=\frac{\alpha}{\log Q-\alpha}>0.
\end{equation*}
We put $\Delta=\alpha/8$. We partition $\mathcal G_n$ into bins of width
$\Delta n$ according to $-\log p_n(o,y)$. By \eqref{bc:eq:prob-bounds},
there are at most $1+\lceil\log Q/\Delta\rceil$ bins. A bin $\mathcal H_n$ of
maximal probability therefore satisfies, for a fixed $c>0$,
\[
 p_n(o,\mathcal H_n)\ge c,\qquad
 a_n n\le-\log p_n(o,y)<(a_n+\Delta)n\quad(y\in \mathcal H_n).
\]
We choose $F_n=\mathcal H_n\cap S_{R_n}(o)$ of maximal cardinality. We break ties by
the smallest bin index and then the smallest radius, so all choices
are deterministic. Since $\mathcal H_n\subset B_n(o)$, \eqref{bc:eq:prob-bounds} gives
\begin{equation}\label{eq:block-bin-cardinality}
 |F_n|\ge\frac{|\mathcal H_n|}{n+1}\ge\frac{ce^{a_n n}}{n+1},\qquad
 |F_n|\ge\frac{c(Q/r)^n}{n+1}.
\end{equation}
In particular, $|B_R(o)|\le C Q^R$ and \eqref{eq:block-bin-cardinality} imply
\[
 \frac{\log(Q/r)}{\log Q}\,n-O(\log n)\le R_n\le n.
\]
Membership in $\mathcal G_n$ and \eqref{eq:block-concatenated-cost} give
\[
 \operatorname*{ess\,sup}_{E\in I}
 \frac1{|F_n|}\sum_{y\in F_n}T_E(o,y)\le(a_n+\Delta-\alpha)n.
\]
Subtracting this from the logarithm of \eqref{eq:block-bin-cardinality},
\[
 \log|F_n|-\operatorname*{ess\,sup}_{E\in I}
 \frac1{|F_n|}\sum_{y\in F_n}T_E(o,y)
 \ge\frac{7\alpha n}{8}-O(\log n).
\]
The sequence $R_n^{-1}\log|F_n|$ is bounded above and away from zero.
We pass to a subsequence on which $R_n$ strictly increases and this sequence
converges to $\gamma>0$. Dividing the last inequality by $R_n$ proves,
for example with $c_1=\alpha/2$,
\begin{equation*}
 \frac{\log|F_n|}{R_n}\longrightarrow\gamma,\qquad
 \limsup_n\operatorname*{ess\,sup}_{E\in I}
 \frac1{R_n|F_n|}\sum_{y\in F_n}T_E(o,y)\le\gamma-c_1.
\end{equation*}

The two-site Grushin identity
$G_{oy}=\tau_{oy}G_{oo}G^{(o)}_{yy}$ and
\eqref{bc:eq:generic-plus} give
\begin{equation*}
 \E\log^-|\tau_{oy}|
 \le T_E(o,y)+2C_+.
\end{equation*}
After averaging over $F_n$ and dividing by $R_n$,
\begin{align}
 \mathcal L_{F,R_n}(E)
 &\le\frac1{R_n|F_n|}\sum_{y\in F_n}T_E(o,y)+\frac{2C_+}{R_n},
 \notag\\
 \limsup_n\operatorname*{ess\,sup}_{E\in I}\mathcal L_{F,R_n}(E)
 &\le\gamma-c_1<\gamma.
 \label{eq:block-strict-selected-cost}
\end{align}
The bounded single-site law can be written in the form
\eqref{eq:model}--\eqref{eq:density-upper} after rescaling its compact
support.  Moreover \eqref{eq:product-mixing} implies
\begin{equation*}
 C(1+R)^s(2R+1)^{d/2}
 \le C'(1+R)^{s+d/2},
\end{equation*}
so \eqref{eq:selected-mixing} holds for the deterministic sets $F_n$.
Together with \eqref{eq:product-poincare},
\eqref{eq:block-strict-selected-cost} is the strict hypothesis of
\Cref{prop:selected-ac-point}.  Thus every local a.c. density is positive
for almost every $E\in I$, and \Cref{lem:no-singular} removes singular
spectrum.  Since the sets $F_n$ are deterministic, the conclusion holds
simultaneously at all vertices.
\end{proof}

\subsection{Perturbative stability}\label{sec:law-stability}
We now allow bounded independent noise to be added to the
Cauchy variable. If $\nu$ is the probability distribution of
that noise, the sum has density
\[
 \psi_{\eps,\nu}(t)\coloneq(c_\eps*\nu)(t)
 =\int_{\mathbb R}c_\eps(t-y)\,\dd\nu(y).
\]
This is convolution of a density with a probability measure.
The two geometric margins below bound the size and mean of
the added noise. The resulting estimates are uniform over
all noise distributions satisfying these bounds.

\begin{theorem}\label{thm:law-stability}
Fix a finite integer $d\ge0$ and $J=[a,b]$ with
$r_\square+2d<a<b<5+2d$.  We put $r=r_\square+2d$ and $Q=5+2d$.
Choose $u,v\ge0$ such that
\begin{equation*}
 r+u<a,\qquad b+v<Q.
\end{equation*}
For every $K\ge2$ there is $\eps_0>0$ such that, for each
$0<\eps\le\eps_0$, there is $t_*>0$ with the following property.
Let $\nu$ be a probability measure satisfying
\begin{equation}\label{st:eq:intro-noise}
 \supp\nu\subset[-u,u],\qquad
 \left|\int y\,\dd\nu(y)\right|\le v,
\end{equation}
and let $f$ be a compactly supported probability density satisfying
\begin{equation*}
 0\le f\le K\psi_{\eps,\nu}\quad\text{a.e.},\qquad
 d_{\rm TV}(f,\psi_{\eps,\nu})<t_*.
\end{equation*}
Assume also that, for some $\ell>b-r$ and $c_f>0$,
\begin{equation}\label{st:eq:intro-floor}
 f(t)\ge c_f\quad\text{for a.e. }t\in[-\ell,\ell].
\end{equation}
For iid potentials with density $f$, the operator
$H_f=A_d+\diag(V_x)$ then satisfies, almost surely,
\begin{equation*}
 \begin{gathered}
 J\cup(-J)\subset\sigma_{\ac}(H_f),\qquad
 \one_{J\cup(-J)}(H_f)P_{\sing}(H_f)=0,\\
 \frac{d\mu_{x,f}^{\ac}}{dE}(E)>0
 \quad\text{for a.e. }E\in J\cup(-J),\quad\text{simultaneously for all }x.
 \end{gathered}
\end{equation*}
The constants can be chosen uniformly over all $\nu$ satisfying
\eqref{st:eq:intro-noise}, including asymmetric choices of $f$ and $\nu$.
\end{theorem}

We prove the block inequality in \Cref{bc:prop:block} uniformly on $J$,
using logarithmic bounds and the imaginary potential from
\Cref{bc:lem:finite-jensen}.

We fix $J=[a,b]$, $u,v,K,\nu,f$ as in \Cref{thm:law-stability}.
The density $c_\eps$ is defined in \Cref{sec:further-results},
and $\psi_{\eps,\nu}$ is defined above. We write $H_f=A+\diag(V_x)$ for iid density $f$,
$G^f_{xy}(z)=\langle\delta_x,(H_f-z)^{-1}\delta_y\rangle$, and
\[
 T^f_{E,\eta}(x,y)=\mathbb E\log^-|G^f_{xy}(E+i\eta)|,\qquad
 T^f_E(x,y)=\mathbb E\log^-|G^f_{xy}(E+i0)|\quad(x\ne y),
\]
with both costs set to zero for $x=y$. A hybrid law replaces $f$ by
$\psi_{\eps,\nu}$ at a finite set of sites, with all coordinates still
independent. The superscript $\mathrm{hyb}$ denotes this operator and
its Green function. Each expectation uses the stated product law.

If $|y|\le u$, then
\[
 \frac{c_\eps(t-y)}{c_\eps(t)}
 \le 2+2u^2/\eps^2=:D_{\eps,u}.
\]
Indeed, $t^2\le2(t-y)^2+2y^2$.  Hence every density $\rho$ equal either
to $f$ or to $\psi_{\eps,\nu}$ obeys
\begin{equation}\label{st:eq:uniform-controls}
 \rho\le KD_{\eps,u}c_\eps,\qquad
 \|\rho\|_\infty\le\frac K{\pi\eps},
\end{equation}
and, for $0<s<1$,
\[
 \int |t|^s\rho(t)\,dt
 \le K\left(\int|z|^s c_\eps(z)\,dz+u^s\right).
\]
For every fixed $0<s<1$,
$\log^2(1+|t|)\le C_s(1+|t|^s)$.  Thus the same bounds give
\begin{equation}\label{st:eq:uniform-log-moment}
 \sup_\rho\int_\R\log^2(1+|t|)\rho(t)\,dt<\infty,
\end{equation}
where the supremum is over the single-site laws $f$ and
$\psi_{\eps,\nu}$ occurring below.

\begin{lemma}\label{st:lem:logs}
For fixed $\eps,u,K,b,Q$ and finite $m$, there is $L\ge1$, independent
of $\nu$, the support of $f$, and the number of sites whose law is
changed from $f$ to $\psi_{\eps,\nu}$, such that
\begin{equation}\label{st:eq:log-two}
 \sup_{\substack{|E|\le b,\ 0<\eta\le1\\1\le d(x,y)\le m}}
 \|\log|G_{xy}(E+i\eta)|\|_{L^2}\le L.
\end{equation}
One may take
\[
 F_{\eps,K}=8\sqrt{\frac{2K}{\pi\eps}},\qquad
 C_+(\eps,K)=2\log(1+F_{\eps,K}),\qquad
 C_{\rm gl}(\eps,K)=1+\log^+\frac{8K}{\pi\eps}.
\]
For every retained graph $C$,
$\E\log^+|G^C_{xy}|\le C_+(\eps,K)$.  For the iid law $f$, we define
$T_E(x,x)=0$ and
$T_E(x,y)=\E\log^-|G_{xy}(E+i0)|$ for $x\ne y$.  At almost every energy,
\begin{equation}\label{st:eq:glue}
 T_E(x,y)\le T_E(x,z)+T_E(z,y)+C_{\rm gl}(\eps,K).
\end{equation}
The expected logarithmic parts converge to their boundary values at
almost every energy.
\end{lemma}
\begin{proof}
We put $M_\eps=\frac{K}{\pi\eps}.$
All single-site laws appearing in the finite replacements have density at
most $M_\eps$. For every probability density $\rho$ with this bound,
\cite[equation~(7)]{Tautenhahn} gives
\[
 \int_\R|v-\zeta|^{-1/2}\rho(v)\,dv\le4\sqrt{M_\eps},
 \qquad \zeta\in\C.
\]
The conditional one- and two-site estimates in
\cite[equation~(6) and Lemma~3.1]{Tautenhahn}, used in
\eqref{eq:apriori}, require only independence and these individual
integration bounds. They therefore apply to the hybrid laws as well,
including the Cauchy-convolution coordinates, and give
\begin{equation}\label{st:eq:fractional-half}
 \sup_{C,x,y,|E|\le b,\ 0<\eta\le1}
 \E|G^C_{xy}(E+i\eta)|^{1/2}
 \le F_{\eps,K}:=8\sqrt{\frac{2K}{\pi\eps}}.
\end{equation}
Hence
\begin{align}
 \E\log^+|G^C_{xy}|
 &\le2\E\log(1+|G^C_{xy}|^{1/2})
 \le2\log(1+F_{\eps,K})=:C_+(\eps,K),
 \notag\\
 \E(\log^+|G^C_{xy}|)^2
 &\le C\bigl(1+\E|G^C_{xy}|^{1/2}\bigr)
 \le C_{\eps,K}.
 \notag
\end{align}

For the negative logarithm, we first let \(x\sim y\), retain
\[
 F=\{x\}\cup N(x),\qquad |F|\le Q+1,
\]
and condition on \(F^c\).  The Grushin problem with auxiliary space
$\ell^2(F)$ gives
\begin{equation*}
 E_{F,-+}=\mathcal D_F(E+i\eta)-\diag(V_F),
 \qquad
 \mathcal M(V_F):=-E_{F,-+}=\diag(V_F)-\mathcal D_F(E+i\eta).
\end{equation*}
We set $q=|F|$ and
$P_{xy}(V_F)=\langle\delta_x,\adj(\mathcal M(V_F))\delta_y\rangle$.
Since $G_F=\mathcal M(V_F)^{-1}$,
\[
 G_{xy}=\frac{P_{xy}(V_F)}{\det\mathcal M(V_F)}.
\]
Both polynomials are multiaffine. We write $[\mathfrak m]P$ for the
coefficient of a monomial $\mathfrak m$ in $P$. Since $A_{xF^c}=0$,
\eqref{eq:grushin-effective} leaves the $x$ row of the effective matrix
unchanged. Its cofactor expansion gives
\[
 \left[\prod_{v\in F\setminus\{x,y\}}V_v\right]P_{xy}
 =\pm A_{xy}=\pm1.
\]
We set $V_j=\eps U_j$, let $c_*$ be the largest coefficient modulus of
$P_{xy}(\eps U)$, and put
\[
 c_*\ge\eps^{q-2},\qquad P(U)=P_{xy}(\eps U)/c_*.
\]
By \eqref{st:eq:uniform-controls}, the scaled variables have densities
bounded by $D c_1$, where $D=KD_{\eps,u}$. We first prove the required
one-variable estimate. If $U$ has this density bound,
$M=\max\{|a|,|b|\}>0$, and $0<s\le1/4$, then
\[
 \begin{aligned}
 |a|\ge\sqrt{s}M
 &\ \Longrightarrow\
 \PP\{|aU+b|<sM\}\le\frac{2D}{\pi}\frac{sM}{|a|}
 \le\frac{2D}{\pi}\sqrt{s},\\
 |a|<\sqrt{s}M
 &\ \Longrightarrow\
 \{|aU+b|<sM\}\subseteq\{|U|>(1-s)/\sqrt{s}\}.
 \end{aligned}
\]
The second case uses $|b|=M$; the Cauchy tail satisfies
$\PP\{|U|>R\}\le2D/(\pi R)$. Increasing the constant for $s>1/4$ gives
\begin{equation}\label{st:eq:affine-cauchy}
 \PP\{|aU+b|<s\max(|a|,|b|)\}\le C_D\sqrt{s},\qquad 0<s\le1.
\end{equation}
We let $F_q(t)$ be the supremum of $\PP\{|P(U)|<e^{-t}\}$ over normalized
multiaffine polynomials in $q$ variables and independent laws bounded
by $D c_1$. We write
\[
 P(U)=A(U')U_q+B(U'),\qquad M(U')=\max\{|A(U')|,|B(U')|\}.
\]
One of $A,B$ has largest coefficient modulus one. Splitting at
$M(U')=e^{-t/2}$ and applying \eqref{st:eq:affine-cauchy} conditionally yields
\[
 F_q(t)\le F_{q-1}(t/2)+C_De^{-t/4} \text{ with }q\ge2,\qquad
 F_1(t)\le C_De^{-t/2}.
\]
Induction therefore proves
\begin{equation}\label{st:eq:poly-small-value}
 \PP\{|P(U)|<e^{-t}\}\le C_{q,D}e^{-t/2^q},\qquad t\ge0.
\end{equation}
By \eqref{st:eq:poly-small-value} and the layer-cake formula,
\begin{equation}\label{st:eq:poly-log-square}
 \E(\log^-|P|)^2
 =2\int_0^\infty t\,\PP\{|P|<e^{-t}\}\,dt
 \le2C_{q,D}\int_0^\infty te^{-t/2^q}\,dt<\infty.
\end{equation}
All constants depend only on $\eps,u,K,Q$.

For the denominator, Hadamard's inequality gives
\begin{equation}\label{st:eq:det-hadamard}
 |\det\mathcal M(V_F)|
 \le(1+\|\mathcal D_F\|)^{|F|}
      \prod_{j=1}^{|F|}(1+|V_j|).
\end{equation}
Apart from the bounded terms $zI_F-A_F$, every entry of $\mathcal D_F$
is a sum of at most $Q^2$ exterior resolvent entries. Thus
\eqref{st:eq:fractional-half} implies
\begin{equation}\label{st:eq:exterior-log-square}
 \E\log^2(1+\|\mathcal D_F\|)\le C.
\end{equation}
Together with \eqref{st:eq:uniform-log-moment} and
\eqref{st:eq:det-hadamard}, this gives a uniform $L^2$ bound for
$\log^+|\det\mathcal M(V_F)|$. Since
\[
 \log^-|G_{xy}|
 \le\log^-|P(U)|+\log^-c_*+\log^+|\det\mathcal M(V_F)|,
\]
we obtain
\begin{equation}\label{st:eq:edge-log-square}
 \sup_{x\sim y,\ |E|\le b,\ 0<\eta\le1}
 \E|\log|G_{xy}(E+i\eta)||^2\le C.
\end{equation}
For $1<d(x,y)=\ell\le m$, we choose a shortest path
$\gamma=(x_0=x,\ldots,x_\ell=y)$ and put
\[
 F=\bigcup_{j=0}^{\ell-1}(\{x_j\}\cup N(x_j)),\qquad
 |F|\le\ell(Q+1),\qquad \{x_0,\ldots,x_\ell\}\subset F.
\]
Every nonterminal path vertex has all its neighbors in $F$, so
\eqref{eq:grushin-effective} has zero exterior correction in those rows
and columns. A shortest path is chordless. Selecting the diagonal
factors off $\gamma$ leaves the endpoint cofactor of this chain, whose
successive rows force the consecutive path edges. Thus
\[
 \left[\prod_{v\in F\setminus\{x_0,\ldots,x_\ell\}}V_v\right]P_{xy}
 =\pm\prod_{j=0}^{\ell-1}A_{x_jx_{j+1}}=\pm1,\qquad
 c_*\ge\eps^{|F|-\ell-1}.
\]
We apply \eqref{st:eq:poly-log-square}--\eqref{st:eq:exterior-log-square}
with at most $m(Q+1)$ variables. Together with
\eqref{st:eq:edge-log-square}, this gives
\begin{equation}\label{st:eq:path-log-square}
 \sup_{\substack{|E|\le b,\ 0<\eta\le1\\1\le d(x,y)\le m}}
 \E|\log|G_{xy}(E+i\eta)||^2\le L^2,
\end{equation}
which is \eqref{st:eq:log-two}.

For gluing, we fix distinct vertices $x,y,z$ and put $\zeta=E+i\eta$.
We condition outside $z$ and set
\[
 R^z=(H_{X\setminus\{z\}}-\zeta)^{-1},\quad \Sigma=\Sigma_z(\zeta),\quad
 \alpha_j=\langle\delta_j,R^zb_z\rangle\quad(j=x,y).
\]
The one-site Grushin equations give
\[
 G_{jz}=-\frac{\alpha_j}{V_z-\Sigma},\qquad
 G_{xy}=G^{(z)}_{xy}+\frac{\alpha_x\alpha_y}{V_z-\Sigma}.
\]
For this fixed triple we choose
\[
 m_{x,y,z}\ge\max\{d(x,z),d(y,z)\}.
\]
Equation~\eqref{st:eq:path-log-square}, with $m=m_{x,y,z}$, makes
$G_{xz}$ and $G_{yz}$ nonzero almost surely for every fixed $\eta>0$.
Hence $\alpha_x\alpha_y\ne0$ almost surely.  The constant
$C_{\rm gl}(\eps,K)$ obtained below is independent of $m_{x,y,z}$. We put
$k=G^{(z)}_{xy}/(\alpha_x\alpha_y)$; then
\begin{equation*}
 G_{xy}=G_{xz}G_{zy}\,p(V_z),
 \qquad
 p(t)=(t-\Sigma)\bigl(1+k(t-\Sigma)\bigr),
\end{equation*}
whence
\begin{equation}\label{st:eq:glue-log}
 \log^-|G_{xy}|
 \le\log^-|G_{xz}|+\log^-|G_{zy}|+\log^-|p(V_z)|.
\end{equation}
If \(k\ne0\) and \(r_1,r_2\) are the roots of \(p\), then
\[
 p(t)=k(t-r_1)(t-r_2),
 \qquad |k|\,|r_1-r_2|=1,
\]
and therefore
\[
 \max_j|t-r_j|\ge\tfrac12|r_1-r_2|,\qquad
 |p(t)|\ge\tfrac12\min_j|t-r_j|.
\]
Thus $\{|p(t)|<s\}\cap\R$ lies in the union of the two real slices
$\{|t-r_j|<2s\}\cap\R$, each of length at most $4s$. For $k=0$,
$p(t)=t-\Sigma$ gives the same bound:
\begin{equation}\label{st:eq:quadratic-sublevel}
 |\{t\in\R:|p(t)|<s\}|\le8s.
\end{equation}
Combining \eqref{st:eq:quadratic-sublevel} with
$\|f\|_\infty\le K/(\pi\eps)$ gives
\begin{equation}\label{st:eq:p-small-prob}
 \PP\{|p(V_z)|<s\mid V_{X\setminus\{z\}}\}
 \le\min\left\{1,\frac{8K}{\pi\eps}s\right\}.
\end{equation}
For $A>0$, the elementary integral is
\[
 \int_0^\infty\min\{1,Ae^{-t}\}\,dt
 =\begin{cases}A,&A\le1,\\1+\log A,&A\ge1,\end{cases}
 \qquad\le1+\log^+A.
\]
Thus \eqref{st:eq:p-small-prob} gives
\begin{align}
 \E[\log^-|p(V_z)|\mid V_{X\setminus\{z\}}]
 &=\int_0^\infty\PP\{|p(V_z)|<e^{-t}\mid\cdot\}\,dt\le1+\log^+\frac{8K}{\pi\eps}
 =:C_{\rm gl}(\eps,K).
 \notag
\end{align}
Taking expectations gives the gluing inequality at $E+i\eta$.
With the diagonal cost defined to be zero, it also holds when vertices coincide.

For each fixed hybrid law, scalar Cauchy-transform boundary values
exist a.e. \cite[Chapter~II]{Duren}. Fubini and countability of the
vertex pairs and triples give a deterministic full-measure energy set. At such an
energy, \eqref{st:eq:path-log-square} and Fatou imply
\[
 \PP\{G_{xy}(E+i0)=0\}=0,\qquad
 \E|\log|G_{xy}(E+i0)||^2\le L^2.
\]
Indeed, on a zero boundary value the negative logarithm diverges, which
would contradict the uniform second-moment bound. We put
$Z_\eta=\log|G_{xy}(E+i\eta)|$ and $Z_0=\log|G_{xy}(E+i0)|$. Then
\[
 Z_\eta\longrightarrow Z_0\quad\text{a.s.},\qquad
 \sup_{0<\eta\le1}\E\bigl(|Z_\eta|\one_{\{|Z_\eta|>T\}}\bigr)
 \le\frac{L^2}{T}\longrightarrow0.
\]
Uniform integrability therefore gives
\begin{equation*}
 \E\log^\pm|G_{xy}(E+i\eta)|
 \longrightarrow\E\log^\pm|G_{xy}(E+i0)|.
\end{equation*}
Taking boundary limits in \eqref{st:eq:glue-log} after expectation proves
\eqref{st:eq:glue}; the other bounds pass to the boundary in the same way.
\end{proof}

For a real sequence $w$, we write $|w|\le u$ when $\sup_x|w_x|\le u$, and
put $A_w=A+\diag(w)$.  Then $\|A_w\|\le r+u<a$.  We write
\[
 G_w^0(z)=(A_w-z)^{-1},\qquad F_w(E)=-G_w^0(E)_{ox}>0\quad(E>r+u).
\]
We set $B_w=A+\diag(w+u)$.  It has nonnegative entries and
$\|B_w\|\le r+2u$, so
\begin{equation}\label{st:eq:positive-series}
 F_w(E)=\sum_{k\ge0}(B_w^k)_{ox}(E+u)^{-k-1}.
\end{equation}
Thus $F_w(E)\ge(E+u)^{-d(o,x)-1}$.

\begin{lemma}\label{st:lem:phase}
There is $D<\infty$, depending only on $d,J,u$, such that for
$1\le d(o,x)\le m$, $E\in J$, $t>0$, and every $|w|\le u$,
\begin{equation}\label{st:eq:phase}
 |G_w^0(E+it)_{ox}-G_w^0(E)_{ox}|
 \le Dt(m+1)F_w(E).
\end{equation}
\end{lemma}
\begin{proof}
We fix a number $t_0$ with $r+u<t_0<a$
and put $c_0=(t_0-r-u)^{-1}$.  Since $\|A_w\|\le r+u$, the resolvent
bound gives
\begin{equation}\label{st:eq:phase-resolvent-bound}
 0<F_w(t_0)\le c_0.
\end{equation}
On the other hand, if $q=d(o,x)$, then the first nonzero term in the
positive series \eqref{st:eq:positive-series} occurs by order
$q$, and its coefficient is at least one.  Hence
\begin{equation*}
 F_w(E)\ge(E+u)^{-q-1}.
\end{equation*}

For fixed $E$ we normalize the positive terms in
\eqref{st:eq:positive-series} by setting
\begin{equation}\label{st:eq:phase-probabilities}
 \pi_k=\frac{(B_w^k)_{ox}(E+u)^{-k-1}}{F_w(E)},
 \qquad \pi_k\ge0,
 \qquad \sum_{k\ge0}\pi_k=1.
\end{equation}
We write
\[
 \rho_E=\frac{E+u}{t_0+u}>1.
\]
Then \eqref{st:eq:phase-resolvent-bound}--\eqref{st:eq:phase-probabilities}
give
\begin{align}
 \sum_{k\ge0}\pi_k\rho_E^{k+1}
 &=\frac{F_w(t_0)}{F_w(E)}\le c_0(E+u)^{q+1}.
 \label{st:eq:phase-exponential-moment}
\end{align}
Because $E\in[a,b]$, both $\rho_E$ and $E+u$ stay in compact intervals
bounded away from one and zero, respectively.  Taking logarithms and
using Jensen's inequality in \eqref{st:eq:phase-exponential-moment},
\[
 \rho_E^{\sum_k\pi_k(k+1)}
 \le\sum_k\pi_k\rho_E^{k+1},
\]
we obtain
\begin{equation}\label{st:eq:phase-first-moment}
 \sum_{k\ge0}\pi_k(k+1)
 \le C_1(q+1)\le C_1(m+1),
\end{equation}
where $C_1$ depends only on $d,J,u$.

We now use the same Neumann series at the complex point $E+it$:
\[
 -G_w^0(E+it)_{ox}
 =\sum_{k\ge0}(B_w^k)_{ox}(E+u+it)^{-k-1}.
\]
Subtracting the series at $t=0$ and factoring out $F_w(E)$ gives
\begin{align*}
 \frac{|G_w^0(E+it)_{ox}-G_w^0(E)_{ox}|}{F_w(E)}
 &\le\sum_k\pi_k
   \left|\left(1+\frac{it}{E+u}\right)^{-k-1}-1\right|\le\frac{t}{E+u}\sum_k\pi_k(k+1).
\end{align*}
The last inequality follows by integrating the derivative of
$(1+is)^{-k-1}$ from $0$ to $t/(E+u)$.  Since $E+u\ge a-u>0$,
\eqref{st:eq:phase-first-moment} proves \eqref{st:eq:phase} with
$D=C_1/(a-u)$.
\end{proof}

We let $R\ge2$, $\Lambda=B_R(o)$ and $D_0=B_{R-1}(o)$. We define the inner
vertex boundary by
$\partial D_0=\{x\in D_0:N(x)\cap D_0^c\ne\varnothing\}$.
We keep $w$ on $\Lambda$ and
an arbitrary bounded real potential $Z$ on $\Lambda^c$.  We let
$R_\Lambda^w(z)$ denote the resolvent of
\[
 A+\diag(w\one_\Lambda+Z\one_{\Lambda^c})-i\eps P_\Lambda.
\]
We extend $w$ by zero outside $\Lambda$ when defining $A_w$.

\begin{lemma}\label{st:lem:shield}
We set
\begin{equation}\label{st:eq:shield-def}
 S_R=\frac Q{a-r-u}
 \left(\frac1\eps+\frac1{a-r-u}\right)
 \left(\frac{r+u}{a}\right)^{R-1}.
\end{equation}
For $x\in D_0$, $E\ge a$ and $\eta>0$,
\[
 |R_\Lambda^w(E+i\eta)_{ox}-G_w^0(E+i(\eta+\eps))_{ox}|\le S_R.
\]
This bound is uniform in $w$ and $Z$.
\end{lemma}
\begin{proof}
We let
\[
 \mathcal H_\Lambda
 =A+\diag(w\one_\Lambda+Z\one_{\Lambda^c})-i\eps P_\Lambda.
\]
The resolvent in the statement is $(\mathcal H_\Lambda-E-i\eta)^{-1}$.
If
$f$ is supported in $\Lambda$ and
$h=R_\Lambda^w(E+i\eta)f$, then
\[
 (\mathcal H_\Lambda-E-i\eta)h=f.
\]
Taking the imaginary part of the scalar product with $h$ gives
\begin{equation}\label{st:eq:shield-dissipation}
 \eta\|h\|^2+\eps\|P_\Lambda h\|^2
 =\ImPart\ip f h.
\end{equation}
Since $f=P_\Lambda f$, Cauchy--Schwarz in
\eqref{st:eq:shield-dissipation} yields
\begin{equation}\label{st:eq:shield-compression}
 \|P_\Lambda R_\Lambda^w(E+i\eta)P_\Lambda\|
 \le(\eps+\eta)^{-1}\le\eps^{-1}.
\end{equation}
This estimate is independent of the exterior potential $Z$.

We next quantify propagation from the root to the inner boundary.  We put
\[
 T=(A_w)_{D_0},\qquad \|T\|\le r+u<a.
\]
For $E\ge a$, the Dirichlet resolvent has the convergent Neumann series
\[
 (T-E-i(\eta+\eps))^{-1}
 =-\sum_{k\ge0}
   T^k(E+i(\eta+\eps))^{-k-1}.
\]
Every vertex in $\partial D_0$ is at graph distance at least $R-1$ from
$o$, so
\[
 P_{\partial D_0}T^k\delta_o=0
 \text{ for }0\le k<R-1.
\]
Hence
\begin{align}
 \|P_{\partial D_0}(T-E-i(\eta+\eps))^{-1}\delta_o\|
 &\le\sum_{k\ge R-1}\frac{\|T\|^k}{E^{k+1}}\le\frac1{a-r-u}
      \left(\frac{r+u}{a}\right)^{R-1}.
 \label{st:eq:shield-boundary-row}
\end{align}

We let $C=A_{D_0,D_0^c}:\ell^2(D_0^c)\to\ell^2(D_0)$ be the boundary
adjacency block.  Since the graph has degree $Q$,
$\|C\|\le Q$.  We apply the geometric resolvent identity to compare the
pinned resolvent with the Dirichlet resolvent on $D_0$:
\begin{align}
 &P_{D_0}R_\Lambda^w(E+i\eta)P_{D_0}
 -(T-E-i(\eta+\eps))^{-1}\notag\\
 &\qquad=
 -(T-E-i(\eta+\eps))^{-1}
 C\,P_{X\setminus D_0}R_\Lambda^w(E+i\eta)P_{D_0}.
 \label{st:eq:shield-geometric-resolvent}
\end{align}
The outer endpoint of every edge represented by $C$ still lies in the
shell $\Lambda\setminus D_0$.  Hence the relevant second resolvent factor
is controlled by \eqref{st:eq:shield-compression}.  Evaluating
\eqref{st:eq:shield-geometric-resolvent} between $\delta_o$ and
$\delta_x$, $x\in D_0$, and using
\eqref{st:eq:shield-boundary-row} gives the error
\begin{equation}\label{st:eq:shield-first-error}
 \frac Q{\eps(a-r-u)}
 \left(\frac{r+u}{a}\right)^{R-1}.
\end{equation}

We finally compare the Dirichlet reference resolvent on $D_0$ with the
whole-space reference $G_w^0(E+i(\eta+\eps))$.  The same geometric
resolvent identity applies, but now both resolvents are bounded by
$(a-r-u)^{-1}$.  This contributes at most
\begin{equation}\label{st:eq:shield-second-error}
 \frac Q{(a-r-u)^2}
 \left(\frac{r+u}{a}\right)^{R-1}.
\end{equation}
Adding \eqref{st:eq:shield-first-error} and
\eqref{st:eq:shield-second-error} gives exactly
\eqref{st:eq:shield-def}.
\end{proof}

We put
\[
 p_m(o,x)=Q^{-m}(A^m)_{ox},\qquad
 H_m=-\sum_xp_m(o,x)\log p_m(o,x),\qquad
 h_{\rm rw}=\lim_{m\to\infty}H_m/m.
\]
The proof of \eqref{bc:eq:walks} gives
$h_{\rm rw}=\inf_mH_m/m$, and \eqref{bc:eq:prob-bounds} gives
$Q^{-m}\le p_m(o,x)\le(r/Q)^m$ on its support.
We set
\[
 \beta=b+v<Q,
 \qquad \delta=\log(Q/\beta)>0.
\]

\begin{lemma}\label{st:lem:averaged-walk}
We let $w_x$, $x\in B_m(o)$, be iid with law $\nu$, and let $\mathbb E_\nu$
average these variables. We extend $w$ outside the ball with $|w|\le u$.  If $p_m(o,x)>0$ and $E\in J$,
then
\begin{equation}\label{st:eq:averaged-walk}
 \E_\nu\log F_w(E)
 \ge\log(A^m)_{ox}-(m+1)\log\beta
 =\log p_m(o,x)+m\delta-\log\beta.
\end{equation}
\end{lemma}
\begin{proof}
We put $D_E=\diag(E-w_x)$ and $K_E=D_E^{-1/2}AD_E^{-1/2}$. Since
$E>r+u$ and $|w|\le u$,
\[
 \|K_E\|\le\frac r{E-u}<1,\qquad
 E-A_w=D_E^{1/2}(I-K_E)D_E^{1/2}.
\]
The Neumann expansion therefore gives
\[
 (E-A_w)^{-1}=D_E^{-1/2}\sum_{k\ge0}K_E^kD_E^{-1/2},\qquad
 F_w(E)=\sum_{k\ge0}\sum_{\substack{\gamma:o\to x\\|\gamma|=k}}
             \prod_{j=0}^k(E-w_{\gamma_j})^{-1}.
\]
Here $\gamma$ runs over nearest-neighbor walks. Each summand is
nonnegative; retaining $k=m$ yields
\begin{equation}\label{st:eq:walk-fixed-length}
 F_w(E)\ge
 \sum_{\substack{\gamma_0=o,\ \gamma_m=x\\
                   \gamma_j\sim\gamma_{j+1}\ (0\le j<m)}}
 \prod_{j=0}^m(E-w_{\gamma_j})^{-1}.
\end{equation}
The number of summands is precisely $(A^m)_{ox}$.  We denote these positive
summands by $a_1,\dots,a_N$, $N=(A^m)_{ox}$.  The arithmetic--geometric
mean inequality gives
\[
 \log\left(\sum_{\ell=1}^N a_\ell\right)
 \ge \log N+\frac1N\sum_{\ell=1}^N\log a_\ell.
\]
Applying this to \eqref{st:eq:walk-fixed-length}, then averaging over the
iid background variables $w$, yields
\begin{align}
 \E_\nu\log F_w(E)
 &\ge \log(A^m)_{ox}
 -\frac1{(A^m)_{ox}}
   \sum_{\gamma:o\to x,\ |\gamma|=m}
   \sum_{j=0}^m\E_\nu\log(E-w_{\gamma_j})\notag\\
 &=\log(A^m)_{ox}
 -(m+1)\int\log(E-y)\,d\nu(y).
 \label{st:eq:walk-average-log}
\end{align}
Each marginal has law $\nu$; linearity of expectation counts repeated
visits with their multiplicity.

Since $E\ge a>r+u$ and $|y|\le u$, the function
$y\mapsto\log(E-y)$ is concave on the support of $\nu$.  Jensen's
inequality and \eqref{st:eq:intro-noise} imply
\begin{align}
 \int\log(E-y)\,d\nu(y)
 &\le \log\left(E-\int y\,d\nu(y)\right)\le\log(E+v)\le\log(b+v)=\log\beta.
 \label{st:eq:walk-jensen}
\end{align}
We insert \eqref{st:eq:walk-jensen} into
\eqref{st:eq:walk-average-log}.  Since
$(A^m)_{ox}=Q^mp_m(o,x)$ and $\delta=\log(Q/\beta)$,
\begin{align*}
 \E_\nu\log F_w(E)
 &\ge \log p_m(o,x)+m\log Q-(m+1)\log\beta\\
 &=\log p_m(o,x)+m\delta-\log\beta,
\end{align*}
which is \eqref{st:eq:averaged-walk}.  Replacing $\nu$ by its reflection
changes the sign of its mean but preserves the bound
$|\int y\,d\nu(y)|\le v$, so the same proof applies to the reflected law.
\end{proof}

\begin{proposition}\label{st:prop:comparison}
We choose an odd $m$ such that $2D\eps(m+1)\le1/2$, with $D$ from
\Cref{st:lem:phase}.  We choose $R\ge m+2$ so that
\begin{equation}\label{st:eq:R-choice}
 S_R\le\tfrac14(b+u)^{-m-1}.
\end{equation}
We put $N=|B_R(o)|$, let $L$ be given by \Cref{st:lem:logs}, and set
$t=d_{\rm TV}(f,\psi_{\eps,\nu})$.  For almost every $E\in J$ and every
$x\in\supp p_m(o,\cdot)$,
\begin{equation}\label{st:eq:comparison}
 T_E^f(o,x)\le-\log p_m(o,x)-m\delta+\log(4\beta)
 +C_+(\eps,K)+2L\sqrt{Nt}.
\end{equation}
The same estimate holds at $E+i\eta$ for $0<\eta\le\eps$.
\end{proposition}
\begin{proof}
We replace the law only in $\Lambda=B_R(o)$ by $\psi_{\eps,\nu}$ and keep
the exterior $f$ variables unchanged.  We couple the $N=|\Lambda|$ changed
coordinates so that each pair differs with probability $t$; this is a
maximal coupling. If $\mathcal M$ is the event that at least one
coordinate differs, then the union bound gives
\begin{equation*}
 \PP(\mathcal M)\le Nt.
\end{equation*}
On $\mathcal M^c$ the two operators, and hence the relevant resolvent
entries, agree.  We put
\[
 Z_1=\log|G^f_{ox}(E+i\eta)|,
 \qquad Z_2=\log|G^{\rm hyb}_{ox}(E+i\eta)|.
\]
By \Cref{st:lem:logs}, $\|Z_1\|_2,\|Z_2\|_2\le L$.  Therefore
\begin{align}
 |\E Z_1-\E Z_2|
 &=|\E[(Z_1-Z_2)\one_{\mathcal M}]|\le(\|Z_1\|_2+\|Z_2\|_2)\PP(\mathcal M)^{1/2}
 \le2L\sqrt{Nt}.
 \label{st:eq:coupling}
\end{align}

In the hybrid ball we write each potential as
\[
 V_x=C_x+w_x,
\]
where $C_x$ is a full Cauchy variable of scale $\eps$ and $w_x\sim\nu$.
We condition on the exterior variables and on the entire bounded background
$w$.  Lemma~\ref{bc:lem:finite-jensen} replaces the $C_x$ in the finite
ball by $-i\eps$ and gives
\begin{equation}\label{st:eq:hybrid-jensen-lower}
 \E_{C,\Lambda}\log|G^{\rm hyb}_{ox}(E+i\eta)|
 \ge\log|R_\Lambda^w(E+i\eta)_{ox}|.
\end{equation}
To bound the reference kernel, we apply \Cref{st:lem:phase} with
$t=\eta+\eps$. Since $0<\eta\le\eps$,
\begin{equation*}
 D(\eta+\eps)(m+1)
 \le2D\eps(m+1)\le\frac12,
\end{equation*}
and hence
\begin{equation}\label{st:eq:comparison-phase-loss}
 |G_w^0(E+i(\eta+\eps))_{ox}-G_w^0(E)_{ox}|
 \le\frac12F_w(E),
 \qquad 0<\eta\le\eps.
\end{equation}
By \Cref{st:lem:shield} and \eqref{st:eq:R-choice},
\begin{equation}\label{st:eq:comparison-shield-loss}
 |R_\Lambda^w(E+i\eta)_{ox}
   -G_w^0(E+i(\eta+\eps))_{ox}|
 \le\frac14(b+u)^{-m-1}.
\end{equation}
For $p_m(o,x)>0$, $(A^m)_{ox}\ge1$ and $B_w\ge A$ entrywise.  The
$m$-th term of \eqref{st:eq:positive-series} therefore gives
\begin{equation*}
 F_w(E)
 \ge (B_w^m)_{ox}(E+u)^{-m-1}
 \ge (A^m)_{ox}(b+u)^{-m-1}
 \ge(b+u)^{-m-1}.
\end{equation*}
Combining this with
\eqref{st:eq:comparison-phase-loss}--\eqref{st:eq:comparison-shield-loss}
and $G_w^0(E)_{ox}=-F_w(E)$ yields
\begin{align}
 |R_\Lambda^w(E+i\eta)_{ox}|
 &\ge F_w(E)-\frac12F_w(E)
       -\frac14(b+u)^{-m-1}\ge\frac14F_w(E).
 \label{st:eq:reference-quarter}
\end{align}
Taking logarithms in \eqref{st:eq:reference-quarter}, inserting into
\eqref{st:eq:hybrid-jensen-lower}, and averaging over $w$ gives
\begin{align}
 \E\log|G^{\rm hyb}_{ox}(E+i\eta)|
 &\ge \E_\nu\log F_w(E)-\log4\ge \log p_m(o,x)+m\delta-\log(4\beta),
 \label{st:eq:hybrid-log-lower}
\end{align}
where the second line is \Cref{st:lem:averaged-walk}.

The coupling estimate \eqref{st:eq:coupling} transfers
\eqref{st:eq:hybrid-log-lower} back to the target law:
\begin{equation}\label{st:eq:target-log-lower}
 \E\log|G^f_{ox}(E+i\eta)|
 \ge \log p_m(o,x)+m\delta-\log(4\beta)-2L\sqrt{Nt}.
\end{equation}
We finally use the identity
\[
 \log^-s=\log^+s-\log s,
 \qquad s>0.
\]
Taking expectations, applying the positive-log bound from
\Cref{st:lem:logs}, and inserting \eqref{st:eq:target-log-lower} gives
\begin{align*}
 T^f_{E,\eta}(o,x)
 &=\E\log^-|G^f_{ox}(E+i\eta)|\le-\log p_m(o,x)-m\delta+
      \log(4\beta)+C_+(\eps,K)+2L\sqrt{Nt},
\end{align*}
which is \eqref{st:eq:comparison} for $0<\eta\le\eps$.  Because $m$ is odd and the graph is bipartite, $p_m(o,o)=0$, so every
endpoint in the block estimate is off diagonal.  For almost every $E\in J$,
\Cref{st:lem:logs} gives
\begin{equation*}
 \sup_{0<\eta\le\eps}\E|\log|G^f_{ox}(E+i\eta)||^2\le L^2,
 \qquad
 T_E^f(o,x)=\lim_{\eta\downarrow0}T_{E,\eta}^f(o,x).
\end{equation*}
Thus letting $\eta\downarrow0$ in the preceding bound proves
\eqref{st:eq:comparison} for the real boundary values.
\end{proof}

\subsubsection{Stability proof}
\begin{proof}[Proof of \Cref{thm:law-stability}]
We let $m$ be the smallest odd integer satisfying $m\ge\eps^{-1/2}$ and put
\[
 K_\eps=C_{\rm gl}(\eps,K)+C_+(\eps,K)+\log(4\beta)+1.
\]
Then $K_\eps=O(\log(1/\eps))$.  Since $m$ is the smallest odd integer
with $m\ge\eps^{-1/2}$,
\begin{equation*}
 \eps^{-1/2}\le m\le\eps^{-1/2}+2,
 \qquad
 \eps(m+1)\le\sqrt\eps+3\eps\longrightarrow0,
 \qquad
 \frac{K_\eps}{m}=O\!\left(\sqrt\eps\log\frac1\eps\right)\longrightarrow0.
\end{equation*}
Together with $H_m/m\to h_{\rm rw}$, this gives, for all sufficiently small
$\eps$,
\begin{equation*}
 2D\eps(m+1)\le\tfrac12,\qquad
 H_m/m\le h_{\rm rw}+\delta/4,\qquad
 K_\eps/m\le\delta/4.
\end{equation*}
We choose $R$ by \eqref{st:eq:R-choice}, then fix $L$ from
\Cref{st:lem:logs}, and set
\begin{equation}\label{st:eq:tstar-definition}
 N=|B_R(o)|<\infty,
 \qquad
 t_*=(16NL^2)^{-1}>0.
\end{equation}
These choices depend only on $d,J,u,v,K,\eps$ and are uniform over
all $\nu$ satisfying \eqref{st:eq:intro-noise}.

For $x\in\supp p_m(o,\cdot)$ we define
\[
 b_x=-\log p_m(o,x)-m\delta+K_\eps.
\]
The probability bound in \eqref{bc:eq:prob-bounds} gives
$-\log p_m(o,x)\ge m\log(Q/r)$.  Since
$\beta=b+v>r$, we have $\delta=\log(Q/\beta)<\log(Q/r)$, and therefore
\begin{equation*}
 b_x\ge m\log\frac{Q}{r}-m\log\frac{Q}{\beta}
 =m\log\frac\beta r\ge0.
\end{equation*}
We set $t=d_{\rm TV}(f,\psi_{\eps,\nu})<t_*$. By
\eqref{st:eq:tstar-definition},
\begin{equation*}
 2L\sqrt{Nt}<\frac12.
\end{equation*}
Thus \Cref{st:prop:comparison} and the definition of $K_\eps$ yield, for
almost every $E\in J$,
\begin{align}
 T_E^f(o,x)+C_{\rm gl}(\eps,K)
 &\le -\log p_m(o,x)-m\delta+K_\eps-\frac12=b_x-\frac12<b_x.
 \label{st:eq:block-pointwise}
\end{align}
Averaging over the block endpoints gives
\begin{align}
 \sum_xp_m(o,x)b_x
 &=H_m-m\delta+K_\eps\le m\left(h_{\rm rw}+\frac\delta4\right)-m\delta
      +m\frac\delta4
 =m\left(h_{\rm rw}-\frac\delta2\right),
 \label{st:eq:block}\\
 mh_{\rm rw}-\sum_xp_m(o,x)b_x
 &\ge\frac{m\delta}{2}>0.
 \notag
\end{align}
Equations \eqref{st:eq:block-pointwise}--\eqref{st:eq:block} are exactly
the strict finite-block inequalities required in \Cref{bc:prop:block}.
The density floor \eqref{st:eq:intro-floor} and
\Cref{cor:shifted-dos} give positivity of the averaged local density on
$J$, since $r\in\sigma(A)$ and $\ell>b-r$.

The hypotheses of \Cref{bc:prop:block} hold with
$C_{\rm gl}=C_{\rm gl}(\eps,K)$ and $C_+=C_+(\eps,K)$ by
\Cref{st:lem:logs} and \eqref{st:eq:block}.  Hence the spectrum is purely
a.c. on $J$ and every local a.c. density is positive there almost
everywhere.

For $-J$, we put $f^-(t)=f(-t)$ and $\nu^-(S)=\nu(-S)$ for Borel sets $S$.
The same hypotheses and constants apply. For $H(V)=A+\diag(V)$, the bipartite sign unitary
$(\mathcal J\psi)(x)=(-1)^{d(o,x)}\psi(x)$ satisfies
$-\mathcal J H(V)\mathcal J^*=H(-V)$. Thus the conclusion holds on $-J$;
intersect the two probability-one events.
\end{proof}

\begin{proof}[Proof of \Cref{thm:stability_under_change_density}]
Apply \Cref{thm:law-stability} with $\nu=\delta_0$, the unit
point mass at zero, $u=v=0$, and $K=2$.
\end{proof}

\begin{proof}[Proof of \Cref{thm:bounded-cauchy}]
We apply \Cref{thm:law-stability} with $\nu=\delta_0$, the unit point mass
at zero, hence $u=v=0$, and
$K=2$. Define the retained Cauchy mass by
\[
 q_{\eps,M}\coloneq\int_{-M}^M c_\eps(v)\,\dd v
 =\frac2\pi\arctan(M/\eps).
\]
For $M\ge\eps$,
\[
 \varrho_{\eps,M}\le2c_\eps,\qquad
 d_{\rm TV}(\varrho_{\eps,M},c_\eps)
 =1-q_{\eps,M}\le\frac{2\eps}{\pi M}.
\]
After choosing $\eps$ sufficiently small, we choose $M$ so large that this
total-variation distance is below the tolerance in
\Cref{thm:law-stability} and $M>b-r$.  Then
$\varrho_{\eps,M}$ has a positive lower bound on some
$[-\ell,\ell]$ with $\ell>b-r$, so all hypotheses of that theorem hold.
This proves \eqref{bc:eq:conclusion}--\eqref{bc:eq:local-density}.
\end{proof}


\begin{thebibliography}{99}

\bibitem{AizenmanMolchanov}
M.~Aizenman and S.~Molchanov,
\emph{Localization at large disorder and at extreme energies: an elementary derivation},
Comm. Math. Phys. \textbf{157} (1993), no.~2, 245--278.
\doi{10.1007/BF02099760}.

\bibitem{ASFH}
M.~Aizenman, J.~H.~Schenker, R.~M.~Friedrich, and D.~Hundertmark,
\emph{Finite-volume fractional-moment criteria for Anderson localization},
Comm. Math. Phys. \textbf{224} (2001), no.~1, 219--253.
\doi{10.1007/s002200100441}.

\bibitem{AWtree}
M.~Aizenman and S.~Warzel,
\emph{Resonant delocalization for random Schr\"odinger operators on tree graphs},
J. Eur. Math. Soc. \textbf{15} (2013), no.~4, 1167--1222.
\doi{10.4171/JEMS/389}.

\bibitem{AWBook}
M.~Aizenman and S.~Warzel,
\emph{Random Operators: Disorder Effects on Quantum Spectra and Dynamics},
Graduate Studies in Mathematics, vol.~168,
American Mathematical Society, Providence, RI, 2015.

\bibitem{AW}
M.~Aizenman and S.~Warzel,
\emph{Boosted Simon--Wolff spectral criterion and resonant delocalization},
Comm. Pure Appl. Math. \textbf{69} (2016), no.~11, 2195--2218.
\doi{10.1002/cpa.21625}; \arxiv{1412.3681v2}.

\bibitem{ASWstability}
M.~Aizenman, R.~Sims, and S.~Warzel,
\emph{Stability of the absolutely continuous spectrum of random Schr\"odinger operators on tree graphs},
Probab. Theory Related Fields \textbf{136} (2006), no.~3, 363--394.
\doi{10.1007/s00440-005-0486-8}.

\bibitem{ChenMaciejkoBoettcher}
A.~Chen, J.~Maciejko, and I.~Boettcher,
\emph{Anderson localization transition in disordered hyperbolic lattices},
Phys. Rev. Lett. \textbf{133} (2024), 066101.
\doi{10.1103/PhysRevLett.133.066101}.

\bibitem{BridsonHaefliger}
M.~R.~Bridson and A.~Haefliger,
\emph{Metric Spaces of Non-Positive Curvature},
Grundlehren der mathematischen Wissenschaften, vol.~319,
Springer-Verlag, Berlin, 1999.
\doi{10.1007/978-3-662-12494-9}.

\bibitem{CarberyWright}
A.~Carbery and J.~Wright,
\emph{Distributional and $L^q$ norm inequalities for polynomials over
convex bodies in $\mathbb R^n$},
Math. Res. Lett. \textbf{8} (2001), no.~3, 233--248.
\doi{10.4310/MRL.2001.v8.n3.a1}.

\bibitem{CarmonaLacroix}
R.~Carmona and J.~Lacroix,
\emph{Spectral Theory of Random Schr\"odinger Operators},
Probability and Its Applications,
Birkh\"auser, Boston, 1990.
\doi{10.1007/978-1-4612-4488-2}.

\bibitem{ChristiansenSimonZinchenko}
J.~S.~Christiansen, B.~Simon, and M.~Zinchenko,
\emph{Finite gap Jacobi matrices, I. The isospectral torus},
Constr. Approx. \textbf{32} (2010), 1--65.
\arxiv{0810.3273}.

\bibitem{DavisOkun}
M.~W.~Davis and B.~Okun,
\emph{Vanishing theorems and conjectures for the $\ell^2$-homology of
right-angled Coxeter groups},
Geom. Topol. \textbf{5} (2001), 7--74.
\doi{10.2140/gt.2001.5.7}; \arxiv{math/0102104}.

\bibitem{DG}
B.~Datta and S.~Gupta,
\emph{Semi-regular tilings of the hyperbolic plane},
Discrete Comput. Geom. \textbf{65} (2021), 531--553.
\doi{10.1007/s00454-019-00156-0}.

\bibitem{deLaHarpe}
P.~de~la~Harpe,
\emph{Groupes hyperboliques, alg\`ebres d'op\'erateurs et un th\'eor\`eme de Jolissaint},
C. R. Acad. Sci. Paris S\'er. I Math. \textbf{307} (1988), no.~14, 771--774.

\bibitem{DubininKarp}
V.~N.~Dubinin and D.~Karp,
\emph{Two-sided bounds for the logarithmic capacity of multiple intervals},
J. Anal. Math. \textbf{113} (2011), 227--239.
\arxiv{0905.3283}.

\bibitem{Duren}
P.~L.~Duren,
\emph{Theory of $H^p$ Spaces},
Pure and Applied Mathematics, vol.~38, Academic Press, New York, 1970.

\bibitem{FroeseHaslerSpitzer}
R.~Froese, D.~Hasler, and W.~Spitzer,
\emph{Transfer matrices, hyperbolic geometry and absolutely continuous spectrum for some discrete Schr\"odinger operators on graphs},
J. Funct. Anal. \textbf{230} (2006), no.~1, 184--221.
\doi{10.1016/j.jfa.2005.04.004}.

\bibitem{Garnett}
J.~B.~Garnett,
\emph{Bounded Analytic Functions},
revised first edition, Graduate Texts in Mathematics, vol.~236,
Springer, New York, 2007.
\doi{10.1007/0-387-49763-3}.

\bibitem{Green}
E.~R.~Green,
\emph{Graph products of groups},
Ph.D. thesis, University of Leeds, 1990.

\bibitem{HJL}
O.~H\"aggstr\"om, J.~Jonasson, and R.~Lyons,
\emph{Explicit isoperimetric constants and phase transitions in the
random-cluster model},
Ann. Probab. \textbf{30} (2002), no.~1, 443--473.
\doi{10.1214/aop/1020107775}; \arxiv{math/0008191}.

\bibitem{HislopMuller}
P.~D.~Hislop and P.~M\"uller,
\emph{A lower bound for the density of states of the lattice Anderson model},
Proc. Amer. Math. Soc. \textbf{136} (2008), no.~8, 2887--2893.
\doi{10.1090/S0002-9939-08-09361-1}.

\bibitem{Jolissaint}
P.~Jolissaint,
\emph{Rapidly decreasing functions in reduced $C^*$-algebras of groups},
Trans. Amer. Math. Soc. \textbf{317} (1990), no.~1, 167--196.
\doi{10.1090/S0002-9947-1990-0943303-2}.

\bibitem{Klein}
A.~Klein,
\emph{Extended states in the Anderson model on the Bethe lattice},
Adv. Math. \textbf{133} (1998), no.~1, 163--184.
\doi{10.1006/aima.1997.1688}.

\bibitem{KleinSadel}
A.~Klein and C.~Sadel,
\emph{Absolutely continuous spectrum for random Schr\"odinger operators on the Bethe strip},
Math. Nachr. \textbf{285} (2012), no.~1, 5--26.
\doi{10.1002/mana.201100019}.

\bibitem{KellerLenzWarzel}
M.~Keller, D.~Lenz, and S.~Warzel,
\emph{Absolutely continuous spectrum for random operators on trees of finite cone type},
J. Anal. Math. \textbf{118} (2012), 363--396.
\doi{10.1007/s11854-012-0040-4}.

\bibitem{LiPengWangHu}
T.~Li, Y.~Peng, Y.~Wang, and H.~Hu,
\emph{Anderson transition and mobility edges on hyperbolic lattices with randomly connected boundaries},
Commun. Phys. \textbf{7} (2024), 371.
\doi{10.1038/s42005-024-01848-7}.

\bibitem{LiawTreil}
C.~Liaw and S.~Treil,
\emph{Matrix measures and finite rank perturbations of self-adjoint operators},
J. Spectr. Theory \textbf{10} (2020), no.~4, 1173--1210.
\doi{10.4171/JST/324}.
Theorem and lemma numbers refer to \arxiv{1806.08856v3}.

\bibitem{Marx}
C.~A.~Marx,
\emph{Continuity of spectral averaging},
Proc. Amer. Math. Soc. \textbf{139} (2011), no.~1, 283--291.
\doi{10.1090/S0002-9939-2010-10629-9}; \arxiv{1009.3694}.

\bibitem{RaumSkalski}
S.~Raum and A.~Skalski,
\emph{Classifying right-angled Hecke $C^*$-algebras via $K$-theoretic
invariants},
Adv. Math. \textbf{407} (2022), 108559.
\href{https://arxiv.org/abs/2103.02962}{arXiv:2103.02962v4}.
See Corollary~4.9 (the undeformed reduced group algebra).

\bibitem{SadelTreeStrip}
C.~Sadel,
\emph{Absolutely continuous spectrum for random Schr\"odinger operators on tree-strips of finite cone type},
Ann. Henri Poincar\'e \textbf{14} (2013), no.~4, 737--773.
\doi{10.1007/s00023-012-0203-y}.

\bibitem{SaffSurvey}
E.~B.~Saff,
\emph{Logarithmic potential theory with applications to approximation theory},
Surv. Approx. Theory \textbf{5} (2010), 165--200.
\arxiv{1010.3760}.

\bibitem{SaffTotik}
E.~B.~Saff and V.~Totik,
\emph{Logarithmic Potentials with External Fields},
Grundlehren der mathematischen Wissenschaften, vol.~316,
Springer-Verlag, Berlin, 1997.
\doi{10.1007/978-3-662-03329-6}.

\bibitem{SewardKoopman}
B.~Seward,
\emph{The Koopman representation and positive Rokhlin entropy},
Int. Math. Res. Not. IMRN \textbf{2023} (2023), no.~1, 350--371.
\doi{10.1093/imrn/rnab268}.

\bibitem{SimonCapacity}
B.~Simon,
\emph{Equilibrium measures and capacities in spectral theory},
Inverse Probl. Imaging \textbf{1} (2007), no.~4, 713--772.
\arxiv{0711.2700}.

\bibitem{SimonRankOne}
B.~Simon,
\emph{Spectral analysis of rank one perturbations and applications},
in \emph{Mathematical Quantum Theory II: Schr\"odinger Operators},
CRM Proceedings \& Lecture Notes, vol.~8,
American Mathematical Society, Providence, RI, 1995, pp.~109--149.

\bibitem{SimonTrace}
B.~Simon,
\emph{Trace Ideals and Their Applications},
2nd ed., Mathematical Surveys and Monographs, vol.~120,
American Mathematical Society, Providence, RI, 2005.
\doi{10.1090/surv/120}.

\bibitem{SimonAveraging}
B.~Simon,
\emph{Spectral averaging and the Krein spectral shift},
Proc. Amer. Math. Soc. \textbf{126} (1998), no.~5, 1409--1413.
\doi{10.1090/S0002-9939-98-04261-0}.

\bibitem{Tautenhahn}
M.~Tautenhahn,
\emph{Localization criteria for Anderson models on locally finite graphs},
J. Stat. Phys. \textbf{144} (2011), no.~1, 60--75.
\doi{10.1007/s10955-011-0248-1}; \arxiv{1008.4503}.

\bibitem{SjostrandZworski}
J.~Sj\"ostrand and M.~Zworski,
\emph{Elementary linear algebra for advanced spectral problems},
Ann. Inst. Fourier (Grenoble) \textbf{57} (2007), no.~7, 2095--2141.
\doi{10.5802/aif.2328}.

\bibitem{Teschl}
G.~Teschl,
\emph{Mathematical Methods in Quantum Mechanics: With Applications to Schr\"odinger Operators},
2nd ed., Graduate Studies in Mathematics, vol.~157,
American Mathematical Society, Providence, RI, 2014.

\bibitem{Wegner}
F.~Wegner,
\emph{Bounds on the density of states in disordered systems},
Z. Phys. B \textbf{44} (1981), 9--15.
\doi{10.1007/BF01292646}.

\bibitem{Woess}
W.~Woess,
\emph{Random Walks on Infinite Graphs and Groups},
Cambridge Tracts in Mathematics, vol.~138,
Cambridge University Press, Cambridge, 2000.
\doi{10.1017/CBO9780511470967}.

\end{thebibliography}
\end{document}